\relax
\RequirePackage{etex}
\documentclass[a4paper, 11pt]{article} 
\usepackage[margin=1in]{geometry}
\usepackage[hyphens]{url}  
\usepackage{graphicx} 
\usepackage{graphicx}  
\usepackage{natbib}  
\usepackage{caption} 
\usepackage{amssymb,amsmath,amsthm,microtype,booktabs,csquotes,enumitem,tikz,subcaption,mathtools,pgfplots,array}
\usepackage{algorithmicx,algorithm}
\usepackage[noend]{algpseudocode}
\usepackage[export]{adjustbox}
\pgfplotsset{compat=newest}
\usetikzlibrary{intersections,pgfplots.fillbetween}
\allowdisplaybreaks

\renewcommand{\(}{\left(}
\renewcommand{\)}{\right)}
\newcommand{\lt}{\left[}
\newcommand{\rt}{\right]}

\DeclareMathOperator*{\argmax}{arg\,max}
\newcommand{\ip}[2]{\langle #1,#2 \rangle}
\newcommand{\N}{\mathcal N}
\newcommand{\x}{x_{ki}}
\newcommand{\xm}{x_{\text{mix}}}
\newcommand{\ym}{y_{\text{mix}}}
\definecolor{scolor}{rgb}{0.733,0.843,0.890}
\definecolor{ocolor}{rgb}{0.783,0.893,0.940}
\definecolor{oocolor}{rgb}{0.783,0.893,0.940}

\newtheorem{theorem}{Theorem}[section]
\newtheorem{lemma}[theorem]{Lemma}

\theoremstyle{definition}

\theoremstyle{remark}
\newtheorem*{remark}{Remark}

\title{Exploration-Exploitation in Multi-Agent Learning:\\
Catastrophe Theory Meets Game Theory\thanks{Stefanos Leonardos gratefully acknowledgeS MOE AcRF Tier 2 Grant 2016-T2-1-170 and NRF 2018 Fellowship NRF-NRFF2018-07. Georgios Piliouras gratefully acknowledges MOE AcRF Tier 2 Grant 2016-T2-1-170, grant PIE-SGP-AI-2018-01, NRF2019-NRF-ANR095 ALIAS grant and NRF 2018 Fellowship NRF-NRFF2018-07.}}

\author{Stefanos Leonardos,\textsuperscript{\rm 1} Georgios Piliouras\textsuperscript{\rm 1}\\}
\date{\textsuperscript{\rm 1}Singapore University of Technology and Design\\ 
8 Somapah Road, 487372 Singapore\\
\{stefanos\_leonardos~;~georgios\}@sutd.edu.sg}

\begin{document}

\maketitle

\begin{abstract}
Exploration-exploitation is a powerful and practical tool in multi-agent learning (MAL), however, its effects are far from understood. To make progress in this direction, we study a smooth analogue of Q-learning. We start by showing that our learning model has strong theoretical justification as an optimal model for studying exploration-exploitation. Specifically, we prove that smooth Q-learning has bounded regret in arbitrary games for a cost model that explicitly captures the balance between game and exploration costs and that it always converges to the set of quantal-response equilibria (QRE), the standard solution concept for games under bounded rationality, in weighted potential games with heterogeneous learning agents. In our main task, we then turn to measure the effect of exploration in collective system performance. We characterize the geometry of the QRE surface in low-dimensional MAL systems and link our findings with catastrophe (bifurcation) theory. In particular, as the exploration hyperparameter evolves over-time, the system undergoes phase transitions where the number and stability of equilibria can change radically given an infinitesimal change to the exploration parameter. Based on this, we provide a formal theoretical treatment of how tuning the exploration parameter can provably lead to equilibrium selection with both positive as well as negative (and potentially unbounded) effects to system performance. 
\end{abstract}

\section{Introduction}
\label{sec:intro}

The problem of optimally balancing exploration and exploitation in \emph{multi-agent systems (MAS)} has been a fundamental motivating driver of online learning, optimization theory and evolutionary game theory \cite{Cla98,Pan05}. From a behavioral perspective, it involves the design of realistic models to capture complex human behavior, such as the standard Experienced Weighted Attraction model~\cite{EWA1,EWA3}. Learning agents use time varying parameters to explore suboptimal, boundedly rational decisions, while at the same time, they try to coordinate with other interacting agent and maximize their profits~\cite{EWA2, Bow02,Kai09}.\par
From an AI perspective, the exploration-exploitation dilemma is related to the optimization of adaptive systems. For example, neural networks are trained to parameterize policies ranging from very exploratory to purely exploitative, whereas meta-controllers decide which policy to prioritize~\cite{2020arXiv200313350P}. Similar techniques have been applied to rank agents in tournaments according to performance for preferential evolvability~\cite{Lan17,omidshafiei2019alpha,Row19} and to design \emph{multi-agent learning (MAL)} algorithms that prevent collective learning from getting trapped in local optima~\cite{Kai10,Kai11}.\par
Despite these notable advances both on the behavioral modelling and AI fronts, the theoretical foundations of learning in MAS are still largely incomplete even in simple settings~\cite{Wun10,Blo15}. While there is still no theory to formally explain the performance of MAL algorithms, and in particular, \emph{the effects of exploration in MAS}~\cite{Klo10}, existing research suggests that many pathologies of exploration already emerge at stateless matrix games at which naturally emerging collective learning dynamics exhibit a diverse set of outcomes \cite{Sat03,Sat05,Tuy12}. \par
The reasons for the lack of a formal theory are manifold. First, even without exploration, MAL in games can result in complex behavior that is hard to analyze~\cite{Bal20,Mer18,2018arXiv180405464M}. Once explicit exploration is enforced, the behavior of online learning becomes even more intractable as Nash Equilibria (NE) are no longer fixed points of agents' behavior. Finally, if parameters are changed enough, then we get bifurcations and possibly chaos~\cite{Wol12,Pal17,San18}. \\[0.2cm]
{\bf Our approach \& results.} Motivated by the above, we study a smooth variant of stateless Q-learning, with softmax or Boltzmann exploration (one of the most fundamental models of exploration-exploitation in MAS), termed Boltzmann Q-learning or \emph{smooth Q-learning (SQL)}, which has recently received a lot of attention due to its connection with evolutionary game theory \cite{Tuy03,Kia12}. Informally (see Section~\ref{sec:model} for the rigorous definition), each agent $k$ updates her choice distribution $x=\(x_i\)$ according to the rule
\[\textstyle \dot x_i/x_i=\beta_k \(u_i-\bar u\)-\alpha_k\(\ln{x_i}-\sum_j x_j \ln{x_j}\)\]
where $u_i,\bar u$ denote agent $k$'s utility from action $i$ and average utility, respectively, given all other agents' actions and $\alpha_k/\beta_k$ is agent $k$'s exploration rate.\footnote{This variant of Q-learning has been also extensively studied in the economics and reinforcement learning literature under various names, see e.g.,  \cite{Alo10,San18} and \cite{Kae96,Mer16}, respectively.} Agents tune the exploration parameter to increase/decrease exploration during the learning process. We analyze the performance of SQL dynamics along the following axes. \par
\emph{Regret and Equilibration.} First, we benchmark their performance against the optimal choice distribution in a cost model that internalizes agents' utilities from exploring the space (Lemma~\ref{lem:modified}), and show that in this context, the SQL dynamics enjoy a \textit{constant} total regret bound in arbitrary games that depends logarithmically in the number of actions (Theorem~\ref{thm:regret}). Second, we show that they converge to \emph{Quantal Response Equilibria} (QRE)\footnote{The prototypical extension of NE for games with bounded rationality \cite{Mck95}.} in weighted potential games with heterogeneous agents of arbitrary size (Theorem~\ref{thm:potential}).\footnote{Apart from their standard applications, see \cite{Pan16,Swe18,Per20} and references therein, weighted potential games naturally emerge in distributed settings such as recommendation systems \cite{Ben18}.} The underpinning intuition is that agents' deviations from pure exploitation are not a result of their bounded rationality but rather a perfectly rational action in the quest for more information about unexplored choices which creates value on its own. This is explicitly captured by a correspondingly modified Lyapunov function (potential) which combines the original potential with the entropy of each agent's choice distribution (Lemma~\ref{lem:potential}). \par
While previously not formally known, these properties mirror results  of strong regret guarantees for online algorithms (see e.g., \citet{cesa2006prediction,Kwo17,Mer18};   convergence results for SQL in more restricted settings (\citet{Les05,Cou15}).\footnote{They are also of independent interest in the limited literature on the properties of the softmax function~\cite{Gao17}.} However, whereas in previous works such results corresponded to main theorems, in our case they are only our starting point as they clearly not suffice to explain the 
  disparity between the regularity of the SQL dynamics in theory and their unpredictable performance in practice. \par
We are faced with two major unresolved complications. First, the outcome of the SQL algorithm in MAS is highly sensitive on the exploration parameters~\cite{Tuy06}. The set of QRE ranges from the NE of the underlying game when there is no exploration to uniform randomization when exploration is constantly high (agents never learn). Second, the collective system evolution is \emph{path dependent}, i.e., in the case of time-evolving parameters, the equilibrium selection process cannot be understood by only examining the final exploration parameter but rather depends on its whole history of play \cite{Goc02,Rom15,Yan17}.\par

\emph{Catastrophe theory and equilibrium selection.} We explain the fundamentally different outcomes of exploration-exploitation with SQL in games via catastrophe theory. The link between these two distinct fields lies on the properties of the underlying game which, in turn, shape the geometry of the QRE surface. As agents' exploration parameters change, the QRE surface also changes. This prompts dramatic phase transitions in the exploration path that ultimately dictate the outcome of the learning process. Catastrophe theory reveals that such transitions tend to occur as part of well-defined qualitative geometric structures. \par
In particular, the SQL dynamics may induce a specific type of catastrophes, known as \emph{saddle-node} bifurcations \cite{Str00}. At such bifurcation points, small changes in the exploration parameters change the stability of QRE and cause QRE to merge and/or disappear. However, as we prove, this is not always sufficient to stabilize desired states; the decisive feature is whether the QRE surface is connected or not (see Theorem~\ref{thm:location_qre} and Figure~\ref{fig:location}) which in turn, determines the possible types of bifurcations, i.e., whether there are one or two branches of saddle-node bifurcation curves, that may occur as exploration parameters change (Figure~\ref{fig:intro}).

\begin{figure}[!htb]
\includegraphics[width=0.48\textwidth]{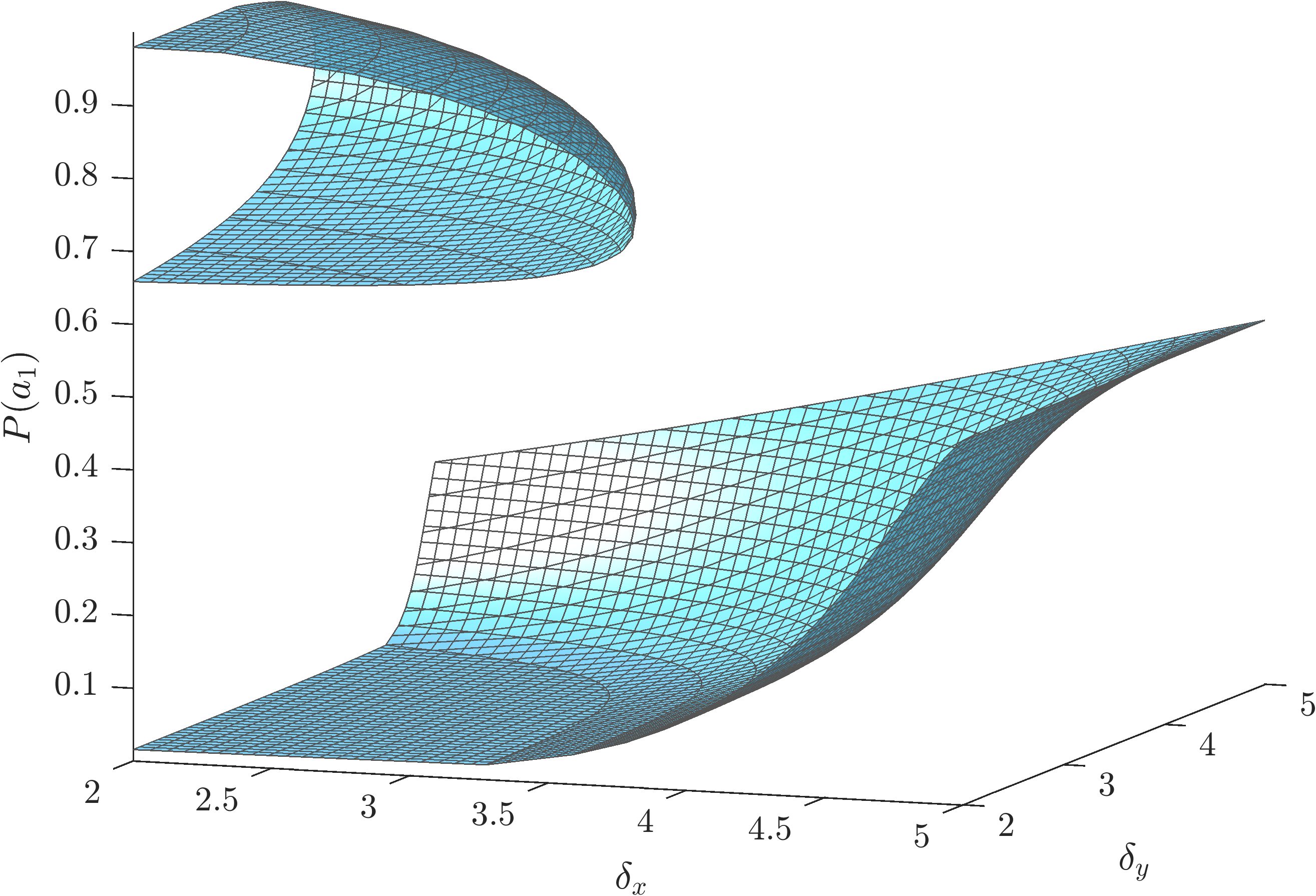}\hspace{10pt}
\includegraphics[width=0.48\textwidth]{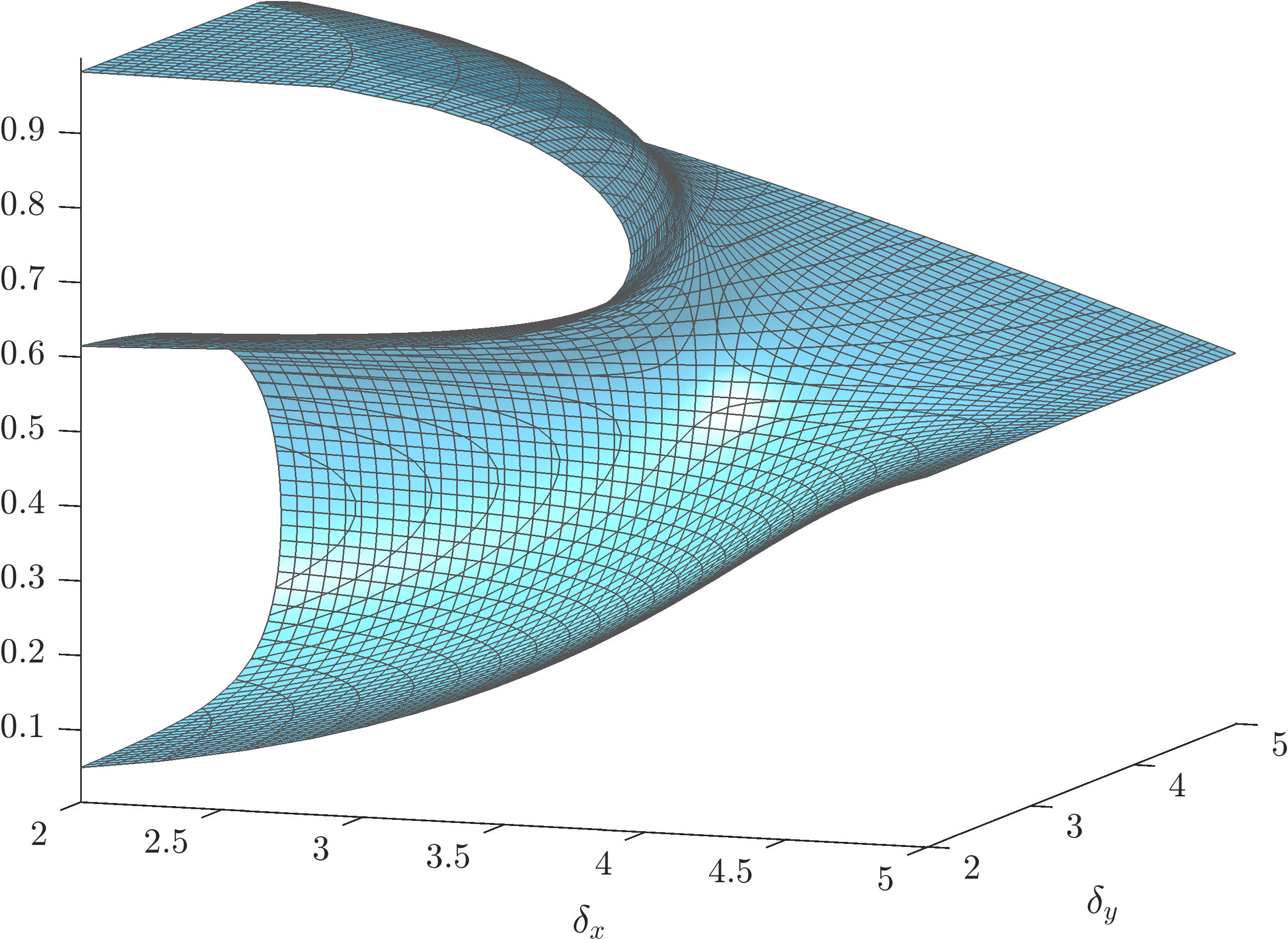}
\caption{Single \emph{saddle-node bifurcation} curve (left) vs two branches of saddle-node bifurcation curves meeting which is consistent with the emergence of a \emph{co-dimension $2$ cusp bifurcation} (right) on the QRE manifold of two player, two action games as function of the exploration rates, $\delta_x,\delta_y$ (see also Figures~\ref{fig:payrisk},\ref{fig:nopay}). The possible learning paths before, during and after exploration depend on the geometry of the QRE surface (Theorems~\ref{thm:catastrophe},\ref{thm:location_qre}).}
\label{fig:intro}
\end{figure}
In terms of performance, this is formalized in Theorem~\ref{thm:catastrophe} which states that even in the simplest of MAS, \textit{exploration can lead under different circumstances both to unbounded gain as well as unbounded loss}. While existential in nature, Theorem~\ref{thm:catastrophe} does not merely say that anything goes when exploration is performed. When coupled with the characterization of the geometric locus of QRE in Theorem~\ref{thm:location_qre}, it suggests that we can identify cases where exploration can be provably beneficial or damaging. This provides a formal geometric argument why exploration is both extremely powerful but also intrinsically unpredictable. \par
The above findings are visualized in systematic experiments in both low and large dimensional games along two representative exploration-exploitation policies, \emph{explore-then-exploit and cyclical learning rates} (Section~\ref{sec:applications}). We also visualize the modified potential (and how it changes during exploration) in weighted potential games of arbitrary size by properly adapting the technique of \citet{Li18} for visualizing high dimensional loss functions in deep learning  (Figure~\ref{fig:large}). Omitted materials, all proofs, and more experiments are included in Appendices~\ref{app:derivation}-\ref{app:experiments}.


%
\section{Preliminaries: Game Theory and SQL}
\label{sec:model}

We consider a finite set $\N$ of interacting agents indexed by $k=1,2,\dots,N$. Each agent $k\in \N$ can take an action from a finite set $A_k=\{1,2,\dots,n_k\}$. Accordingly, let $A:=\prod_{k=1}^NA_k$ denote the set of joint actions or pure action profiles, with generic element $a=\(a_1,a_2,\dots,a_N\)$. To track the evolution of the agents' choices, let $X_k=\{x_k\in\mathbb R^{n_k}:\sum_{i=1}^{n_k}\x=1, \x\ge0\}$ denote the set of all possible choice distributions $x_k:=\(\x\)_{i\in A_k}$ of agent $k\in \N$.\footnote{Depending on the context, we will use either the indices $i,j\in A_k$ or the symbol $a_k\in A_k$ to denote the pure actions of agent $k$.} We consider the dynamics in the collective state space $X:=\prod_{k=1}^N X_k$, i.e., the space of all joint choice distributions $x=\(x_k\)_{k\in\N}$. Using conventional notation, we will write $\(a_k;a_{-k}\)$ to denote the pure action profile at which agent $k\in \N$ chooses action $a_k\in A_k$ and all other agents in $\N$ choose actions $a_{-k}\in A_{-k}:=\prod_{l\neq k}A_l$. Similarly, for choice distribution profiles, we will write $\(x_k,x_{-k}\)$ with $x_{-k}\in X_{-k}:=\prod_{l\neq k}X_l$. When time is relevant, we will use the index $t$ for agent $k$'s choice distribution $x_k\(t\):=\(\x\(t\)\)_{i\in A_k}$ at time $t\ge0$.\par
When selecting an action $i\in A_k$, agent $k\in \N$ receives a reward $u_k\(i;a_{-k}\)$ which depends on the choices $a_{-k}\in A_{-k}$ of all other agents. Accordingly, the expected reward of agent $k\in \N$ for a choice distribution profile $x=\(x_k,x_{-k}\)\in X$ is equal to $u_k\(x\)=\sum_{a\in A}\(\x u_k\(i;a_{-k}\)\prod_{l\neq k}x_{la_l}\)$. We will also write $r_{ki}\(x\):=u_k\(i;x_{-k}\)$ or equivalently $r_{ki}\(x_{-k}\)$ for the reward of pure action $i\in A_k$ at the joint choice distribution profile $x=\(x_k;x_{-k}\)\in X$ and $r_k\(x\):=\(r_{ki}\(x\)\)_{i\in A_k}$ for the resulting reward vector of all pure actions of agent $k$. Using this notation, we have that $u_k\(x\)=\ip{x_k}{r_{k}\(x\)}$, where $\ip{\cdot}{\cdot}$ denotes the usual inner product in $\mathbb R^{n_k}$, i.e., $\ip{x_k}{r_k\(x\)}=\sum_{j\in A_k}x_{kj}r_{kj}\(x\)$. In particular, $\partial u_k\(x\)/\partial \x=r_{ki}\(x\)$. To sum up, the above setting can be represented in compact form with the notation $\Gamma=\(\N,\(A_k,u_k\)_{k\in \N}\)$.\par
We assume that the updates in the choice distribution $x_k$ of agent $k\in N$ are governed by the dynamics
\begin{subequations}
\label{eq:main}
\begin{align}\label{eq:qre}
\dot x_{ki}/\x &=\beta_k[r_{ki}\(x\)-\sum_{j\in A_k}x_{kj}r_{kj}\(x\)]-\alpha_k[\ln{\x}-\sum_{j\in A_k}x_{kj}\ln{x_{kj}}]\\[0.2cm]
\label{eq:qre_ip}
& =\beta_k[r_{ki}\(x\)-\ip{x_k}{r_k\(x\)}]-\alpha_k[\ln{\x}-\ip{x_k}{\ln{x_k}}]
\end{align}
\end{subequations}
where $\beta_k\in\lt 0,+\infty\)$ and $\alpha_k\in \lt 0,1\)$ are positive constants that control the rate of choice adaptation and memory loss, respectively of the learning agent $k\in \N$ and $\ln{x_k}:=\(\ln{\x}\)_{i\in A_k}$ for $x_k\in X_k$. The first term, $r_{ki}\(x\)-\sum_{j\in A_k}^nx_{kj}r_{kj}\(x\)$, corresponds to the vector field of the replicator dynamics and captures the adaptation of the agents' choices towards the best performing strategy (exploitation). The second term, $\ln{x_{ki}} -\sum_{j\in A_k}x_{kj}\ln{x_{kj}}$, corresponds to the memory of the agent and the exploration of alternative choices. Due to their mathematical connection with Q-learning, we will refer to the dynamics in \eqref{eq:main} as \emph{smooth Q-learning} (SQL) dynamics.\footnote{An explicit derivation due to \cite{Tuy03,Sat05,Wol12,Kia12} (among others) of the connection between Q-learning \cite{Wat92} and the above dynamics (including their resting points) is given in Appendix~\ref{app:derivation}.} The interior fixed points $x^Q\in X$ of the dynamics in equations \eqref{eq:main} are the \emph{Quantal Response Equilibria} (QRE) of $\Gamma$. In particular, each $x^Q_k\in X_k$ for $k=1,2,\dots,N$ satisfies
\begin{equation}\label{eq:fixed}
 x^Q_{ki}=\exp{(r_{ki}(x^Q_{-k})/\delta_k)}/\sum_{i\in A_k}\exp{(r_{kj}(x^Q_{-k})/\delta_k)},
\end{equation}
for $i\in A_k$, where $\delta_k:=\alpha_k/\beta_k$ denotes the \emph{exploration rate} for each agent $k\in \N$.\par
\section{Bounded Regret in All Games and Convergence in Weighted Potential Games}
\label{sec:results}
Our first observation is that the SQL dynamics in \eqref{eq:main} can be considered as replicator dynamics in a modified game with the same sets of agents and possible actions for each agent but with modified utilities. 

\begin{lemma}\label{lem:modified}
Given $\Gamma=\(\N,\(A_k,u_k\)_{k\in \N}\)$, consider the modified utilities $\(u^H_k\)_{k\in \N}$ defined by $u^H_k\(x\):=\beta_k\ip{x_k}{r_k\(x\)}-\alpha_k\ip{x_k}{\ln{x_k}}$, for $x\in X$. Then, the dynamics described by the differential equation $\dot x_{ik}/x_{ik}$ in \eqref{eq:main} can be written as 
\begin{equation}\label{eq:replicator}
\dot x_{ki}/\x=r^H_{ki}\(x\)-\ip{x_k}{r^H_k\(x\)}
\end{equation}
where $r^H_{ki}\(x\):=\frac{\partial}{\partial x_{ki}}u_k^H\(x\)=\beta_kr_{ki}\(x\)-\alpha_k\(\ln{\x}+1\)$. In particular, the dynamics in \eqref{eq:main} describe the replicator dynamics in the modified setting $\Gamma^H=\(\N,\(A_k,u^H_k\)_{k\in \N}\)$.
\end{lemma}
The superscript $H$ refers to the regularizing term, $H\(x_k\):=-\ip{x}{\ln{x_k}}=-\sum_{j\in A_k}x_{kj}\ln{x_{kj}}$ which denotes the \emph{Shannon entropy} of choice distribution $x_k\in X_k$.\\[0.2cm]
\textbf{Bounded regret.} To measure the performance of the SQL dynamics in \eqref{eq:main}, we will use the standard notion of (accumulated) \emph{regret} \cite{Mer18}. The regret $R_k\(T\)$ at time $T>0$ for agent $k$ is
\begin{align}\label{eq:regret}
R_k\(T\):=\max_{x'_k\in X_k}\int_{0}^T&\lt u_k\(x'_k;x_{-k}\(t\)\)-u_k\(x_k\(t\),x_{-k}\(t\)\)\rt dt,
\end{align}
i.e., $R_k\(T\)$ is the difference in agent $k$'s rewards between the sequence of play $x_k\(t\)$ generated by the SQL dynamics and the best possible choice up to time $T$ in hindsight. Agent $k$ has \emph{bounded regret} if for every initial condition $x_k\(t\)$ the generated sequence $x_k\(t\)$ satisfies $\lim\sup R_k\(T\)\le 0$ as $T\to\infty$. Our main result in this respect is a constant upper bound on the regret of the SQL dynamics.
\begin{theorem}\label{thm:regret}
Consider the modified setting $\Gamma^H=\(\N,\(A_k,u^H_k\)_{k\in \N}\)$. Then, every agent $k\in \N$ who updates their choice distribution $x_k\in X_k$ according to the dynamics in equation \eqref{eq:replicator} has bounded regret, i.e., there exists a constant $C>0$ such that $\limsup_{T\to\infty}R^H_k\(T\)\le C$.
\end{theorem}

From the proof of Theorem~\ref{thm:regret}, it follows that the constant $C$ is logarithmic in the number of actions given a uniformly random initial condition as is the standard. This yields an optimal bound which is powerful in general MAL settings. In particular, regret minimization by the SQL dynamics at an optimal $O(1/T)$ rate implies that their time-average converges fast to \emph{coarse correlated equilibria (CCE)}. These are CCE of the perturbed game, $\Gamma^H$, but if exploration parameter is low, they are approximate CCE of the original game as well. Even $\epsilon$-CCE are $(\frac{\lambda (1+\epsilon)}{1-\mu(1+\epsilon)})$-optimal for $\lambda-\mu$ smooth games, see e.g., \cite{Rou15}. However, for games that are not smooth (e.g., games with NE that have widely different performance and hence, a large Price of Anarchy), we need more specialized tools (see Section~\ref{sec:performance}). \\[0.2cm]
\textbf{Convergence to QRE in weighted potential games with heterogeneous agents.} If $\Gamma=\(\N,\(A_k,u_k\)_{k\in \N}\)$ describes a potential game, then more can be said about the limiting behavior of the SQL dynamics. Formally, $\Gamma$ is called a \emph{weighted potential game} if there exists a function $\phi:A\to \mathbb R$ and a vector of positive weights $w=\(w_k\)_{k\in\N}$ such that for each player $k\in \N$, $u_k\(i,a_{-k}\)-u_k\(j,a_{-k}\)=w_k\(\phi\(i,a_{-k}\)-\phi\(j,a_{-k}\)\)$, for all $i\neq j\in A_k$, and $a_{-k}\in A_{-k}$. If $w_k=1$ for all $k\in \N$, then $\Gamma$ is called an \emph{exact potential game}. Let $\Phi:X\to\mathbb R$ denote the multilinear extension of $\phi$ defined by $\Phi\(x\)=\sum_{a\in A}\phi\(a\)\prod_{k\in N}x_{ka_k}$, for $x\in X.$ We will refer to $\Phi$ as the \emph{potential function} of $\Gamma$. Using this notation, we have the following.
\begin{theorem}\label{thm:potential}
If $\Gamma=\(\N,\(A_k,u_k\)_{k\in\N}\)$ admits a potential function, $\Phi:X\to\mathbb R$, then the sequence of play generated by the SQL dynamics in \eqref{eq:main} converges to a compact connected set of QRE of $\Gamma$.
\end{theorem}
Intuitively, the first term, $\beta_k\(r_{ki}\(x\)-\ip{x_k}{r_k\(x\)}\)$, in equation \eqref{eq:main} corresponds to agent $k$'s replicator dynamics in the underlying game (with utilities rescaled by $\beta_k$ that can also absorb agent $k$'s weight) and thus, it is governed by the potential function. The second term, $-\alpha_k\(\ln{\x}-\ip{x_k}{\ln{x_k}}\)$, is an idiosyncratic term which is independent from the environment, i.e., the other agents' choice distributions. Hence, the potential game structure is preserved --- up to a multiplicative constant for each player which represents that players' exploration rate $\delta_k$ --- and Theorem~\ref{thm:potential} can be established by extending the techniques of \cite{Kle09,Cou15} to the case of weighted potential games. This is the statement of Lemma~\ref{lem:potential} (which is also useful for the numerical experiments). 
\begin{lemma}\label{lem:potential}
Let $\Phi:X\to \mathbb R$ denote a potential function for $\Gamma=\(\N,\(A_k,u_k\)_{k\in\N}\)$, and consider the modified utilities $u^H_k\(x\):=\beta_k\ip{x_k}{r_k\(x\)}-\alpha_k\ip{x_k}{\ln{x_k}}$, for $x\in X$. Then, the function $\Phi^H\(x\)$ defined by
\begin{equation}\label{eq:potential}
\Phi^H\(x\):=\Phi\(x\)+\sum_{k\in\N}\delta_k H\(x_k\), \quad \text{for }x\in X,
\end{equation}
is a potential function for the modified game $\Gamma^H=\(\N,\(A_k,u^H_k\)_{k\in\N}\)$. The time derivative $\dot\Phi^H\(x\)$ of the potential function is positive along any sequence of choice distributions generated by the dynamics of equation \eqref{eq:replicator}
except for fixed points at which it is $0$. 
\end{lemma}

\section{From Topology to Performance}\label{sec:performance}
While the above establish some desirable topological properties of the SQL dynamics, the effects of exploration are still unclear in practice both in terms of equilibrium selection and agents' individual performance (utility). As we formalize in Theorem~\ref{thm:catastrophe} and visualize in Section~\ref{sec:applications}, exploration -- exploitation may lead to (unbounded) improvement, but also to (unbounded) performance loss even in simple settings. \par
To compare agents' utility for different exploration-exploitation policies, it will be convenient to denote the sequence of utilities of agent $k\in \N$ by $u_k^{\text{exploit}}\(t\), t\ge0$ if there exist thresholds $\delta_k>0$ (that may depend on the initial condition $x_k\(0\)$ of agent $k$) such that $\delta_k\(t\)<\delta_k$ for all $k\in \N$, i.e., if exploration remains \emph{low} for all agents, and by $u_k^{\text{explore}}\(t\), t\ge0$ otherwise. Then we have the following. 
\begin{theorem}[\textbf{Catastrophes in Exploration-Exploitation}]
\label{thm:catastrophe}
For any number $M>0$, there exist potential games $\Gamma^M_u=\{\N,\(X_k,u_k\)_{k\in\N}\}$ and $\Gamma^M_v=\{\N,\(X_k,v_k\)_{k\in\N}\}$, positive-measure sets of initial conditions $I_u,I_v\subset X$, and exploration rates $\delta_k>0$, so that 
\begin{align*}&\lim_{t\to\infty}\(u_k^{\text{exploit}}\(t\)/u_k^{\text{explore}}\(t\)\)\ge M, \text{ and }\\
&\lim_{t\to\infty}\(v_k^{\text{exploit}}\(t\)/v_k^{\text{explore}}\(t\)\)\le 1/M
\end{align*}
for all $k\in \N$, whenever $\limsup_{t\to\infty}\delta_k\(t\)=0$ for all $k\in \N$, i.e., whenever, after some point, exploration stops for all agents. In particular, for all agents $k\in \N$, the individual --- and hence, also the aggregate --- performance loss (gain) in terms of utility due to exploration can be unbounded, even if exploration is only performed by a single agent. 
\end{theorem}
The proof of Theorem~\ref{thm:catastrophe} is constructive and relies on Theorem~\ref{thm:location_qre} discussed next. Theorem~\ref{thm:location_qre} characterizes the geometry of the QRE surface (connected or disconnected) which determines the bifurcation type that takes place during exploration. In turn, this dictates the possible outcomes --- and hence, the individual and collective performance --- after the exploration process, as formalized by Theorem~\ref{thm:catastrophe}. \\[0.3cm] 
\textbf{Classification of $2\times2$ coordination games and geometry of the QRE surface.} 
First, we introduce some minimal additional notation and terminology regarding \emph{coordination games}.\footnote{A more general description is in Appendix~\ref{app:mechanism}.} Two player $\N=\{1,2\}$, two action $A_k=\(a_1,a_2\), k=1,2$, \emph{coordination} games are games in which the payoffs satisfy $u_{11}>u_{21}, u_{22}>u_{12}$ and $v_{11}>v_{21}, v_{22}>v_{12}$ where $u_{ij}$ ($v_{ij}$) denotes the payoff of agent $1$ (2) when that agent selects action $i$ and the other agent action $j$. Such games admit three NE, two pure on the diagonal and one fully mixed $\(\xm,\ym\)$, with $\xm,\ym\in\(0,1\)$ (see Appendix~\ref{app:mechanism} for details). The equilibrium $\(a_2,a_2\)$ is called \emph{risk-dominant} if 
\begin{equation}\label{eq:riskdominant}
\(u_{22}-u_{12}\)\(v_{22}-v_{12}\)>\(u_{11}-u_{21}\)\(v_{11}-v_{21}\).
\end{equation} 
In particular, a NE is risk dominant if it has the largest basin of attraction (is less risky) \cite{Har88}. For symmetric games, inequality \eqref{eq:riskdominant} has an intuitive interpretation: the choice at the risk dominant NE is the one that yields the highest expected payoff under complete ignorance, modelled by assigning $\(1/2,1/2\)$ probabilities to the other agent's choices. If $u_{22}\ge u_{11}$ and $v_{22}\ge v_{11}$ with at least one inequality strict, then $\(a_2,a_2\)$ is called \emph{payoff-dominant}. 
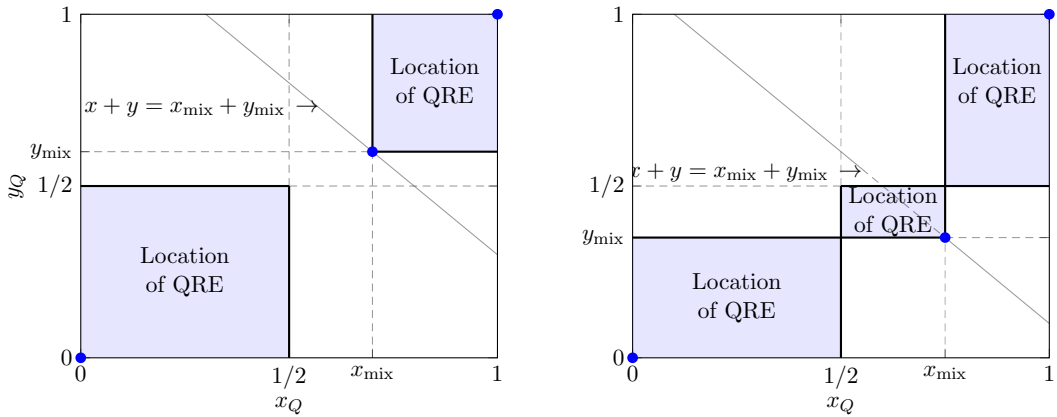
\begin{figure}[!hb]
\centering
\begin{tikzpicture}[scale=0.8, every node/.style={scale=0.8},shift=({10,6})]
\begin{axis}[ymin=0,ymax=1,xmin=0,xmax=1, xlabel=$x_Q$, ylabel=$y_Q$, axis lines =box,
xtick={0, 0.5, 0.7, 1}, xticklabels={0, 1/2, $\xm$, 1}, ytick={0, 0.5, 0.6,1}, yticklabels={0,1/2,$\ym$,1}]
\plot[name path=f11, densely dashed, opacity=0.4,domain=0:0.7] {0.6};
\plot[name path=f12, thick, domain=0.7:1] {0.6};
\plot[name path=f21, thick, domain=0:0.5] {0.5};
\plot[name path=f22, densely dashed, opacity=0.4,domain=0.5:1] {0.5};
\addplot[name path=f31, thick] coordinates {(0.5, 0) (0.5, 0.5)};
\addplot[name path=f32, densely dashed,opacity=0.4] coordinates {(0.5, 0.5) (0.5, 1)};
\addplot[name path=f41, densely dashed, opacity=0.4] coordinates {(0.7, 0) (0.7, 0.6)};
\addplot[name path=f42, thick] coordinates {(0.7, 0.6) (0.7, 1)};
\addplot[name path=f5, opacity=0.4]{1.3-x};
\node at (axis cs: 0.29,0.73) {$x+y=\xm+\ym \,\to$};
\addplot[thick,draw=blue, mark=*,mark options={color=blue}]  coordinates {(1,1)};
\addplot[thick,draw=blue, mark=*,mark options={color=blue}]  coordinates {(0,0)};
\addplot[thick,draw=blue, mark=*,mark options={color=blue}]  coordinates {(0.7,0.6)};
\plot[name path=fxb,opacity=0,domain=0:0.5] {0};
\plot[name path=fxt,opacity=0,domain=0:1] {1};
\addplot fill between[of = f12 and fxt, soft clip={domain=0.7:1}, every even segment/.style  = {blue,opacity=.1}];
\addplot fill between[of = f21 and fxb, soft clip={domain=0:0.5}, every even segment/.style  = {blue,opacity=.1}];
\node[align=center] at (axis cs: 0.25,0.25) {Location\\ of QRE};
\node[align=center] at (axis cs: 0.85,0.8) {Location\\ of QRE};
\end{axis}
\end{tikzpicture}%
\hspace{20pt}
\begin{tikzpicture}[scale=0.8, every node/.style={scale=0.8}]
\begin{axis}[ymin=0,ymax=1,xmin=0,xmax=1, xlabel=$x_Q$, axis lines =box,
xtick={0, 0.5, 0.75, 1}, xticklabels={0, 1/2, $\xm$, 1}, ytick={0, 0.5, 0.35,1}, yticklabels={0,1/2,$\ym$,1}]
\plot[name path=f11, densely dashed, opacity=0.4,domain=0.75:1] {0.35};
\plot[name path=f12, thick, domain=0:0.75] {0.35};
\plot[name path=f21, thick, domain=0.5:1] {0.5};
\plot[name path=f22, densely dashed, opacity=0.4,domain=0:0.5] {0.5};
\addplot[name path=f31, thick] coordinates {(0.5, 0) (0.5, 0.5)};
\addplot[name path=f32, densely dashed,opacity=0.3] coordinates {(0.5, 0.5) (0.5, 1)};
\addplot[name path=f41, densely dashed, opacity=0.3] coordinates {(0.75, 0) (0.75, 0.35)};
\addplot[name path=f42, thick] coordinates {(0.75, 0.35) (0.75, 1)};
\addplot[name path=f51, opacity=0.4,domain=0:0.55]{1.1-x};
\addplot[name path=f52, opacity=0.4,domain=0.75:1]{1.1-x};
\addplot[name path=f53, densely dashed,opacity=0.4,domain=0.55:0.75]{1.1-x};
\node at (axis cs: 0.27,0.54) {$x+y=\xm+\ym \,\to$};
\addplot[thick,draw=blue, mark=*,mark options={color=blue}]  coordinates {(1,1)};
\addplot[thick,draw=blue, mark=*,mark options={color=blue}]  coordinates {(0,0)};
\addplot[thick,draw=blue, mark=*,mark options={color=blue}]  coordinates {(0.75,0.35)};
\plot[name path=fxb,opacity=0,domain=0:0.5] {0};
\plot[name path=fxt,opacity=0,domain=0:1] {1};
\addplot fill between[of = f21 and fxt, soft clip={domain=0.75:1}, every even segment/.style  = {blue,opacity=.1}];
\addplot fill between[of = f12 and f21, soft clip={domain=0.5:0.75}, every even segment/.style  = {blue,opacity=.1}];
\addplot fill between[of = f12 and fxb, soft clip={domain=0:0.5}, every even segment/.style  = {blue,opacity=.1}];
\node[align=center] at (axis cs: 0.625,0.425) {Location\\ of QRE};
\node[align=center] at (axis cs: 0.25,0.175) {Location\\ of QRE};
\node[align=center] at (axis cs: 0.875,0.8) {Location\\ of QRE};
\end{axis}
\end{tikzpicture}
\caption{Geometric locus of QRE in $2\times2$ coordination games for all possible exploration rates in the two cases 
(i) $\xm,\ym\ge1/2$ (upper panel) and (ii) $\xm\ge1/2,\ym<1/2$ (bottom panel) of Theorem~\ref{thm:location_qre}. The blue dots are the NE of the underlying game $\Gamma$ (when exploration is zero). In both panels, the risk-dominant equilibrium is $\(0,0\)$.}
\label{fig:location}
\end{figure}

Depending on whether the interests of both agents are perfectly aligned --- in the sense that $\(u_{11}-u_{22}\)\(v_{11}-v_{22}\)>0$ --- or not, the QRE surface can be disconnected or connected. A formal characterization is provided in Theorem~\ref{thm:location_qre}.

\begin{theorem}[Geometric locus of the QRE equilibria in coordination games]
\label{thm:location_qre}
Consider a two-player, $\N=\{1,2\}$, two-action, $A_1=A_2=\{a_1,a_2\}$, coordination game  $\Gamma=\(\N,\(A_k,u_k\)_{k\in \N}\)$ with payoff functions $\(u_1,u_2\)$ as in equations \eqref{eq:coordination}. If $\xm+\ym>1$, then, for any exploration-exploitation rates $\alpha_x,\beta_x,\alpha_y,\beta_y>0$ it holds that
\begin{enumerate}[label=(\roman*)\,\,, wide=0pt, leftmargin=\parindent, labelsep=0pt,noitemsep]
\item If $\xm,\ym>1/2$, then any QRE $\(x_Q,y_Q\)$ satisfies either $x_Q>\xm, y_Q>\ym$ or $x_Q,y_Q<1/2$.
\item If $\xm>1/2,\ym\le 1/2$, then any QRE $\(x_Q,y_Q\)$ satisfies one of: $x_Q<1/2,y_Q<\ym$, $1/2<x_Q<\xm, \ym<y_Q<1/2$ and $x_Q>\xm,y_Q>\ym$. 
\end{enumerate}
In particular, if $\Gamma$ is symmetric, i.e., if $u_2=u_1^T$, then there exist no symmetric QRE, $\(x_Q,x_Q\)$, with $1/2<x_Q<\xm$.
\end{theorem}

The statement of Theorem~\ref{thm:location_qre} is visualized in Figure~\ref{fig:location}. In the first case, disconnected QRE surface, the dynamics select the risk-dominant equilibrium after a \emph{saddle-node bifurcation} in the exploration phase, regardless of whether it coincides with the payoff dominant equilibrium or not. In the second case, the QRE surface is connected via two branches of saddle-node bifurcations which is consistent with the emergence of a \emph{cusp bifurcation} point. Hence, after exploration which it may rest to either of the two boundary equilibria. In short, the collective outcome of the SQL dynamics depends on the geometry of the QRE surface which is illustrated next.

\section{Experiments: Phase Transitions in Games}
\label{sec:applications}
To visualize the above, we start with $2\times 2$ coordination games and then proceed to potential games with action spaces of arbitrary size. In all cases, we consider two representative exploration-exploitation policies: an \emph{Explore-Then-Exploit} (ETE) policy \citep{Bai20}, which starts with (relatively) high exploration that reduces linearly to zero and a \emph{Cyclical Learning Rate with one cycle} (CLR-1) policy \citep{Smi17}, which starts with low exploration, increases to high exploration around the middle of the cycle and decays to (ultimately) zero exploration (i.e., pure exploitation).\footnote{The findings are qualitatively equivalent for non-linear, e.g., quadratic, changes in the exploration rates in both policies and for more that one learning cycle in the CLR policy.}\\[0.2cm]
\begin{table}
\centering
\setlength{\extrarowheight}{2pt}
\hspace{0.6cm}\begin{tabular}{*{3}{c|}}
\multicolumn{3}{c}{Pareto Coordination}\\
\cline{2-3}
& $a_1$& $a_2$ \\\hline
\multicolumn{1}{|c|}{$a_1$}& $1,1$ & $0,0$ \\\hline
\multicolumn{1}{|c|}{$a_2$} & $0,0$ & $1.5,1.8$ \\\hline
\end{tabular}\hspace{0.6cm}
\begin{tabular}{*{3}{c|}}
\multicolumn{3}{c}{Battle of the Sexes}\\
\cline{2-3}
& $a_1$& $a_2$ \\\hline
\multicolumn{1}{|c|}{$a_1$}& $1.5,1$ & $0,0$ \\\hline
\multicolumn{1}{|c|}{$a_2$} & $0,0$ & $1,2$ \\\hline
\end{tabular}\hspace{0.6cm}
\begin{tabular}{*{3}{c|}}
\multicolumn{3}{c}{Stag Hunt}\\
\cline{2-3}
& $a_1$& $a_2$ \\\hline
\multicolumn{1}{|c|}{$a_1$}& $3,3$ & $0,2$ \\\hline
\multicolumn{1}{|c|}{$a_2$} & $2,0$ & $1.5,1.5$ \\\hline
\end{tabular}
\caption{Payoffs of the games in Section~\ref{sec:applications}.}
\label{tab:games}
\end{table}
\begin{figure}[!htb]
\centering
\includegraphics[width=0.45\linewidth]{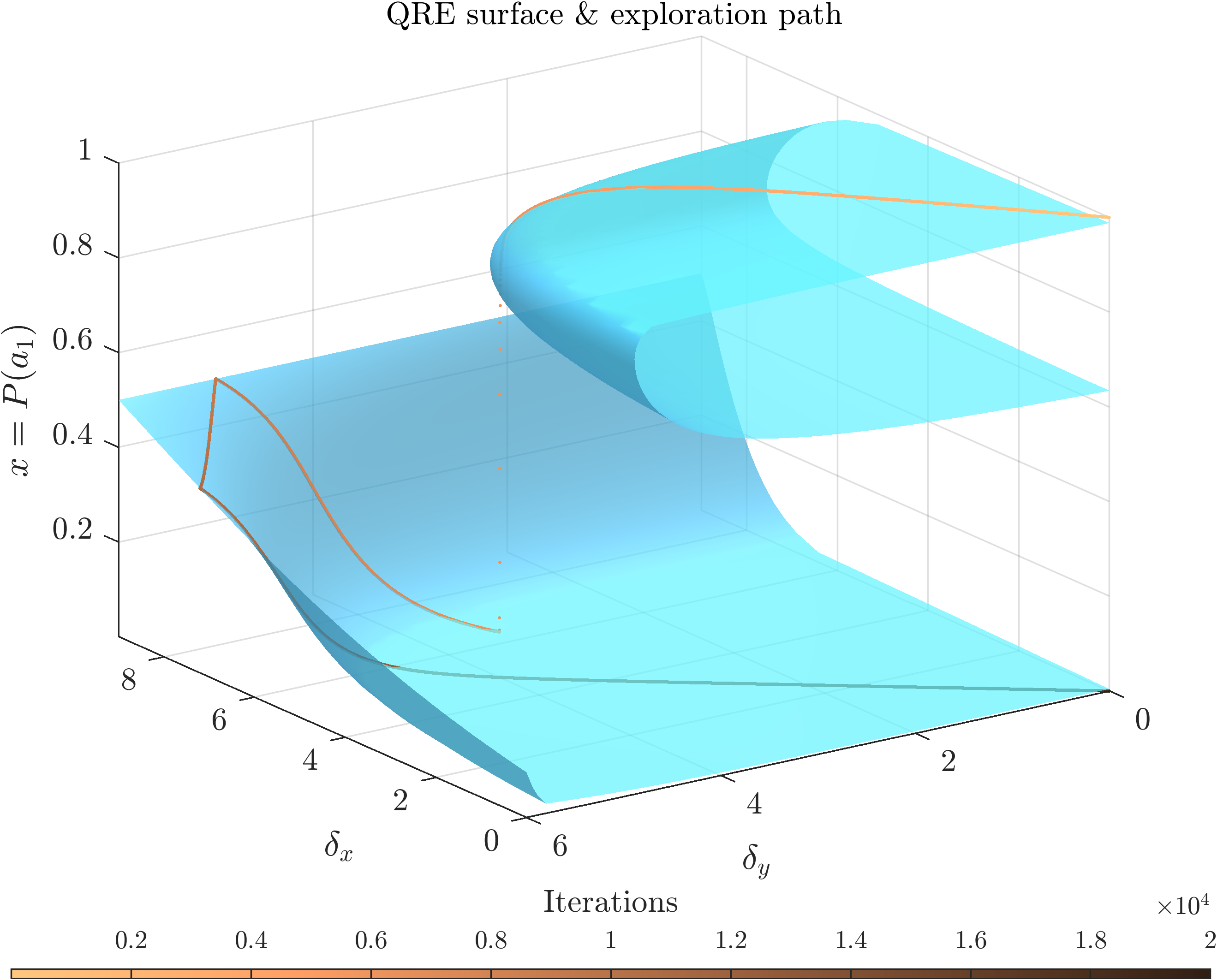} \hspace{15pt}
\includegraphics[width=0.48\linewidth]{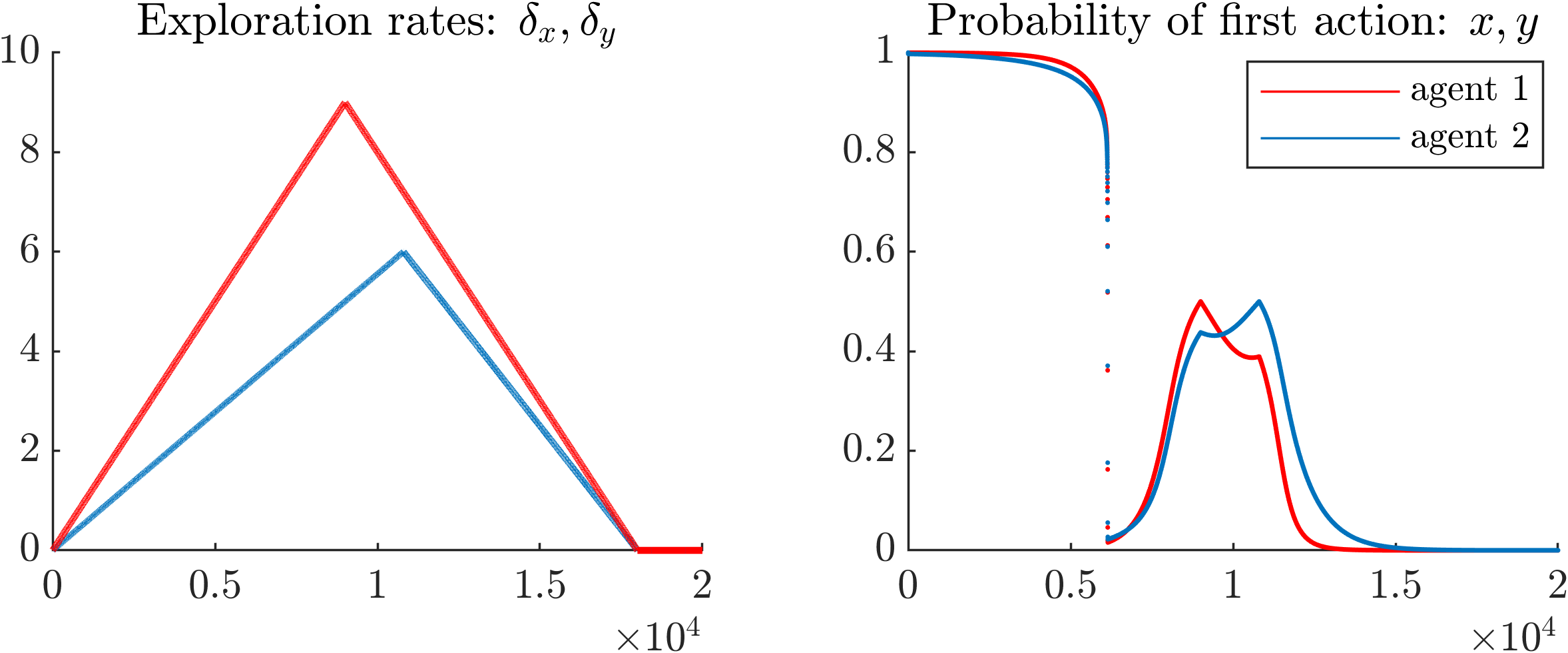}
\caption{SQL in Stag Hunt. The upper panel shows the QRE surface and the exploration path of agent 1 (light to dark line). The bottom panels show the CLR-1 exploration rates (left) and the probability of action 1 during the learning process for both agents (right). As agents increase exploration, their choice distributions undergo a \emph{saddle-node bifurcation} (disconnected surface). This prompts a permanent transition from the vicinity of the payoff dominant action profile, $\(x,y\)=\(1,1\)$, in the upper component of the QRE surface to the (0,0) equilibrium when exploration reduces back to zero (right corner of the lower component).}
\label{fig:payrisk}
\end{figure}
\begin{figure}[!ht]
\centering
\includegraphics[width=0.45\linewidth]{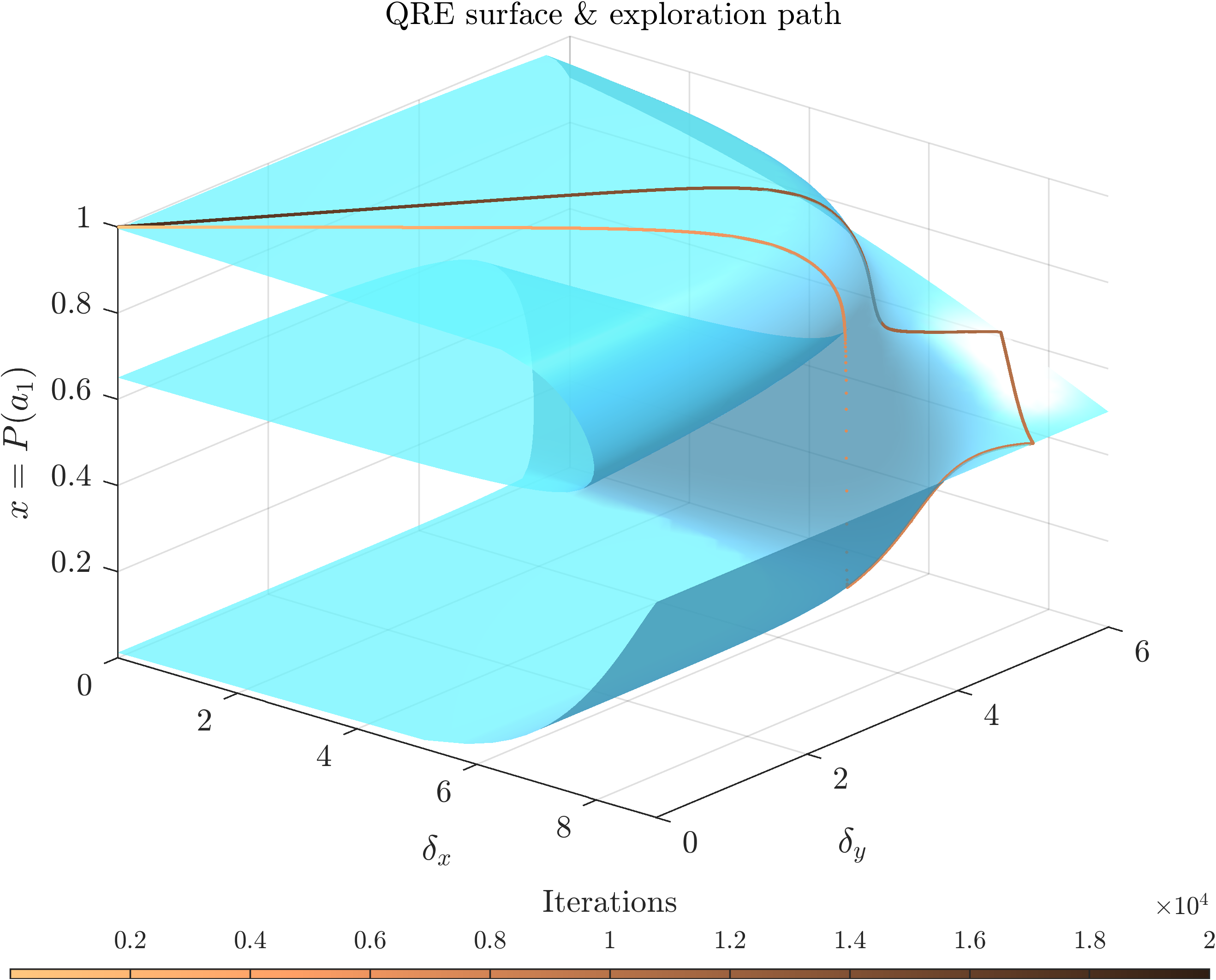}\hspace{18pt}
\includegraphics[width=0.48\linewidth]{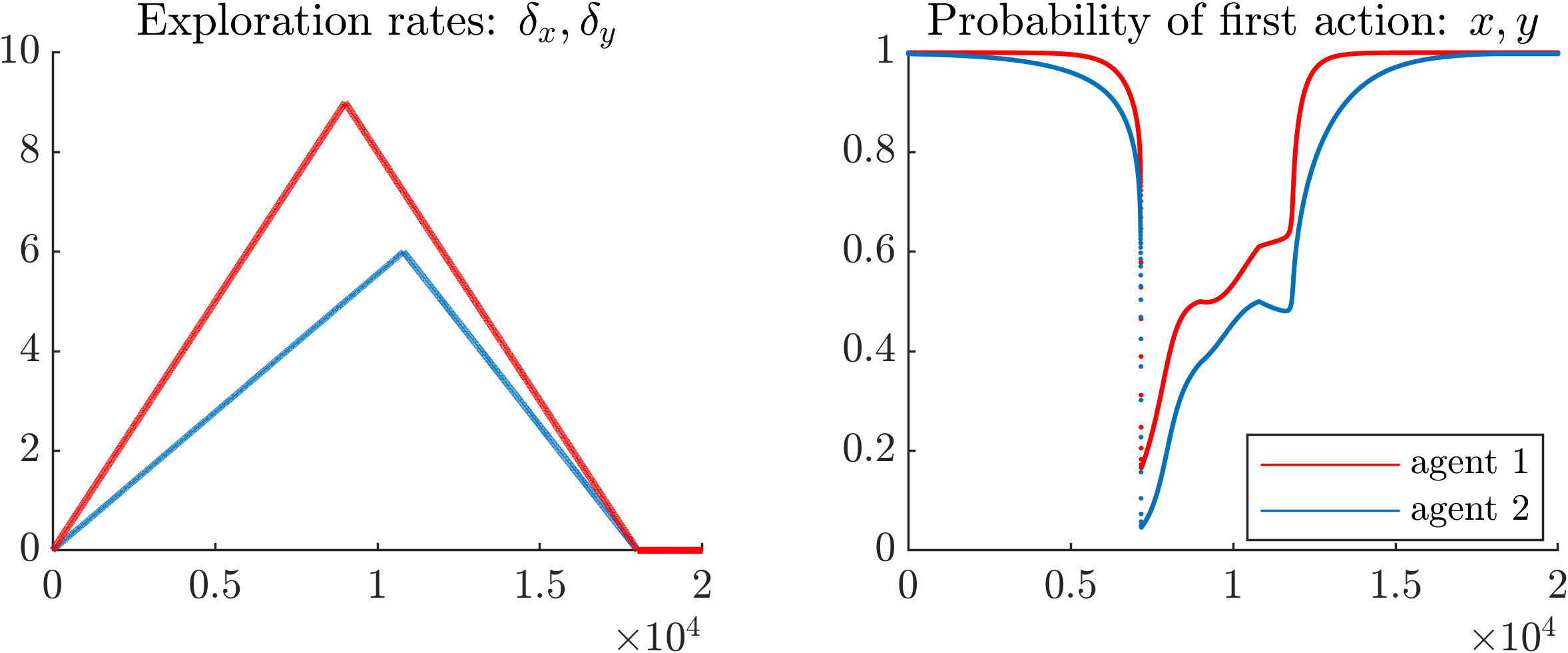}\vspace*{0.5cm}
\includegraphics[width=0.45\linewidth]{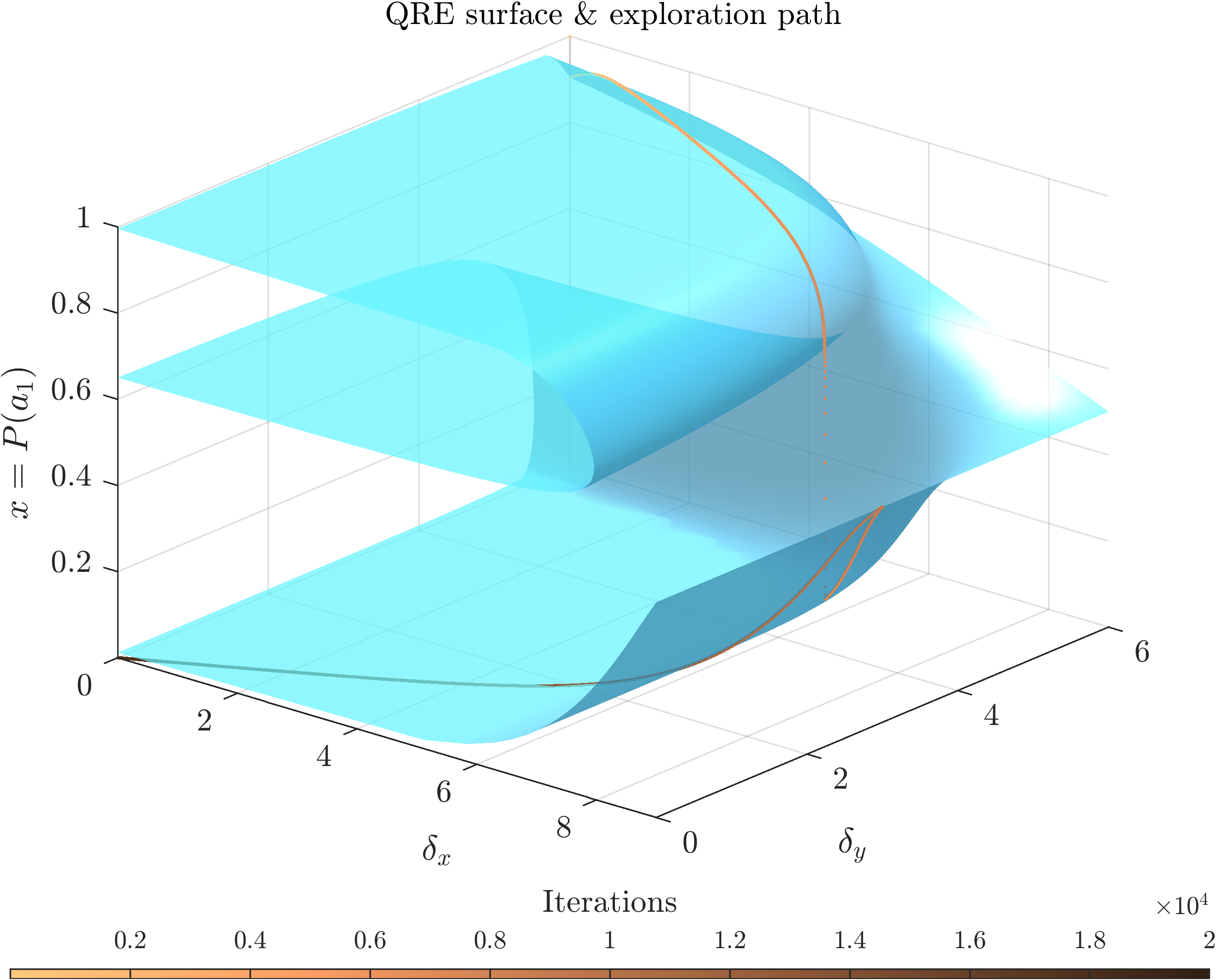}\hspace{18pt}
\includegraphics[width=0.48\linewidth]{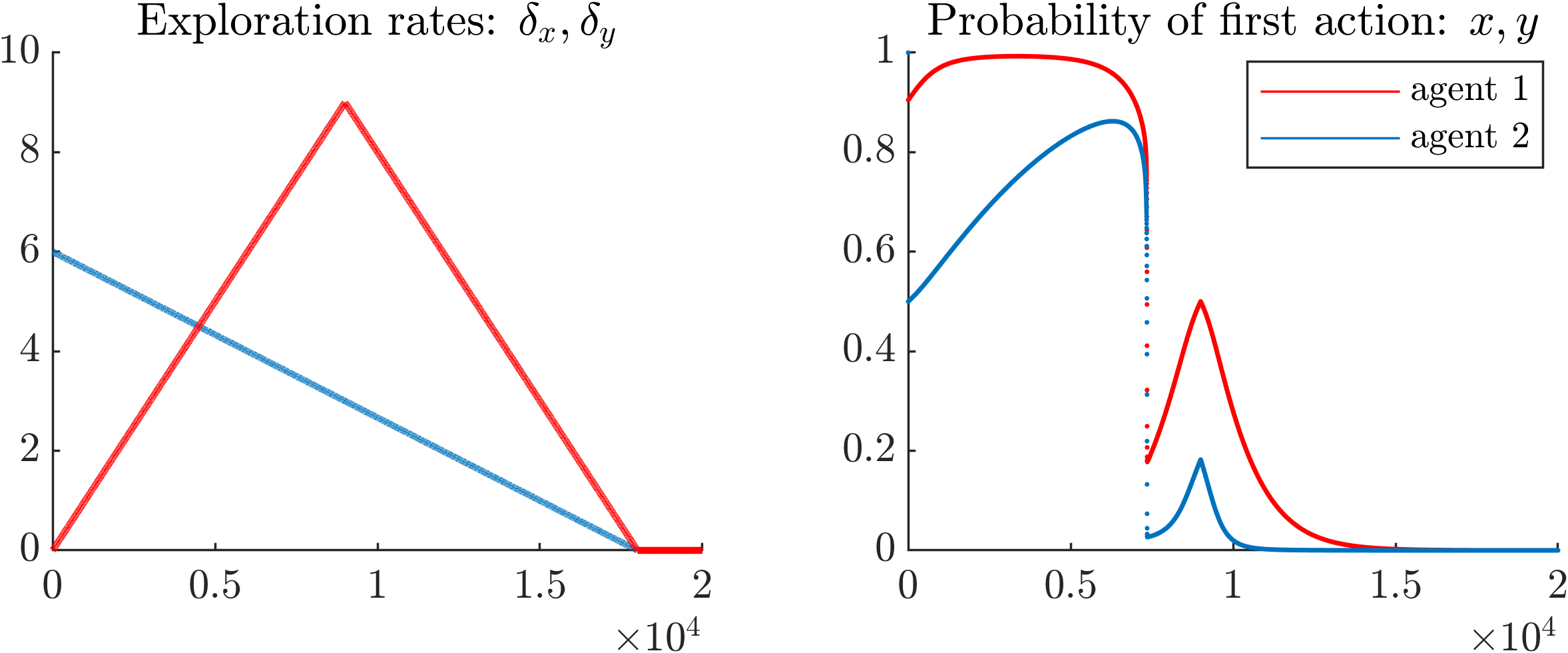}
\caption{Exploration-Exploitation in Battle of the Sexes. In contrast to Stag Hunt, the QRE manifold has two branches of saddle-node bifurcation curves (consistent with the emergence of a co-dimension $2$ cusp point) and the phase transition to the lower part of the QRE surface may not be permanent. These two cases are illustrated via the CLR-1 vs CLR-1 policies (up) and the CLR-1 vs ETE policies (down).}
\label{fig:nopay}
\end{figure}
\textbf{Coordination Games $2\times2$.}
As long as agents' interests are aligned, sufficient exploration even by a single agent leads the learning process (after exploration is reduced back to zero) to the risk dominant equilibrium regardless of whether this equilibrium coincides with the payoff dominant equilibrium or not. Typical realizations of these cases are the Pareto Coordination and Stag Hunt games (Table~\ref{tab:games}).\footnote{In the numerical experiments, we have the used the transformations in Lemma~\ref{lem:recursion} and Remark~\ref{rem:recursion} in Appendix~\ref{app:derivation} which lead to a robust discretization of the ODEs in the theoretical analysis.} \par
In Pareto Coordination, $\(a_2,a_2\)$ is both the risk- and payoff-dominant equilibrum whereas in Stag Hunt, the payoff dominant equilibrium is $\(a_1,a_1\)$. However, in both games $\xm,\ym>1/2$ (due to the aligned interests of the players) which implies that the location of the QRE is described by the upper panel in Figure~\ref{fig:location}. Accordingly, the QRE surface is disconnected and if any agent sufficiently increases their exploration rate, the SQL dynamics converge to the risk-dominant equilibrium independently of the starting point and the exploration policy of the other agent. This is illustrated and explained in Figure~\ref{fig:payrisk} (and in a similar fashion in Figure~\ref{fig:grid} in Appendix~\ref{app:experiments}). Note that in both theses cases, the risk-dominant equilibrium is the global maximizer of the potential function, see Lemma~\ref{lem:risk_dominant} and \cite{Alo10,Sch03}. \par

\begin{figure}[!ht]
\centering
\includegraphics[width=0.312\linewidth]{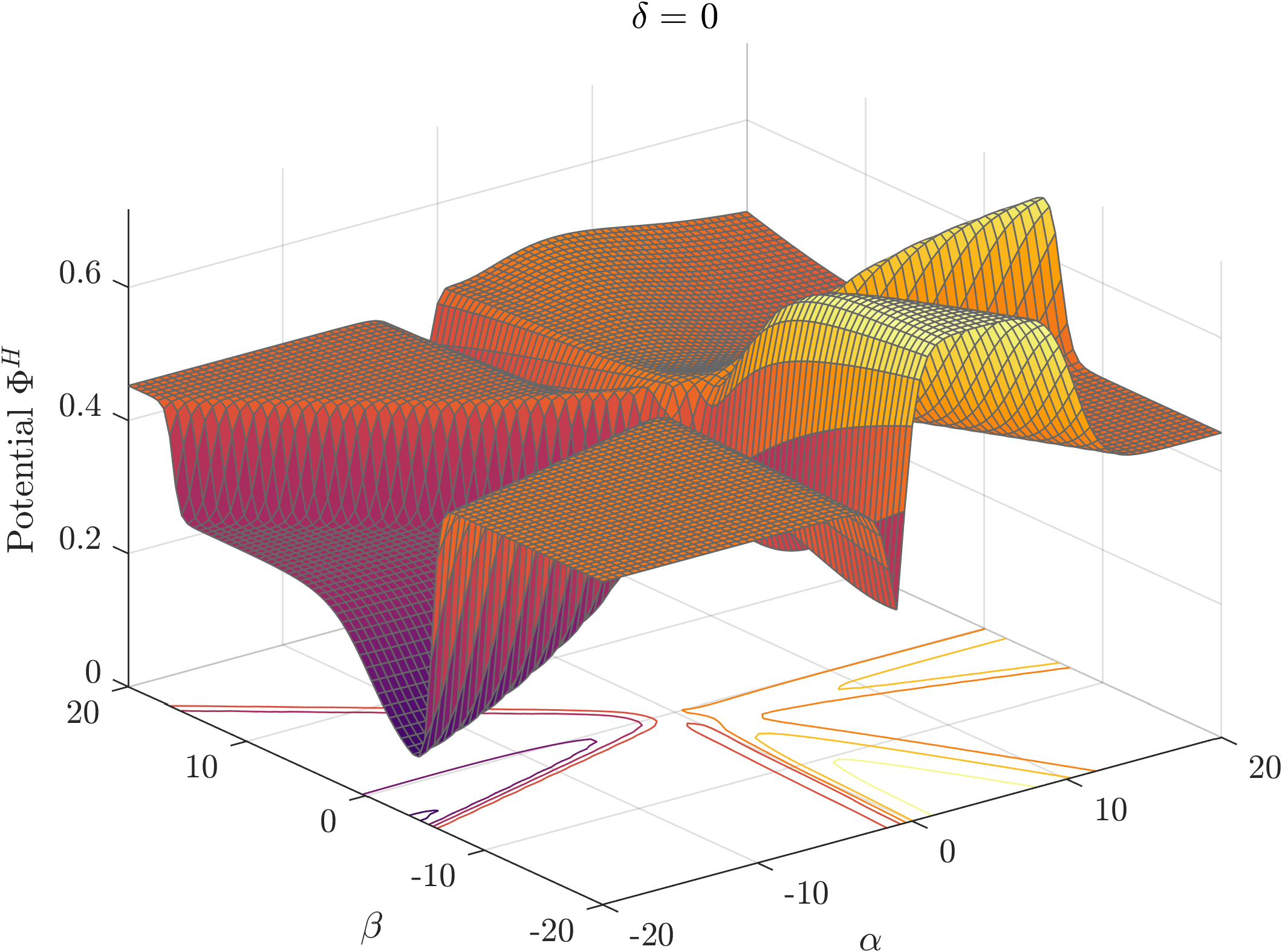}\hspace{10pt}
\includegraphics[width=0.312\linewidth]{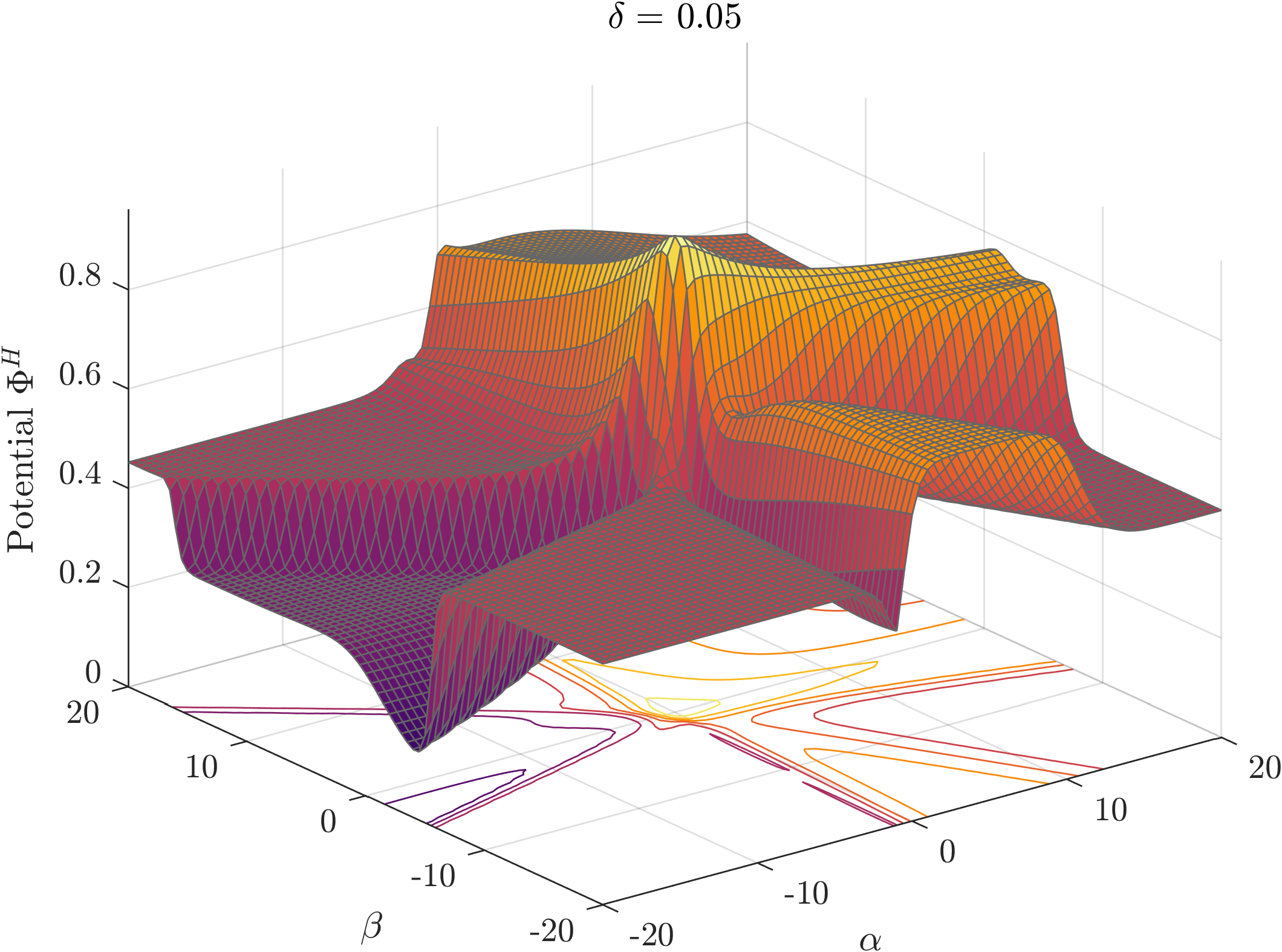}\hspace{10pt}
\includegraphics[width=0.312\linewidth]{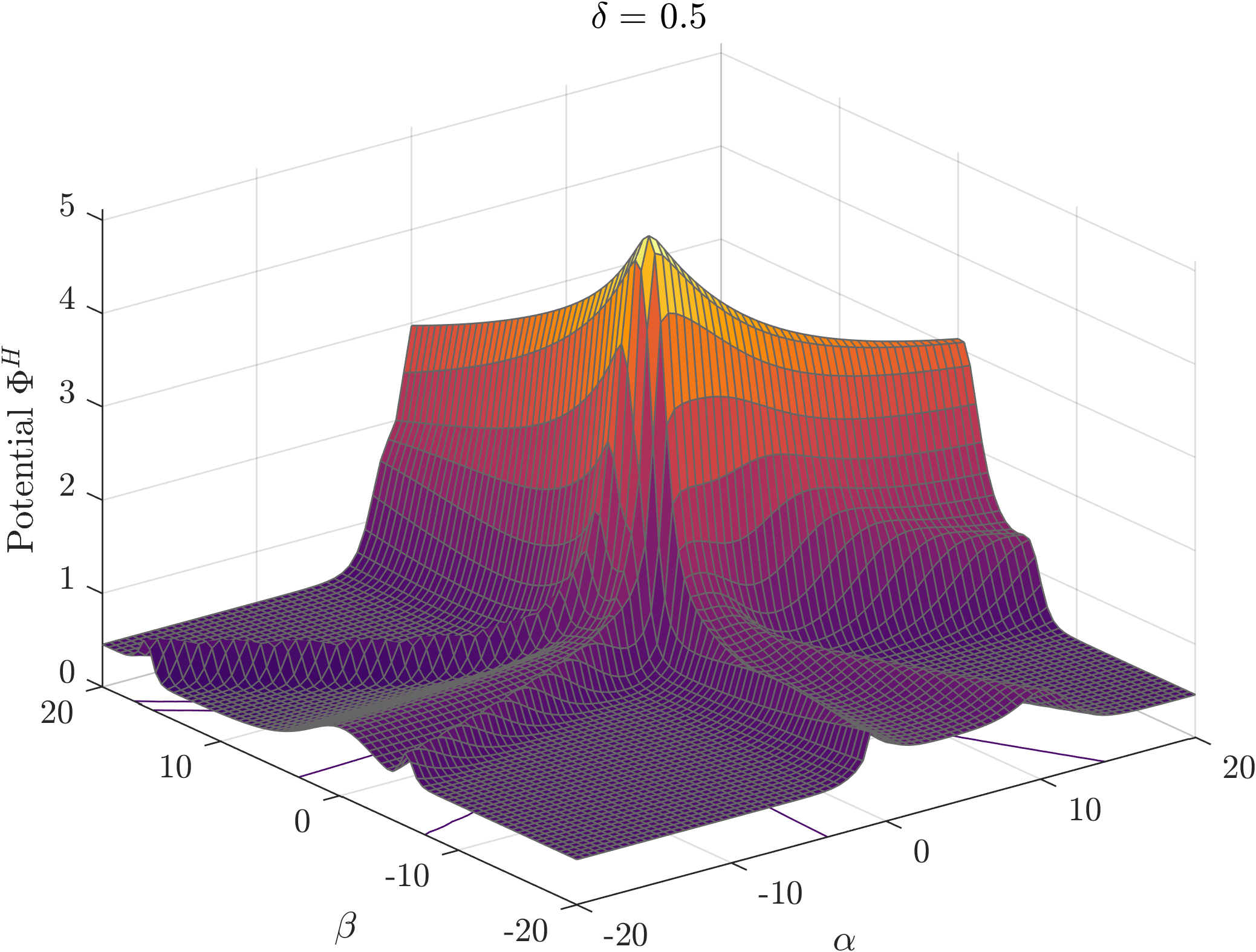}
\caption{Snapshots of the modified potential $\Phi^H$ for different exploration rates in a symmetric 2-player potential game with random payoffs in $[0,1]$. Unlike Figures~\ref{fig:payrisk} and \ref{fig:nopay}, we now visualize the potential function instead of the QRE surface. Hence, we cannot reason about the bifurcation types. However, we see that without exploration, $\delta=0$, the potential (equal to the potential, $\Phi$, of the original game) has various local maxima, whereas as exploration increases, a unique remaining attractor (maximum) forms at the vicinity of the uniform distribution, $\(0,0\)$ in the transformed coordinates.}
\label{fig:large}
\end{figure}

By contrast, if agents' interests are not perfectly aligned, then the outcome of the exploration process is not unambiguous (even if the game remains a coordination game). A representative game of this class, in which no payoff dominant equilibrium exists, is the Battle of the Sexes in Table~\ref{tab:games}. The most preferable outcome is now different for the two agents which implies that there is no payoff dominant equilibrium. However, the pure joint profile $\(a_2,a_2\)$ remains the risk-dominant equilibrium.\footnote{Note that, although well defined, risk-dominance seems to be now less appealing: if agent 1 is completely ignorant about the equilibrium selection of agent 2 (and assigns a uniform distribution to agent 2's actions), then agent 1 is better off to select action $a_1$, despite the fact that $\(a_2,a_2\)$ is the risk-dominant equilibrium.} In this class of games, the location of the QRE is described by the bottom panel in Figure~\ref{fig:location}. The QRE surface is connected and the collective output of the exploration process depends on the exploration policies (timing and intensity) of the two agents. This is illustrated in Figure~\ref{fig:nopay} which denotes two different outcomes of the learning process under the same exploration policy for agent 1 but different exploration policies for agent 2. In Appendix~\ref{app:experiments}, we provide an exhaustive treatment of the possible outcomes under the ETE and CLR-1 exploration policies.\par
Figure~\ref{fig:media} shows the QRE manifolds (surfaces) from a different perspective and highlights the bifurcation curves in the Stag Hunt and Battle of the Sexes games (Pareto Coordination is equivalent to Stag Hunt in this respect). Rotations of the images are included in the Multimedia Appendix.\par

\begin{figure}[!htb]
\includegraphics[width=0.48\linewidth]{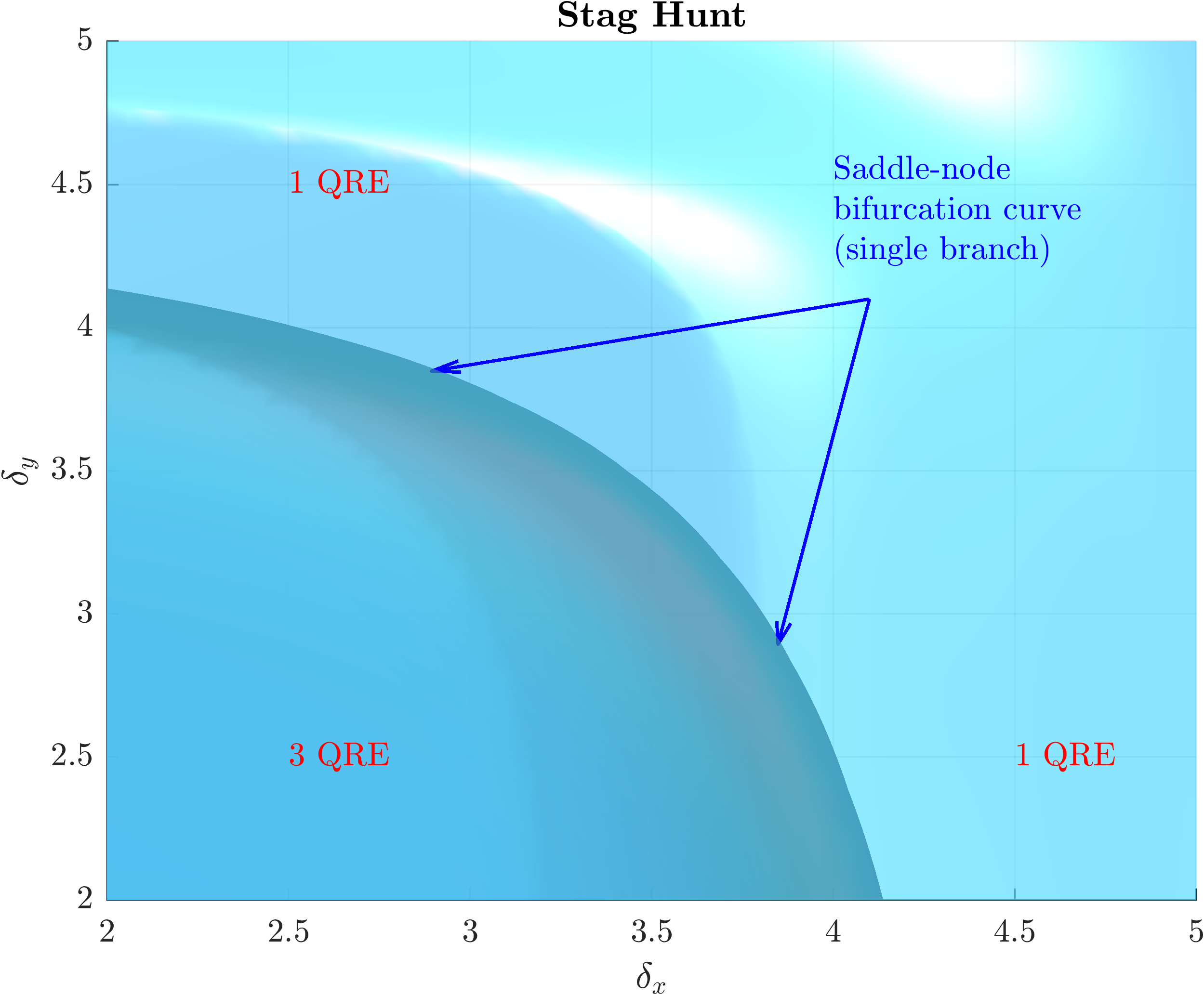}\hspace{10pt}
\includegraphics[width=0.48\linewidth]{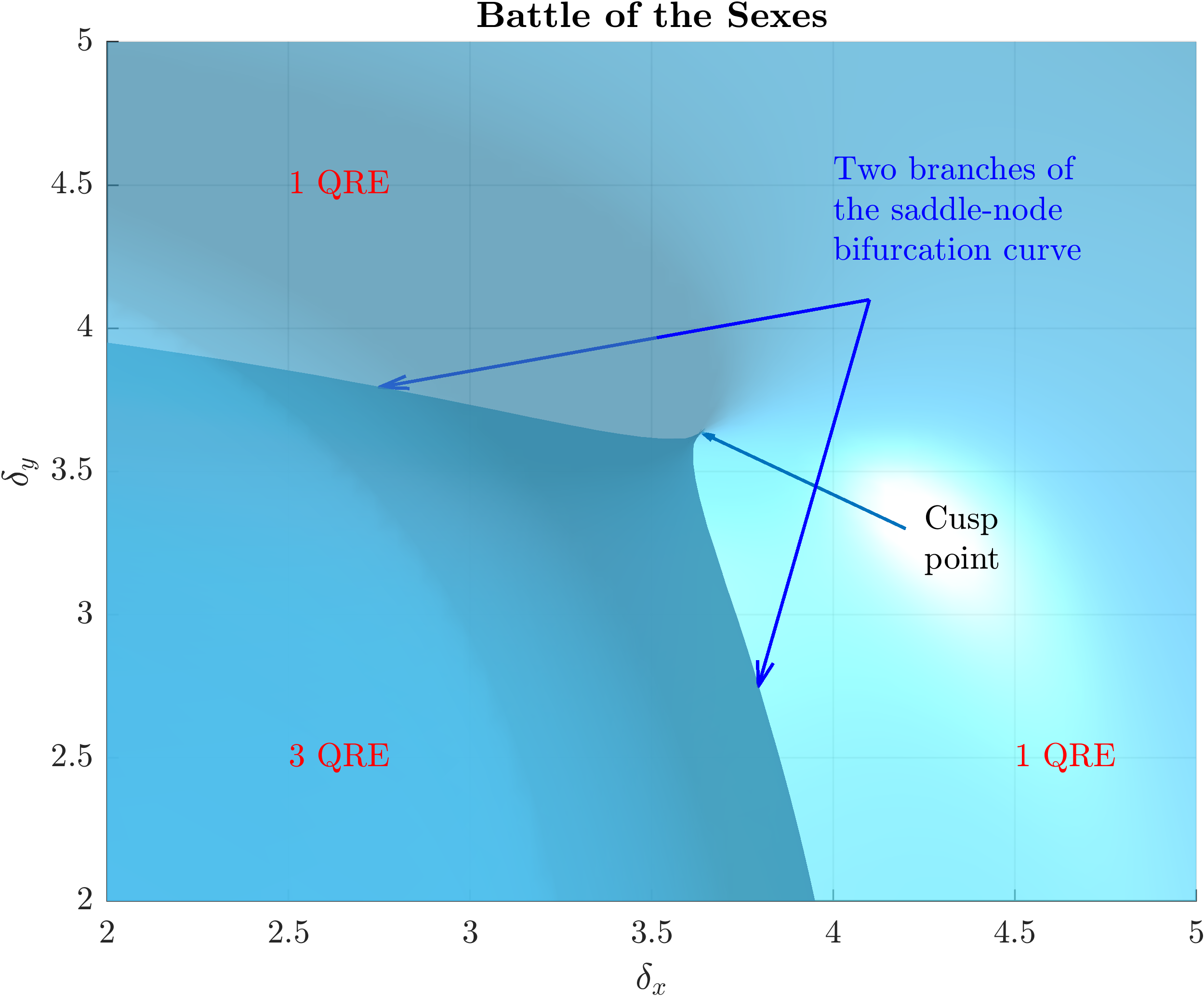}
\caption{Stag Hunt (left panel) and Battle of the Sexes (right panel). The depicted surfaces are the same as in Figures~\ref{fig:intro},\ref{fig:payrisk},\ref{fig:nopay} and \ref{fig:grid}. The current perspective is from the top of the z-axis. The x-y axes have been set between 2 and 5 for a better focus on the bifurcation curve. 360\textdegree{} rotations of these images are shown in the Multimedia Appendix.}
\label{fig:media}
\end{figure}

\textbf{Potential games in larger dimensions}
To visualize the modified potential in equation \eqref{eq:potential} of Lemma~\ref{lem:potential}, we adapt the two-dimensional projection technique of \citet{Li18}. Given a potential game with potential $\Phi$ and $n,m$ actions for agents 1 and 2, we first embed their choice distributions into $\mathbb R^{n+m-2}$ to remove the Simplex restrictions via the transformation $y_i:=\log{x_i/x_n}$ from $\mathbb R^n\to \mathbb R^{n-1}$ (with $\sum_{i=1}^n x_i=1$) for the first agent (and similarly for the second agent) and then choose two arbitrary directions in $\mathbb R^{n+m-2}$ along which we plot the modified potential $\Phi^H\(x\)=\Phi\(x\)+\sum_{k\in\N}\delta_k H\(x_k\), \quad \text{for }x\in X$, cf. equation \eqref{eq:potential}. For simplicity, we keep the exploration ratio $\delta_k:=\beta_k/\alpha_k$ equal to a common $\delta$ for both players.\footnote{A detailed description of the routine is in Appendix~\ref{app:experiments}. This method produces similar visualizations for any number of players.} \par
A visualization of randomly generated 2-player potential game is given in Figure~\ref{fig:large}. As players modify their exploration rates, the SQL dynamics converge to different QRE (local maxima) of these changing surfaces. However, when exploration is large, a single attracting QRE remains (similar to the low dimensional case).\par
In Figure~\ref{fig:potential_10}, we plot the SQL dynamics ($1e-20$ Q-value updates for each of $1e-03$ choice distribution updates) in a $2$-player potential game with $n=10$ actions and potential, $\Phi$, with random values in $[0,10]$. Both agents use CLR-1 policies. Starting from a grid of initial conditions, one close to each pure action pair, the SQL dynamics rest at different local optima before the exploration, converge to the uniform distribution when exploration rates reach their peak and then converge to the same (in this case, global) optimum when exploration is gradually reduced back to zero (horizontal line and vanishing shaded region). 
\begin{figure}[htbp!]
\centering
\includegraphics[width=0.48\linewidth]{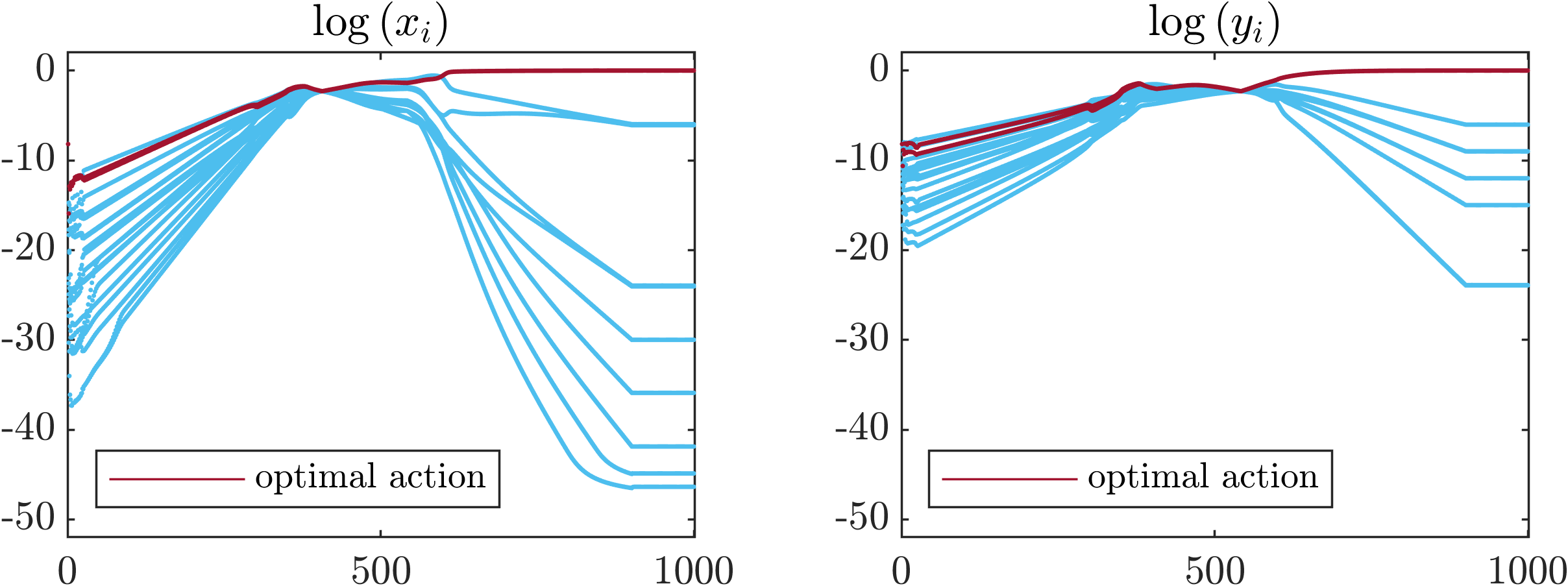}\hspace{10pt}
\includegraphics[width=0.48\linewidth]{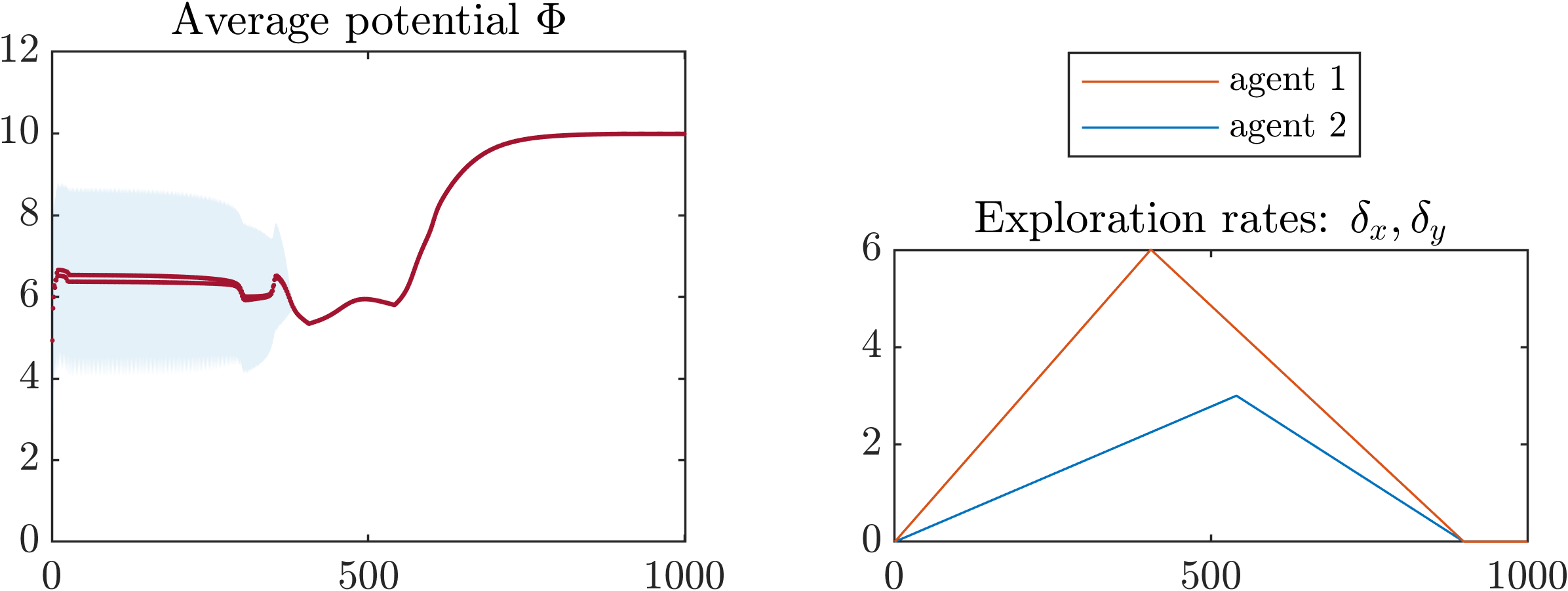}
\caption{Exploration-Exploitation with the SQL dynamics in a potential game with $n=10$ actions. The first two panels show the (log) choice distributions (with the optimal action in different color). The third panel shows the average potential over a set of $10\times10$ different trajectories (starting points) and one standard deviation (shaded region that disappears after all trajectories converge to the same choice distribution). The fourth panel shows the selected CLR-1 policies.}
\label{fig:potential_10}
\end{figure}
\bibliographystyle{plainnat}
\bibliography{catastrophe_bib}

\begin{thebibliography}{55}
\providecommand{\natexlab}[1]{#1}
\providecommand{\url}[1]{\texttt{#1}}
\expandafter\ifx\csname urlstyle\endcsname\relax
  \providecommand{\doi}[1]{doi: #1}\else
  \providecommand{\doi}{doi: \begingroup \urlstyle{rm}\Url}\fi

\bibitem[Alós-Ferrer and Netzer(2010)]{Alo10}
C.~Alós-Ferrer and N.~Netzer.
\newblock {The logit-response dynamics}.
\newblock \emph{Games and Economic Behavior}, 68\penalty0 (2):\penalty0
  413--427, 2010.
\newblock \doi{10.1016/j.geb.2009.08.004}.

\bibitem[Bai and Jin(2020)]{Bai20}
Yu~Bai and Chi Jin.
\newblock {Provable Self-Play Algorithms for Competitive Reinforcement
  Learning}.
\newblock In \emph{Proceedings of the 37th International Conference on Machine
  Learning}, ICML’20, Madison, WI, USA, 2020. Omnipress.

\bibitem[Balduzzi et~al.(2020)Balduzzi, Czarnecki, Anthony, Gemp, Hughes,
  Leibo, Piliouras, and Graepel]{Bal20}
D.~Balduzzi, W.~M. Czarnecki, T.~Anthony, I.~Gemp, E.~Hughes, J.~Leibo,
  G.~Piliouras, and T.~Graepel.
\newblock {Smooth markets: A basic mechanism for organizing gradient-based
  learners}.
\newblock In \emph{International Conference on Learning Representations}, 2020.

\bibitem[Ben-Porat and Tennenholtz(2018)]{Ben18}
O.~Ben-Porat and M.~Tennenholtz.
\newblock {A Game-Theoretic Approach to Recommendation Systems with Strategic
  Content Providers}.
\newblock In S.~Bengio, H.~Wallach, H.~Larochelle, K.~Grauman, N.~Cesa-Bianchi,
  and R.~Garnett, editors, \emph{Advances in Neural Information Processing
  Systems 31}, pages 1110--1120. Curran Associates, Inc., 2018.

\bibitem[Bloembergen et~al.(2015)Bloembergen, Tuyls, Hennes, and
  Kaisers]{Blo15}
D.~Bloembergen, K.~Tuyls, D.~Hennes, and M.~Kaisers.
\newblock {Evolutionary Dynamics of Multi-Agent Learning: A Survey}.
\newblock \emph{J. Artif. Int. Res.}, 53\penalty0 (1):\penalty0 659--697, May
  2015.

\bibitem[Bowling and Veloso(2002)]{Bow02}
Michael Bowling and Manuela Veloso.
\newblock {Multiagent learning using a variable learning rate}.
\newblock \emph{Artificial Intelligence}, 136\penalty0 (2):\penalty0 215--250,
  2002.
\newblock \doi{10.1016/S0004-3702(02)00121-2}.

\bibitem[Cesa-Bianchi and Lugosi(2006)]{cesa2006prediction}
Nicolo Cesa-Bianchi and G{\'a}bor Lugosi.
\newblock \emph{Prediction, learning, and games}.
\newblock Cambridge university press, 2006.

\bibitem[Claus and Boutilier(1998)]{Cla98}
Caroline Claus and Craig Boutilier.
\newblock {The Dynamics of Reinforcement Learning in Cooperative Multiagent
  Systems}.
\newblock In \emph{Proceedings of the Fifteenth National/Tenth Conference on
  Artificial Intelligence/Innovative Applications of Artificial Intelligence},
  AAAI '98/IAAI '98, page 746–752, 1998.

\bibitem[Coucheney et~al.(2015)Coucheney, Gaujal, and Mertikopoulos]{Cou15}
P.~Coucheney, B.~Gaujal, and P.~Mertikopoulos.
\newblock {Penalty-Regulated Dynamics and Robust Learning Procedures in Games}.
\newblock \emph{Mathematics of Operations Research}, 40\penalty0 (3):\penalty0
  611--633, 2015.
\newblock \doi{10.1287/moor.2014.0687}.

\bibitem[{Gao} and {Pavel}(2017)]{Gao17}
Bolin {Gao} and Lacra {Pavel}.
\newblock {On the Properties of the Softmax Function with Application in Game
  Theory and Reinforcement Learning}.
\newblock \emph{arXiv e-prints}, art. arXiv:1704.00805, April 2017.

\bibitem[G{\"o}cke(2002)]{Goc02}
Matthias G{\"o}cke.
\newblock {Various Concepts of Hysteresis Applied in Economics}.
\newblock \emph{Journal of Economic Surveys}, 16\penalty0 (2):\penalty0
  167--188, 2002.
\newblock \doi{10.1111/1467-6419.00163}.

\bibitem[Harsanyi and Selten(1988)]{Har88}
J.C. Harsanyi and R.~Selten.
\newblock \emph{{A General Theory of Equilibrium Selection in Games}}.
\newblock The MIT Press, Massachusetts, USA, 1988.

\bibitem[Ho and Camerer(1998)]{EWA2}
Teck-Hua Ho and Colin Camerer.
\newblock Experience-weighted attraction learning in coordination games:
  Probability rules, heterogeneity, and time-variation.
\newblock \emph{Journal of mathematical psychology}, 42:\penalty0 305--326,
  1998.

\bibitem[Ho and Camerer(1999)]{EWA1}
Teck-Hua Ho and Colin Camerer.
\newblock Experience-weighted attraction learning in normal form games.
\newblock \emph{Econometrica}, 67:\penalty0 827--874, 1999.

\bibitem[Ho et~al.(2007)Ho, Camerer, and Chong]{EWA3}
Teck-Hua Ho, Colin~F. Camerer, and Juin-Kuan Chong.
\newblock Self-tuning experience weighted attraction learning in games.
\newblock \emph{Journal of economic theory}, 133:\penalty0 177--198, 2007.

\bibitem[Kaelbling et~al.(1996)Kaelbling, Littman, and Moore]{Kae96}
L.~P. Kaelbling, M.~L. Littman, and A.~W. Moore.
\newblock Reinforcement learning: A survey.
\newblock \emph{Journal of Artificial Intelligence Research}, 4\penalty0
  (1):\penalty0 237--285, 1996.
\newblock \doi{10.1613/jair.301}.

\bibitem[Kaisers et~al.(2009)Kaisers, Tuyls, Parsons, and Thuijsman]{Kai09}
M.~Kaisers, K.~Tuyls, S.~Parsons, and F.~Thuijsman.
\newblock {An Evolutionary Model of Multi-Agent Learning with a Varying
  Exploration Rate}.
\newblock In \emph{Proceedings of The 8th International Conference on
  Autonomous Agents and Multiagent Systems - Volume 2}, AAMAS '09, page
  1255–1256, 2009.

\bibitem[Kaisers and Tuyls(2010)]{Kai10}
Michael Kaisers and Karl Tuyls.
\newblock Frequency adjusted multi-agent q-learning.
\newblock In \emph{Proceedings of the 9th International Conference on
  Autonomous Agents and Multiagent Systems: Volume 1 - Volume 1}, AAMAS '10,
  page 309–316, 2010.

\bibitem[Kaisers and Tuyls(2011)]{Kai11}
Michael Kaisers and Karl Tuyls.
\newblock {FAQ-Learning in Matrix Games: Demonstrating Convergence near Nash
  Equilibria, and Bifurcation of Attractors in the Battle of Sexes}.
\newblock In \emph{Proceedings of the 13th AAAI Conference on Interactive
  Decision Theory and Game Theory}, AAAIWS'11-13, page 36–42. AAAI Press,
  2011.

\bibitem[Kianercy and Galstyan(2012)]{Kia12}
Ardeshir Kianercy and Aram Galstyan.
\newblock {Dynamics of Boltzmann $Q$ learning in two-player two-action games}.
\newblock \emph{Phys. Rev. E}, 85:\penalty0 041145, Apr 2012.
\newblock \doi{10.1103/PhysRevE.85.041145}.

\bibitem[Kim(1996)]{Kim96}
Y.~Kim.
\newblock {Equilibrium Selection in n-Person Coordination Games}.
\newblock \emph{Games and Economic Behavior}, 15\penalty0 (2):\penalty0
  203--227, 1996.
\newblock \doi{10.1006/game.1996.0066}.

\bibitem[Kleinberg et~al.(2009)Kleinberg, Piliouras, and Tardos]{Kle09}
R.~Kleinberg, G.~Piliouras, and E.~Tardos.
\newblock {Multiplicative Updates Outperform Generic No-Regret Learning in
  Congestion Games}.
\newblock In \emph{Proceedings of the Forty-First Annual ACM Symposium on
  Theory of Computing}, STOC ’09, pages 533--542, 2009.
\newblock \doi{10.1145/1536414.1536487}.

\bibitem[Klos et~al.(2010)Klos, {Van Ahee}, and Tuyls]{Klo10}
T.~Klos, G.~J. {Van Ahee}, and K.~Tuyls.
\newblock Evolutionary dynamics of regret minimization.
\newblock In J.~L. Balc{\'a}zar, F.~Bonchi, A.~Gionis, and M.~Sebag, editors,
  \emph{Machine Learning and Knowledge Discovery in Databases}, pages 82--96,
  Berlin, Heidelberg, 2010. Springer Berlin Heidelberg.

\bibitem[Kwoon and Mertikopoulos(2017)]{Kwo17}
J.~Kwoon and P.~Mertikopoulos.
\newblock {A continuous-time approach to online optimization}.
\newblock \emph{Journal of Dynamics \& Games}, 4\penalty0 (2):\penalty0
  125--148, 2017.
\newblock \doi{10.3934/jdg.2017008}.

\bibitem[Lanctot et~al.(2017)Lanctot, Zambaldi, Gruslys, Lazaridou, Tuyls,
  Perolat, Silver, and Graepel]{Lan17}
M.~Lanctot, V.~Zambaldi, A.~Gruslys, A.~Lazaridou, K.~Tuyls, J.~Perolat,
  D.~Silver, and T.~Graepel.
\newblock {A Unified Game-Theoretic Approach to Multiagent Reinforcement
  Learning}.
\newblock In I.~Guyon, U.~V. Luxburg, S.~Bengio, H.~Wallach, R.~Fergus,
  S.~Vishwanathan, and R.~Garnett, editors, \emph{Advances in Neural
  Information Processing Systems 30}, pages 4190--4203. Curran Associates,
  Inc., 2017.

\bibitem[Leslie and Collins(2005)]{Les05}
D.~S. Leslie and E.~J. Collins.
\newblock {Individual Q-Learning in Normal Form Games}.
\newblock \emph{SIAM Journal on Control and Optimization}, 44\penalty0
  (2):\penalty0 495--514, 2005.
\newblock \doi{10.1137/S0363012903437976}.

\bibitem[Li et~al.(2018)Li, Xu, Taylor, Studer, and Goldstein]{Li18}
Hao Li, Zheng Xu, Gavin Taylor, Christoph Studer, and Tom Goldstein.
\newblock {Visualizing the Loss Landscape of Neural Nets}.
\newblock In S.~Bengio, H.~Wallach, H.~Larochelle, K.~Grauman, N.~Cesa-Bianchi,
  and R.~Garnett, editors, \emph{Advances in Neural Information Processing
  Systems 31}, pages 6389--6399. Curran Associates, Inc., 2018.

\bibitem[{Mazumdar} et~al.(2018){Mazumdar}, {Ratliff}, and {Shankar
  Sastry}]{2018arXiv180405464M}
Eric {Mazumdar}, Lillian~J. {Ratliff}, and S.~{Shankar Sastry}.
\newblock {On Gradient-Based Learning in Continuous Games}.
\newblock \emph{arXiv e-prints}, art. arXiv:1804.05464, April 2018.

\bibitem[McKelvey and Palfrey(1995)]{Mck95}
R.~D. McKelvey and T.~R. Palfrey.
\newblock {Quantal Response Equilibria for Normal Form Games}.
\newblock \emph{Games and Economic Behavior}, 10\penalty0 (1):\penalty0 6--38,
  1995.
\newblock \doi{10.1006/game.1995.1023}.

\bibitem[Mertikopoulos and Sandholm(2016)]{Mer16}
P.~Mertikopoulos and W.~H. Sandholm.
\newblock {Learning in Games via Reinforcement and Regularization}.
\newblock \emph{Mathematics of Operations Research}, 41\penalty0 (4):\penalty0
  1297--1324, 2016.
\newblock \doi{10.1287/moor.2016.0778}.

\bibitem[Mertikopoulos et~al.(2018)Mertikopoulos, Papadimitriou, and
  Piliouras]{Mer18}
P.~Mertikopoulos, C.~Papadimitriou, and G.~Piliouras.
\newblock {Cycles in Adversarial Regularized Learning}.
\newblock In \emph{Proceedings of the Twenty-Ninth Annual ACM-SIAM Symposium on
  Discrete Algorithms}, SODA ’18, page 2703–2717, USA, 2018. Society for
  Industrial and Applied Mathematics.

\bibitem[Monderer and Shapley(1996)]{Mon96}
D.~Monderer and L.~S. Shapley.
\newblock {Potential Games}.
\newblock \emph{Games and Economic Behavior}, 14\penalty0 (1):\penalty0
  124--143, 1996.
\newblock \doi{10.1006/game.1996.0044}.

\bibitem[Omidshafiei et~al.(2019)Omidshafiei, Papadimitriou, Piliouras, Tuyls,
  Rowland, Lespiau, Czarnecki, Lanctot, Perolat, and
  Munos]{omidshafiei2019alpha}
Shayegan Omidshafiei, Christos Papadimitriou, Georgios Piliouras, Karl Tuyls,
  Mark Rowland, Jean-Baptiste Lespiau, Wojciech~M Czarnecki, Marc Lanctot,
  Julien Perolat, and Remi Munos.
\newblock a-rank: Multi-agent evaluation by evolution.
\newblock \emph{arXiv preprint arXiv:1903.01373}, 2019.

\bibitem[Palaiopanos et~al.(2017)Palaiopanos, Panageas, and Piliouras]{Pal17}
G.~Palaiopanos, I.~Panageas, and G.~Piliouras.
\newblock {Multiplicative Weights Update with Constant Step-Size in Congestion
  Games: Convergence, Limit Cycles and Chaos}.
\newblock In \emph{Proceedings of the 31st International Conference on Neural
  Information Processing Systems}, NIPS’17, page 5874–5884. Curran
  Associates Inc., 2017.
\newblock \doi{10.5555/3295222.3295337}.

\bibitem[Panageas and Piliouras(2016)]{Pan16}
I.~Panageas and G.~Piliouras.
\newblock {Average Case Performance of Replicator Dynamics in Potential Games
  via Computing Regions of Attraction}.
\newblock In \emph{Proceedings of the 2016 ACM Conference on Economics and
  Computation}, EC ’16, page 703–720, 2016.
\newblock \doi{10.1145/2940716.2940784}.

\bibitem[Panait and Luke(2005)]{Pan05}
Liviu Panait and Sean Luke.
\newblock {Cooperative Multi-Agent Learning: The State of the Art}.
\newblock \emph{Autonomous Agents and Multi-Agent Systems}, 11\penalty0
  (3):\penalty0 387--434, Nov 2005.
\newblock \doi{10.1007/s10458-005-2631-2}.

\bibitem[Perolat et~al.(2020)Perolat, Munos, Lespiau, Omidshafiei, Rowland,
  Ortega, Burch, Anthony, Balduzzi, {De Vylder}, Piliouras, Lanctot, and
  Tuyls]{Per20}
J.~Perolat, R.~Munos, J.-B. Lespiau, S.~Omidshafiei, M.~Rowland, P.~Ortega,
  N.~Burch, T.~Anthony, D.~Balduzzi, B.~{De Vylder}, G.~Piliouras, M.~Lanctot,
  and K.~Tuyls.
\newblock {From Poincar{\'e} Recurrence to Convergence in Imperfect Information
  Games: Finding Equilibrium via Regularization}, 2020.

\bibitem[{Puigdom{\`e}nech Badia} et~al.(2020){Puigdom{\`e}nech Badia}, {Piot},
  {Kapturowski}, {Sprechmann}, {Vitvitskyi}, {Guo}, and
  {Blundell}]{2020arXiv200313350P}
Adri{\`a} {Puigdom{\`e}nech Badia}, Bilal {Piot}, Steven {Kapturowski}, Pablo
  {Sprechmann}, Alex {Vitvitskyi}, Daniel {Guo}, and Charles {Blundell}.
\newblock {Agent57: Outperforming the Atari Human Benchmark}.
\newblock \emph{arXiv e-prints}, art. arXiv:2003.13350, March 2020.

\bibitem[Romero(2015)]{Rom15}
J.~Romero.
\newblock {The effect of hysteresis on equilibrium selection in coordination
  games}.
\newblock \emph{Journal of Economic Behavior \& Organization}, 111:\penalty0
  88--105, 2015.
\newblock \doi{10.1016/j.jebo.2014.12.029}.

\bibitem[Roughgarden(2015)]{Rou15}
Tim Roughgarden.
\newblock {Intrinsic Robustness of the Price of Anarchy}.
\newblock \emph{J. ACM}, 62\penalty0 (5), 2015.
\newblock \doi{10.1145/2806883}.

\bibitem[Rowland et~al.(2019)Rowland, Omidshafiei, Tuyls, Perolat, Valko,
  Piliouras, and Munos]{Row19}
M.~Rowland, S.~Omidshafiei, K.~Tuyls, J.~Perolat, M.~Valko, G.~Piliouras, and
  R.~Munos.
\newblock {Multiagent Evaluation under Incomplete Information}.
\newblock In H.~Wallach, H.~Larochelle, A.~Beygelzimer, F.~d\textquotesingle
  Alch\'{e}-Buc, E.~Fox, and R.~Garnett, editors, \emph{Advances in Neural
  Information Processing Systems 32}, pages 12291--12303. Curran Associates,
  Inc., 2019.

\bibitem[Sanders et~al.(2018)Sanders, Farmer, and Galla]{San18}
James B.~T. Sanders, J.~D. Farmer, and T.~Galla.
\newblock {The prevalence of chaotic dynamics in games with many players}.
\newblock \emph{Scientific Reports}, 8\penalty0 (1):\penalty0 4902, 2018.
\newblock \doi{10.1038/s41598-018-22013-5}.

\bibitem[Sato and Crutchfield(2003)]{Sat03}
Y.~Sato and J.~P. Crutchfield.
\newblock {Coupled replicator equations for the dynamics of learning in
  multiagent systems}.
\newblock \emph{Phys. Rev. E}, 67:\penalty0 015206, Jan 2003.
\newblock \doi{10.1103/PhysRevE.67.015206}.

\bibitem[Sato et~al.(2005)Sato, Akiyama, and Crutchfield]{Sat05}
Y.~Sato, E.~Akiyama, and J.~P. Crutchfield.
\newblock {Stability and diversity in collective adaptation}.
\newblock \emph{Physica D:Nonlinear Phenomena}, 210\penalty0 (1):\penalty0
  21--57, 2005.
\newblock \doi{10.1016/j.physd.2005.06.031}.

\bibitem[Schmidt et~al.(2003)Schmidt, Shupp, Walker, and Ostrom]{Sch03}
D.~Schmidt, R.~Shupp, J.~M. Walker, and E.~Ostrom.
\newblock {Playing safe in coordination games:: the roles of risk dominance,
  payoff dominance, and history of play}.
\newblock \emph{Games and Economic Behavior}, 42\penalty0 (2):\penalty0
  281--299, 2003.
\newblock \doi{10.1016/S0899-8256(02)00552-3}.

\bibitem[{Smith} and {Topin}(2017)]{Smi17}
Leslie~N. {Smith} and Nicholay {Topin}.
\newblock {Super-Convergence: Very Fast Training of Neural Networks Using Large
  Learning Rates}.
\newblock \emph{arXiv e-prints}, page arXiv:1708.07120, 2017.

\bibitem[Strogatz(2000)]{Str00}
S.~H. Strogatz.
\newblock \emph{{Nonlinear Dynamics and Chaos: With Applications to Physics,
  Biology, Chemistry, and Engineering}}.
\newblock Westview Press (Studies in nonlinearity collection), Cambridge, MA,
  USA, first edition, 2000.

\bibitem[Swenson et~al.(2018)Swenson, Murray, and Kar]{Swe18}
B.~Swenson, R.~Murray, and S.~Kar.
\newblock {On Best-Response Dynamics in Potential Games}.
\newblock \emph{SIAM Journal on Control and Optimization}, 56\penalty0
  (4):\penalty0 2734--2767, 2018.
\newblock \doi{10.1137/17M1139461}.

\bibitem[Tuyls et~al.(2003)Tuyls, Verbeeck, and Lenaerts]{Tuy03}
K.~Tuyls, K.~Verbeeck, and T.~Lenaerts.
\newblock {A Selection-mutation Model for Q-learning in Multi-agent Systems}.
\newblock In \emph{Proceedings of the Second International Joint Conference on
  Autonomous Agents and Multiagent Systems}, AAMAS '03, pages 693--700, 2003.
\newblock \doi{10.1145/860575.860687}.

\bibitem[Tuyls et~al.(2006)Tuyls, Hoen, and Vanschoenwinkel]{Tuy06}
K.~Tuyls, P.~J.~T. Hoen, and B.~Vanschoenwinkel.
\newblock {An Evolutionary Dynamical Analysis of Multi-Agent Learning in
  Iterated Games}.
\newblock \emph{{Autonomous Agents and Multi-Agent Systems}}, 12\penalty0
  (1):\penalty0 115--153, 2006.
\newblock \doi{10.1007/s10458-005-3783-9}.

\bibitem[Tuyls and Weiss(2012)]{Tuy12}
Karl Tuyls and Gerhard Weiss.
\newblock Multiagent learning: Basics, challenges, and prospects.
\newblock \emph{AI Magazine}, 33\penalty0 (3):\penalty0 41, Sep. 2012.
\newblock \doi{10.1609/aimag.v33i3.2426}.

\bibitem[Watkins and Dayan(1992)]{Wat92}
C.~J.C.H. Watkins and P.~Dayan.
\newblock {Technical Note: Q-Learning}.
\newblock \emph{Machine Learning}, 8\penalty0 (3):\penalty0 279--292, May 1992.
\newblock \doi{10.1023/A:1022676722315}.

\bibitem[Wolpert et~al.(2012)Wolpert, Harr{\'e}, Olbrich, Bertschinger, and
  Jost]{Wol12}
D.~H. Wolpert, M.~Harr{\'e}, E.~Olbrich, N.~Bertschinger, and J.~Jost.
\newblock {Hysteresis effects of changing the parameters of noncooperative
  games}.
\newblock \emph{Phys. Rev. E}, 85:\penalty0 036102, Mar 2012.
\newblock \doi{10.1103/PhysRevE.85.036102}.

\bibitem[Wunder et~al.(2010)Wunder, Littman, and Babes]{Wun10}
M.~Wunder, M.~Littman, and M.~Babes.
\newblock {Classes of Multiagent Q-Learning Dynamics with $\epsilon$-Greedy
  Exploration}.
\newblock In \emph{Proceedings of the 27th International Conference on
  International Conference on Machine Learning}, ICML’10, page 1167–1174,
  Madison, WI, USA, 2010. Omnipress.

\bibitem[Yang et~al.(2017)Yang, Piliouras, and Basanta]{Yan17}
G.~Yang, G.~Piliouras, and D.~Basanta.
\newblock {Bifurcation Mechanism Design - from Optimal Flat Taxes to Improved
  Cancer Treatments}.
\newblock In \emph{Proceedings of the 2017 ACM Conference on Economics and
  Computation}, EC '17, pages 587--587, 2017.
\newblock \doi{10.1145/3033274.3085144}.

\end{thebibliography}

\appendix
\section{Derivation of SQL Dynamics}
\label{app:derivation}
The mathematical connection between Q-learning and the dynamics in \eqref{eq:main} via a smoothening process (continuous time limit) at which the Q-values are interpreted as Boltzmann probabilities for the action selection can be found in \cite{Tuy03} and \cite{Sat03} among others. To make the paper self-contained, we repeat here the main arguments. Each agent $k\in \N$ keeps track of the past performance of their actions $i\in A_k$ via the Q-learning update rule
\begin{equation}\label{eq:update_q}
Q_{ki}\(n+1\)=Q_{ki}\(n\)+\alpha_k\lt r_{ki}\(n\)-Q_{ki}\(n\)\rt,\;i \in A_k
\end{equation}
where $n\in\mathbb N$ denotes the time (in discrete steps) and $\alpha_k\in[0,1]$ the learning rate or memory decay of agent $k$, cf. \cite{Sat03,Kia12}. Here $Q_{ki}\(n\)$ is called the \emph{memory} of agent $k$ about the performance of action $i\in A_k$ up to time step $n\in \mathbb N$. Agent $k\in \N$ updates their actions (choice distributions) according to a Boltzmann type distribution, with 
\begin{equation}\label{eq:update_x}
x_{ki}\(n\)=\frac{\exp{\(\beta_kQ_{ki}\(n\)\)}}{\sum_{j\in A_k}\exp{\(\beta_kQ_{kj}\(n\)\)}},\quad \text{for each } i\in A_k,
\end{equation}
where $\beta_k\in\lt0,+\infty\)$ denotes agent $k$'s learning sensitivity or adaptation, i.e., how much the choice distribution is affected by the past performance. Higher values of $\beta_k$ indicate a higher exploitation rate, i.e., proclivity of the agent towards the best performing action, whereas values of $\beta_k$ close to $0$ lead to higher exploration or randomization among the agent's available choices in $A_k$. Combining the above equations, one obtains the recursive equation of the agent's choice distribution
\begin{align*}&x_{ki}\(n+1\)=\frac{x_{ki}\(n\)\exp{\(\beta_k\(Q_{ki}\(n+1\)-Q_{ki}\(n\)\)\)}}{\sum_{j\in A_k}x_{kj}\(n\)\exp{\(\beta_k\(Q_{kj}\(n+1\)-Q_{kj}\(n\)\)\)}}\,,
\end{align*}
for each $i\in A_k$. In practice, agents perform a large number of actions (updates of Q-values) for each choice-distribution (update of Boltzmann selection probabilities). This motivates to consider a continuous time version of the learning process for each agent $k\in \N$ which results in the following update rules for both the memories $Q_{ki}$ 
\[\dot Q_{ki}=\alpha_k\lt r_{ki}\(x\)-Q_{ki} \rt, \]
where as in the main text, $r_{ki}\(x\)$ denotes $k$'s reward for selecting pure action $i\in A_k$ at state $x=\(x_k,x_{-k}\)\in X$, and the selection probabilities $x_{ki}$ of each action $i\in A_k$
\[\dot x_{ki}=\beta_kx_{ki}\(\dot Q_{ki}-\sum_{j\in A_k}\dot Q_{kj}\).\]
Combining the last two equations under the assumptions that the relationship between pairs of actions is constant over time and that the choice distributions of the various agents are independently distributed, yields the dynamics in equation \eqref{eq:main}, namely
\begin{align*}
\frac{\dot x_{ki}}{\x}&=\beta_k\(r_{ki}\(x\)-\sum_{j\in A_k}x_{kj}r_{kj}\(x\)\)-\alpha_k\(\ln{\x}-\sum_{j\in A_k}x_{kj}\ln{x_{kj}}\)
\end{align*}
for each $i\in A_k$ and for each agent $k\in \N$.

\paragraph{Time scale of updates.} In the above dynamics, there are three relevant time scales or rates from the perspective of agent $k\in \N$: (1) the rate of change of the environment, i.e., of $x_{-k}\in X_{-k}$, (2) the rate of adaptation, i.e., the rate of change of $x_k\in X_k$, which is captured by $\beta_k$ and (3) the rate of memory dissipation or exploration of the action space which is captured by $\alpha_k$. Typically, cf. \cite{Sat05}, the rate of change of the environment --- captured by the choices $x_{-k}\in Y$ of all other agents --- is slower than the changes in the choices $x_k$ of agent $k$ --- updates of equation \eqref{eq:update_x} --- which in turn, are slower than the rate of interaction --- memory updates or equivalently updates of equation \eqref{eq:update_q} --- with the environment. In other words, the agent fixes a choice distribution $x_k\in X_k$ and interacts multiple times with the other agents (environment) before updating their choice distribution according to the above scheme.\par
To determine the interplay between updates of Q-values and choice distributions, let $n$ and $n+1$ denote the time points of two successive updates of the choice distribution $x_k\(n\)$ and $x_k\(n+1\)$ of agent $k\in\N$. For any choice $i\in A_k$, let $n_i:=x_{ki}\(n\)M$ denote the (on average or expectation) number of times that agent $k$ selects action $i\in A_k$ given choice distribution $x_k\(n\)$, where $M\gg0$ is the large number of interactions of the agent with its environment under the fixed choice distribution $x_k\(n\)$. Then, using the index $t\in[0,n_i]$ such that $n=0<1<\dots <t <\dots<n_i=n+1$ to denote the timing of the interactions (the actual timing is irrelevant; only the number of interactions matters), equation \eqref{eq:update_q} yields the following recursion.

\begin{lemma}\label{lem:recursion}
Let $x_{-k}\(n\)\in X_{-k}$ and $x_k\(n\)\in X_k$ denote the choice distributions of agent $k\in\N$ and all other agents in $\N$ other than $k$ at time point $n\ge0$. Let $n+1$ denote the time point of the next update of the choice distribution of agent $k$. Then, if $n=0<1<\dots <t <\dots<n_i=n+1$ denote the time points of $M\gg0$ interactions of agent $k$ with its environment between time points $n$ and $n+1$, the updates in the $Q$-values of agent $k$ are governed in expectation by the equation
\begin{align}\label{eq:q}
Q_{ki}\(n+1\)=\;&\(1-\alpha_k\)^{n_i}Q_{ki}\(n\)+\sum_{t=0}^{n_i}\(1-\alpha_k\)^tr_i\(n_i-t\),
\end{align}
with $n_i:=M\cdot x_{ki}\(n\)$. If $x_{-k}\(n\)$ remains constant between two successive updates of the choice distribution of agent $k$ at time points $n$ and $n+1$, then $r_i\(n_i-t\)=r_i$ for any $t\in [0,n_i]$ and equation \eqref{eq:q} simplifies to 
\begin{align}\label{eq:updates}
Q_{ki}\(n+1\)=\;&\(1-\alpha_k\)^{n_i}Q_{ki}\(n\)+\frac{r_i}{\alpha_k}\(1-\(1-\alpha_k\)^{n_i+1}\).
\end{align}
\end{lemma}

\begin{proof} Assuming that agent $k$ interacts $M$ times with their environment for each (fixed) choice distribution $x_k\(n\)$, where $M$ is a large positive integer, the expected number of times that agent $k$ selects choice $i\in A_k$ is equal to $n_i=M\cdot x_{ki}\(n\)$. Hence, equation \eqref{eq:update_q} yields 
\begin{align*}
Q_{ki}\(n+1\)=\;&r_i\(n_i\)+\(1-\alpha_k\)Q_{ki}\(n_i\)\\
=\;&r_i\(n_i\)+\(1-\alpha_k\)\lt r_i\(n_i-1\)+\(1-\alpha\)Q_{ki}\(n_i-1\)\rt \\
=\;&r_i\(n_i\)+\(1-\alpha_k\)r_i\(n_i-1\)+\(1-\alpha\)^2\lt r_i\(n_i-2\)+Q_{ki}\(n_i-2\)\rt \\
=\; &\dots\\[-0.15cm]
=\; &\(1-\alpha_k\)^{n_i}Q_{ki}\(n\)+\sum_{t=0}^{n_i}\(1-\alpha_k\)^tr_i\(n_i-t\)
\end{align*}
with $t$ indexing the time points at which agent $k$ interacts with their environment between the two successive updates of its choice distribution $x_k\(n\)$ and $x_{k}\(n+1\)$. Finally, assuming that the choice distribution of all other agents remains constant between time points $n$ and $n+1$, it holds that $r_i\(n_i-t\)=r_i$ for any $t\in [0,n_i]$ and equation simpilfies to
\begin{align*}
Q_{ki}\(n+1\)=&\(1-\alpha_k\)^{n_i}Q_{ki}\(n\)\frac{r_i}{\alpha_k}\(1-\(1-\alpha_k\)^{n_i+1}\)
\end{align*}
for $\alpha_k\in[0,1)$, by the formula of the partial sums of the geometric series, \[\sum_{t=0}^{n_i}\(1-\alpha_k\)^t=\frac{1}{\alpha_k}\(1-\(1-\alpha_k\)^{n_i+1}\).\qedhere\]
\end{proof}

\begin{remark}\label{rem:recursion}
Equations \eqref{eq:updates} and \eqref{eq:q} can now be used to compute the memory updates of agent $k$ between two successive updates of agent $k$'s choice distribution via equation \eqref{eq:update_x}. Note that both equations hold in expectation (or on average) since the acting agent chooses their choice according to the choice distribution $x_k\(n\)$ each of the $M$ times that they interact with their environment. However, under the working assumption that $M\gg0$, i.e., that $M$ is a very large number, the law of large numbers suggests that the expected value $n_i=M\cdot x_{ki}\(n\)$ is close to the actual value that agent $k$ uses choice $i$ during the time that agent $k$ draws from the (fixed) choice distribution $x_k\(n\)$. This approach has been used in all experiments resulting in a fast (scalable) implementation of the Q-learning process.\par
Finally, equation \eqref{eq:q} in Lemma~\ref{lem:recursion} suggests that for the extreme values of $\alpha\in[0,1)$, the updates of the $Q$-values are 
\[Q_i\(n+1\)=\begin{dcases}r_i\(n_i\), & \text{for } \alpha\to1,\\ Q_i\(n\)+\sum_{t=0}^{n_i} r_i\(t\), & \text{for } \alpha=0.\end{dcases}\]
This recovers the intuition that when $\alpha=0$, the agent has \emph{perfect memory}, whereas for $\alpha\to1$, the agent is completely oblivious of past rewards.
\end{remark}

\section{Omitted Proofs of Section~\ref{sec:results}}
\label{app:proofs}

\begin{proof}[Proof of Lemma~\ref{lem:modified}]
Using the defintion of $r^H_{ki}\(x\)$, we have that 
\begin{align*}
r^H_{ki}\(x\)-\ip{x_k}{r^H_k\(x\)}=\;&\beta_k\(r_{ki}\(x\)-\ip{x_k}{r_k\(x\)}\)-\alpha_k\(\ln{\x}+1-\ip{x_k}{\ln{x_k}}-\ip{x_k}{1}\). \end{align*}
Since $\ip{x_k}{1}=\sum_{j\in A_k}x_{kj}=1$, the right side reduces to the expression in \eqref{eq:qre_ip}.
\end{proof}

To compare the performance of two different choice distributions $p,x\in X$, we will use the notion of \emph{KL-divergence}, $D\(p\,\|\,x\)$, which is a measure of the distance from $p$ to $x$ defined by $D\(p\,\|\,x\):=-\sum_{i=1}^n p_i\ln\(\frac{x_i}{p_i}\)$. To prove Theorem~\ref{thm:regret}, we will need the following important property that follows immediately from the definition of the KL-divergence.
\begin{lemma}\label{lem:klpositive}
Let $p=\(p_1,p_2,\dots,p_n\)\in X$ be fixed and let $x=\(x_1,x_2,\dots,x_n\)\in X$ denote an arbitrary probability distributions (mixed strategy or population state) in $X$. Also let $\ip{\cdot}{\cdot}$ denote the inner product in $\mathbb R^n$. Then, 
\[\ip{p}{\ln{p}}\ge\ip{p}{\ln{x}}, \quad \text{for any } x\in X,\]
with equality if and only if $x=p$. 
\end{lemma}
\begin{proof}
By linearity of the inner product in $\mathbb R^n$, and non-negativity of the KL-divergence (which follows from Gibbs inequality and the fact that $\ln{x}\le x-1$ for all $x>0$), we have that
\begin{align*}
\ip{p}{\ln{p}-\ln{x}}&=\sum_{i=1}^np_i\(\ln p_i-\ln{x_i}\)=-\sum_{i=1}^n p_i\ln\(\frac{x_i}{p_i}\)=D\(p\|x\)\ge0
\end{align*}
with equality if and only if $x=p$. 
\end{proof}

\begin{proof}[Proof of Theorem~\ref{thm:regret}]
Consider an agent $k\in \N$ and let $p_k\in X_k$ denote the agent's optimal strategy in hindsight, i.e., $p_k=\argmax_{x'_k\in X_k}\int_0^Tu^H_k\(x'_k;x_{-k}\(t\)\)dt$. Let also $x_{k}\(t\)$ denote the sequence of play generated by the dynamics in \eqref{eq:replicator} for an arbitrary initial condition $x_k\(0\)\in X_k$. Then, by taking the time derivative of the term $\sum_{i\in A_k}p_{ki}\ln{\(x_{ki}\(t\)\)}$, we obtain
\begin{align*}
\frac{d}{dt} \ip{p_k}{\ln{x_k\(t\)}}&=\sum_{i\in A_k}p_{ki}\cdot\frac{\dot x_{ki}}{\x}=\sum_{i\in A_k}p_{ki}\lt r^H_{ki}\(x\)-\ip{x_k}{r^H_k\(x\)}\rt\\
& =\sum_{i\in A_k}p_{ik}\lt\beta_k\(r_{ki}-\ip{x_k}{r_k\(x\)}\)-\alpha_k\(\ln{\x}-\ip{x_k}{\ln{x_k}}\)\rt\\
& =\beta_k\ip{p_k}{r_k\(x\)}-\alpha_k\ip{p_k}{\ln{x_k}}-\ip{x_k}{\beta_kr_k\(x\)-\alpha_k\ln{x_k}}\\
& \ge \beta_k\ip{p_k}{r_k\(x\)}-\alpha_k\ip{p_k}{\ln{p_k}}-u_k^H\(x_k;x_{-k}\)\\
& =u_k^H\(p_k;x_{-k}\(t\)\)-u_k^H\(x_k\(t\);x_{-k}\(t\)\)
\end{align*}
where the inequality has been established in Lemma~\ref{lem:klpositive}. Integrating both sides of the previous inequality from timepoint $0$ to $T>0$, and using the definition of $R_k\(T\)$ in equation \eqref{eq:regret} we get
\begin{align*}
&\sum_{i\in A_k}p_{ki}\(\ln{x_{ki}\(T\)}-\ln{x_{ki}\(0\)}\)\ge\int_{0}^T u_k^H\(p_k;x_{-k}\(t\)\)-u_k^H\(x_k\(t\);x_{-k}\(t\)\) dt = R^H_k\(T\).
\end{align*}
Since $x_{ki}\(T\)\in [0,1]$ for all $i=1,2,\dots,n$, and since $\sum_{i\in A_k}p_{ki}=1$, the left side is upper bounded by $-\sum_{i \in A_k}p_{ki}\ln{\(x_{ki}\(0\)\)}$ which is a constant with respect to $T$. Hence, 
\[\limsup_{T\to\infty}R^H_k\(T\)\le -\sum_{i\in A_k}p_{ki}\ln{x_{ki}\(0\)},\]
which concludes the proof.
\end{proof}

\begin{proof}[Proof of Lemma~\ref{lem:potential}]
Before proceeding to the proof of the statement of Lemma~\ref{lem:potential}, observe that the multiplicative constants $\beta_k, k\in\N$ in equation \eqref{eq:main} are essentially equivalent to a rescaling of the payoffs of agent $k$ in the underlying game. Thus, in a weighted potential game $\Gamma$ with vector of positive weights $w=\(w_k\)_{k\in\N}$ satisfying 
\[u_k\(i,a_{-k}\)-u_k\(j,a_{-k}\)=w_k\(\phi\(i,a_{-k}\)-\phi\(j,a_{-k}\)\),\]
for all $i\neq j\in A_k,a_{-k}\in A_{-k}$, for all agents $k\in \N$, one may rescale the parameters $\beta_k$ to $\beta'_k=\beta_k/w_k$ for all $k\in \N$ and consider the resulting \emph{exact} potential game instead. This implies that we can adjust the techniques of \cite{Kle09,Cou15} to prove the statement of Lemma~\ref{lem:potential}. \par
In particular, to see that $\Phi^H\(x\)$ defines a potential for $\Gamma^H$, consider the partial derivatives of $\Phi^H\(x\)$ at a point $x=\(x_{ki}\)_{k\in N,i\in A_k}\in X$,
\begin{align*}
\frac{\partial}{\partial \x}\Phi^H\(x\)&=\frac{\partial}{\partial \x}\Phi\(x\)-\frac{\alpha_k}{\beta_k}\(\ln{\x}+1\)\\&=\frac{\partial}{\partial \x}u_k\(x\)-\frac{\alpha_k}{\beta_k}\(\ln{\x}+1\)\\&= r_{ki}\(x\)-\alpha_k\(\ln{\x}+1\)=\frac{1}{\beta_k}r^H_{ki}\(x\)
\end{align*}
where $\frac{\partial}{\partial \x}u_k\(x\)=r_{ki}\(x\)$ by definition and $\frac{\partial}{\partial \x}\Phi\(x\)=\frac{\partial}{\partial \x}u_k\(x\)$ since $\Phi\(x\)$ is a potential function for the unmodified game with utilities $u_k\(x\), k\in \N$. Hence, $\Gamma^H$ is a weighted potential game with potential function $\Phi^H\(x\)$ for $x \in X$ and vector of weights $\beta=\(\beta_k\)_{k\in \N}$. Given the above, taking the time derivative of the potential $\Phi^H\(x\)$ yields
\begin{align*}
\dot \Phi^H\(x\)& =\sum_{k\in \N}\sum_{i\in A_k}\(\frac{\partial}{\partial x_{ki}}\Phi^H\(x\)\)\dot x_{ki}\\
&=\sum_{k\in \N}\frac{1}{\beta_k}\sum_{i\in A_k}r^H_{ki}\(x\)\dot x_{ki}\\
&=\sum_{k\in \N}\frac1\beta_k\sum_{i\in A_k}r^H_{ki}\(x\) \x \lt r^H_{ki}\(x\)-\ip{x_k}{r^H_k\(x\)}\rt\\
&=\sum_{k\in \N}\frac1\beta_k\lt\sum_{i\in A_k}\x \(r^H_{ki}\(x\)\)^2-\(\sum_{i\in A_k}\x r^H_{ki}\(x\)\)^2\rt \ge 0,
\end{align*}
where the last inequality follows directly from the Cauchy-Schwartz inequality (equivalently by observing that the term in braces is the variance of the quantities $r_{ki}^H\(x\), i\in A_k$ under the distribution $x_k\in X_k$). Accordingly, equality holds if the dynamics are at a fixed point which concludes the proof.
\end{proof}

\begin{proof}[Proof of Theorem~\ref{thm:potential}]
Given Theorem~\ref{lem:potential}, and the fact that the fixed points of \eqref{eq:replicator} are precisely the QRE of $\Gamma$, it remains to show that any sequence of play $\(x_k\(t\)\)_{k\in\N}, t\ge0$ converges to a compact, connected set consisting entirey of equilibria, i.e., of points $x\in X$ for which $\Phi^H\(x\)$ is constant and equal to $0$. \par
To see this, observe that for any sequence of play $\(x\(t\)\)_{t\ge0}\subseteq X$, the limit set $\Omega$ is defined as 
\[\Omega=\bigcap_{s\in \mathbb R_+}\text{cl}\{x\(t\):t>s\}\]
where $\text{cl}\{S\}$ denotes the closure of set $S$. Hence, $\Omega$ is compact and connected as the decreasing intersection of compact, connected sets. Moreover, since $\Phi^H\(x\(t\)\)$ is increasing by Lemma~\ref{lem:potential} and bounded on $X$ by definition, it follows that $\Phi^H\(x\(t\)\)$ converges to a value $\Phi^*=\sup{\Phi^H\(x\(t\)\)}$ for any sequence of play $x\(t\)_{t\ge0}\subseteq X$. By continuity of $\Phi^H$, this implies that $\Phi^*=\Phi^H\(x^*\)$ for all $x^*\in \Omega$. Hence, if $x^*\(t\)_{t\ge0}\subseteq \Omega$, it must be that $\Phi^H\(x^*\)=\Phi^*$ for all $t\ge0$. Since $\Phi^H\(x\(t\)\)$ is strictly increasing in $t$ unless $x\(t\)$ is a sequence equilibria, it follows that $\Omega$ consists entirely of equilibria of the game $\Gamma^H$ which are QRE in $\Gamma$. This concludes the proof. 
\end{proof}

\section{Omitted Materials: Section~\ref{sec:performance}}
\label{app:mechanism}

The following provides a detailed notation for the coordination games discussed in Section~\ref{sec:applications}. Some content is repetitive.

\subsection{Coordination games}
\label{sub:coordination}
Consider a 2-player $\N=\{1,2\}$, 2-action $A_1=A_2=\{a_1,a_2\}$ game $\Gamma=\(\N,\(A_k,u_k\)_{k\in\N}\)$ with payoff matrices \vspace*{-0.5cm}
\begin{subequations}
\begin{align}\label{eq:coordination_1}
u_1=\bordermatrix{
~ & a_1 & a_2 \cr
a_1 & u_{11} & u_{12} \cr
a_2 & u_{21} & u_{22} \cr},\;u_2=\bordermatrix{
~ & a_1 & a_2\cr
a_1 & v_{11} & v_{21} \cr
a_2 & v_{12} & v_{22} \cr},
\intertext{where $u_{ij}$ ($v_{ij}$) denotes the payoff of agent 1 (2) when this agent uses action $a_i$ and the other agent uses action $a_j$ for $i,j=1,2$. $\Gamma$ is a \emph{coordination game}, if}
\label{eq:coordination_2}
u_{11}>u_{21}, u_{22}>u_{12} \text{ and } v_{11}>v_{21}, v_{22}>v_{12}
\end{align}\label{eq:coordination}
\end{subequations}hold for the payoffs of agents 1 and 2, respectively. In this case, the game admits three NE that can be described in terms of the probabilities $x,y\in[0,1]$ of agents 1 and 2 using action $a_1$: two pure NE $\(x,y\)=\(1,1\)$ and $\(x,y\)=\(0,0\)$ that correspond to the pure action profiles $\(a_1,a_1\)$ and $\(a_2,a_2\)$ and one fully mixed at $\(\xm,\ym\)=\(\frac{\lambda_2}{k_2},\frac{\lambda_1}{k_1}\)$ where $\lambda_1:=u_{22}-u_{12}, k_1:={u_{11}-u_{12}-u_{21}+u_{22}}$ and $\lambda_2:=v_{22}-v_{12}, k_2:=v_{11}-v_{12}-v_{21}+v_{22}$, with $\lambda_i,k_i>0$ for $i=1,2$. The equilibrium $\(a_2,a_2\)$ is called \emph{risk-dominant} if $\(u_{22}-u_{12}\)\(v_{22}-v_{12}\)>\(u_{11}-u_{21}\)\(v_{11}-v_{21}\)$. In symmetric games, i.e., when $u_2=u_1^T$, this condition simplifies to $u_{22}+u_{21}>u_{11}+u_{12}$. If the inequality is reversed, then $\(a_1,a_1\)$ is risk dominant and if equality holds, then none of the pure equilibria risk-dominates the other. Finally, if $u_{22}\ge u_{11}$ and $v_{22}\ge v_{11}$ with at least one inequality strict then $\(a_2,a_2\)$ is called \emph{payoff} or \emph{Pareto-dominant}. For such games, we have the following properties. 

\begin{lemma}[Properties of $2\times2$ coordination games]
\label{lem:risk_dominant}
Let $\Gamma=\(\N,\(A_k,u_k\)_{k\in \N}\)$ denote a two-player, $\N=\{1,2\}$, two-action, $A_1=A_2=\{a_1,a_2\}$, coordination game with payoff functions $\(u_1,u_2\)$ as in equations \eqref{eq:coordination}. Then, $\Gamma$ is a weighted potential game with weights $\(1,w\)$ for some $w>0$ and it holds that
\begin{enumerate}[label=(\roman*)\,\,, wide=0pt, leftmargin=\parindent, labelsep=0pt]
\item $\xm+\ym>1$ if and only if $\(0,0\)$, i.e., the pure action profile $\(a_2,a_2\)$, is the risk-dominant equilibrium.
\end{enumerate}
If, in addition, $\Gamma$ is symmetric, i.e., $u_2=u_1^T$, then $\xm=\ym$ and property (i) simplifies to
\begin{enumerate}[label=(i*)\,\,, wide=0pt, leftmargin=\parindent, labelsep=0pt]
\item $\xm>1/2$ if and only if $\(0,0\)$ is the risk-dominant equilibrium. 
\end{enumerate}
In this case, the (exact) potential is globally maximized at the pure action profile $\(a_2,a_2\)$, i.e., at the risk-dominant equilibrium.
\end{lemma}

\begin{proof}[Proof of Lemma~\ref{lem:risk_dominant}] Since $k_1,k_2>0$, there exists $w>0$ such that $k_1=wk_2$. Then, it is immediate to check that 
\[P=\begin{pmatrix}u_{11}-u_{21}& u_{11}-u_{21}-w\(v_{11}-v_{21}\)\\ 0 & k_1-w\(v_{11}-v_{21}\)\end{pmatrix}\]
is a potential function for $\Gamma$ with weights $\(w_1,w_2\)=\(1,w\)$. Hence, $\Gamma$ is a $\(1,w\)$ weighted potential game with potential function $P$. For (i), using the coordination game assumption, i.e., that $u_{11}>u_{21}, u_{22}>u_{12}$ and $v_{11}>v_{21}, v_{22}>v_{12}$ and dividing both sides of inequality \eqref{eq:riskdominant} with the product $k_1k_2$, we have that \eqref{eq:riskdominant} is equivalent to 
\begin{align*}
&\frac{\lambda_1}{k_1}\cdot\frac{\lambda_2}{k_2}>\frac{k_1-\lambda_1}{k_1}\cdot\frac{k_2-\lambda_2}{k_2}\iff\ym\xm>\(1-\ym\)\(1-\xm\)\iff\xm+\ym>1
\end{align*}
as claimed. Finally, if $\Gamma$ is symmetric, i.e., if $u_2^T=u_1$, then $k_1=k_2$ and $\lambda_1=\lambda_2$ which implies that $\xm=\ym$. In this case, inequality \eqref{eq:riskdominant} yields the condition $\xm>1/2$ which proves (*i). To see that the global maximizer of the potential agrees with the risk dominant equilibrium, observe that $P=\begin{pmatrix}u_{11}-u_{21} & 0\\ 0 & u_{22}-u_{12}\end{pmatrix}$ is a potential function for the game. Hence, the global maximum of $P$ is at $\(a_2,a_2\)$ whenever $u_{22}-u_{12}>u_{11}-u_{21}$, i.e., whenever $\(\lambda_1/k_1\)>1-\(\lambda_1/k_1\)$ or equivalently whenever $\xm>1/2$. Since any potential function must satisfy $P+c$ for some constant $c\in \mathbb R$ (see \cite{Mon96}), this concludes the proof.
\end{proof}

As a special case, it is immediate from Lemma~\ref{lem:risk_dominant} that a necessary and sufficient condition for $\Gamma$ to be an exact potential game is that $k_1=k_2$. 


Using Lemma~\ref{lem:risk_dominant}, we can reason about the possible locations of QRE in $2\times2$ coordination games. This is established in Theorem~\ref{thm:location_qre} where we focus on the case $\xm+\ym\ge1$ (wlog). 
To prove Theorem~\ref{thm:location_qre}, we will use that for an arbitrary 2-player, 2-action game $\Gamma=\(\N,\(A_k,u_k\)_{k\in\N}\)$, the coupled dynamic equations in \eqref{eq:main} become
\begin{subequations}
\label{eq:qre_small}
\begin{align}
\frac{\dot x}{x\(1-x\)}=\;&\beta_x\lt y\(u_{11}-u_{21}\)+\(1-y\)\(u_{12}-u_{22}\)\rt+\alpha_x\ln{\(\frac1x-1\)}\\
\frac{\dot y}{y\(1-y\)}=\;&\beta_y\lt x\(v_{11}-v_{21}\)+\(1-x\)\(v_{12}-v_{22}\)\rt+\alpha_y\ln{\(\frac1y-1\)}
\end{align}
\end{subequations}
where $x,y\in[0,1]$ denote the probabilities that agent $1$ and $2$ respectively assign to pure action $a_1$. Using Lemma~\ref{lem:risk_dominant} and the expressions in equations \eqref{eq:qre_small}, we can now prove Theorem~\ref{thm:location_qre}.

\begin{proof}[Proof of Theorem~\ref{thm:location_qre}]
At any QRE $\(x_Q,y_Q\)\in\(0,1\)\times\(0,1\)$, the right sides of the coupled equations in \eqref{eq:qre_small} are simultaneously equal to $0$. Using the introduced notation $\lambda_1:=u_{22}-u_{12}, k_1:={u_{11}-u_{12}-u_{21}+u_{22}}$ and $\lambda_2:=v_{22}-v_{12}, k_2:=v_{11}-v_{12}-v_{21}+v_{22}$, with $\lambda_i,k_i>0$ for $i=1,2$, we can rewrite the QRE conditions as 
\begin{subequations}
\label{eq:tools}
\begin{align}
\label{eq:tools_x}c_1\( y_Q-\ym\)+\ln{\(\frac1{x_Q}-1\)}=0\\
\label{eq:tools_y}c_2\( x_Q-\xm\)+\ln{\(\frac1{y_Q}-1\)}=0
\end{align}
\end{subequations}
where, $c_1:=k_1\cdot\beta_x/\alpha_x$ and $c_2:=k_2\cdot\beta_y/\alpha_y$ are positive constants (with respect to $x,y$) by assumption. As above $\(\xm,\ym\)$ denote the probabilities of pure action $a_1$ at the fully mixed equilibrium for agent 1 and 2, respectively. The cases in the statement of Theorem~\ref{thm:location_qre} now follow from an exhaustive sign analysis of the terms $\ln{\(\frac1{x_Q}-1\)}$ and $\ln{\(\frac1{y_Q}-1\)}$ which are positive for $x_Q,y_Q>1/2$ and negative otherwise. Specifically, let $\xm+\ym>1$ (the case $\xm+\ym<1$ is analogous and the case $\xm+\ym=1$ corresponds to coordination games at which none of the pure equilibria is risk dominant and is treated in \cite{Kia12}). Then, we have following cases
\begin{description}
\item[Case 1:] $\xm,\ym>1/2$. If $x_Q>1/2$, then $\ln\(\frac{1}{x_Q}-1\)<0$ which implies that $y_Q>\ym$. In turn, since $\ym>1/2$, this implies in particular, that $y_Q>1/2$ and hence, that $\ln\(\frac{1}{y_Q}-1\)<0$ which imposes the condition $x_Q>\xm$ for equation \eqref{eq:tools_y} to be feasible. Hence, if $x_Q>1/2$, then a necessary condition for QRE is that $x_Q>\xm$ and $y_Q>\ym$. This establishes the upper right region in the upper panel of Figure~\ref{fig:location}. If $x_Q\le 1/2$, then $\ln\(\frac{1}{x_Q}-1\)\ge0$ and hence, $y_Q<\ym$. However, since $\xm>1/2$ by assumption, $x_Q\le 1/2$ implies that $x_Q-\xm<0$ which imposes the restriction $y_Q<1/2$ for equation \eqref{eq:tools_y} to hold. This yields the feasible region $x_Q,y_Q\le 1/2$ depicted in the bottom left shaded region in the upper panel of Figure~\ref{fig:location}. 
\item[Case 2:] $\xm>1/2,\ym \le 1/2$. Proceeding as in Case 1, assume first $x_Q>1/2$. Then, $y_Q$ must satisfy $y_Q>\ym$ for equation \eqref{eq:tools_x} to hold. However, if $y_Q>1/2$, then $\ln\(\frac{1}{y_Q}-1\)<0$ which imposes the restriction $x_Q<\xm$ for \eqref{eq:tools_y} to be infeasible. Hence, for $x_Q>1/2$, $y_Q$ must satisfy either $\ym<y_Q<1/2$ when $1/2<x_Q<\xm$ or $y_Q>1/2$ when $x_Q>\xm$. This yields the middle and upper right regions in the bottom panel of Figure~\ref{fig:location}. Finally, when $x_Q<1/2$, it holds that $\ln\(\frac{1}{x_Q}-1\)>0$ and hence, $y_Q<\ym$ for equation \eqref{eq:tools_x} to hold. In turn, the condition $y_Q<\ym$ does not impose further restrictions for equation \eqref{eq:tools_y} which results in the bottom left region in the bottom panel of Figure~\ref{fig:location}. 
\item[Case 3:] $\xm\le 1/2,\ym>1/2$. This case is symmetric to Case 2. \qedhere
\end{description}
\end{proof}
The above cases are illustrated in Figure~\ref{fig:location} in the main body of the paper. The case $\xm,\ym>1/2$ corresponds to a situation at which the interests are better aligned than in the case $\xm>1/2,\ym<1/2$. In particular, symmetric games are special instances of the situation depicted in the upper panel of Figure~\ref{fig:location}. In this case, $\xm=\ym$ and there are no QRE in the segment $[1/2,\xm]$. 

\begin{remark}[Intuition of Theorem~\ref{thm:location_qre}]
\label{rem:location}
The main intuition of Theorem~\ref{thm:location_qre} is that the QRE surface may or may not be connected depending on the alignment of the interests of the two agents, i.e., on whether $\xm>1/2,\ym \ge1/2$ (upper panel) or $\xm\ge 1/2, \ym<1/2$ (bottom panel) of Figure~\ref{fig:location}. In turn, the connectedness of the QRE surface crucially affects the convergence of the learning process.\\ Intuitively, if one agent increases their exploration rate, then the choice distribution of that agent at a QRE will shift towards the uniform distribution, i.e., $\(1/2,1/2\)$ in a $2\times2$ game. In case that the payoff parameters of the underlying game are such that the game is described by the upper panel, then the two agents will find themselves in the bottom left (shaded) square which lies entirely in the attracting region of the risk-dominant equilibrium, $\(x,y\)=\(0,0\)$, at the bottom left corner. Hence, after reducing the exploration rate to $0$, the Q-learning dynamics will rest at this equilibrium. Importantly, this outcome is independent of the starting point \emph{and} the exploration-exploitation profile of the other agent. \\
By contrast, this is not always the case if the payoff parameters of the underlying game are such that the game is described by the bottom panel of Figure~\ref{fig:location}. In this case, the QRE surface is connected and the effect on the learning process of the exploration policy --- i.e., of temporarily increasing the exploration rate before reducing it again back to its initial level --- depends on the exploration-exploitation profile of the other agents, and importantly also on their synchronicity. The critical observation is that the middle (shaded) rectangle in the bottom panel of Figure~\ref{fig:location} is transcended by the $x+y=\xm+\ym$ line which is the threshold between the attracting regions of the two corner equilibria, $\(x,y\)=\(0,0\)$ and $\(x,y\)=\(1,1\)$. Hence, when one or both agents increase their exploration rates to reach the middle rectangle, they cannot say in advance in which attracting region the learning process will find itself before they start reducing their exploration levels back to their initial zero levels. In this case, the process may well converge to the equilibrium where it started from (see, e.g., Figure~\ref{fig:nopay} and \ref{fig:grid}).
\end{remark}

\begin{proof}[Proof of Theorem~\ref{thm:catastrophe}]
The proof is constructive and utilizes Theorem~\ref{thm:location_qre}. For $M>0$, let $\Gamma^M_u=\{\N=\{1,2\},\(\{a_1,a_2\}\)_{k\in\N},\(u_1,u_2\)\}$ with $u_1,u_2$ given by $u_1=\bordermatrix{
~ & a_1 & a_2 \cr
a_1 & 2M & 0\cr
a_2 & 2M-1 & 2 \cr},\quad u_2=u_1^T
$. Then, (i) $\xm=\ym=2/3$ for any $M>0$ and (ii) $\(a_2,a_2\)$ is risk-dominant since $2M-1+2>2M+0$. Condition (i) implies that the basin of attraction of the $\(a_1,a_1\)$ equilibrium is equal to $I_u=[2/3,1]^2$ and hence that it has strictly positive measure (and in fact constant for any $M>0$). By Theorem~\ref{thm:location_qre}, (ii) implies that the QRE surface is disconnected --- upper panel of Figure~\ref{fig:location}. Hence, given a starting point in the interior of $I_u$, there exist exploration thresholds $\delta_k$ that depend on $\(x_k\(0\)\)$, for $k=1,2$, so that if the exploration rates remain throughout low, i.e., $\delta_k\(t\)<\delta_k$ for all $t>0$ and for both $k=1,2$, then the dynamics will never escape the basin of attraction of $\(a_1,a_1\)$. Hence, since $\lim_{t\to\infty}\delta_k\(t\)=0$ by assumption --- i.e., at the end of the learning process, agents stop to explore the space --- it holds that $\lim_{t\to\infty} u_k^{\text{exploit}}\(t\)=2M$ for both agents, i.e., their choice distributions will approximate (to arbitrary precision), the $\(a_1,a_1\)$ NE. \par
By contrast, if $\delta_k\(t\)>\delta_k$ for some agent $k\in \N$, then the coupled dynamics in equation \eqref{eq:main} will reach a fixed point close to the uniform distribution which by Theorem~\ref{thm:location_qre} lies in the basin of attraction of the risk-dominant $\(a_2,a_2\)$ equilibrium. When reducing the exploration rates back to zero, the agents' choice distribution will converge to the $\(a_2,a_2\)$ NE, which implies that $\lim_{t\to\infty}u_k^{\text{explore}}\(t\)=2$ for both agents. Since $M>0$ was arbitrary, this concludes the case of unbounded loss.\par
Note that a specific realization of the game $\Gamma^M_u$ that is described above can be obtained by appropriately tuning the payoffs in the Stag Hunt game as described in Table~\ref{tab:games}.\par
To obtain the other direction, i.e., unbounded gain, consider (in a similar fashion) for $M>0$ the game $\Gamma^M_v=\{\N=\{1,2\},\(\{a_1,a_2\}\)_{k\in\N},\(v_1,v_2\)\}$ with $u_1,u_2$ given by $u_1=\bordermatrix{
~ & a_1 & a_2 \cr
a_1 & 2M & 1.5\cr
a_2 & 2M-1 & 2 \cr},\quad u_2=u_1^T
$. Then, (i) $\xm=\ym=1/3$ for any $M>0$ and (ii) $\(a_1,a_1\)$ is risk-dominant since $2M-1+2<2M+1.5$. Proceeding as in the previous case, condition (i) implies that the basin of attraction of the $\(a_2,a_2\)$ equilibrium is equal to $I_v=[0,1/3]^2$ and hence that it has strictly positive measure (and in fact constant for any $M>0$). By Theorem~\ref{thm:location_qre}, (ii) implies that the QRE surface is disconnected --- upper panel of Figure~\ref{fig:location}. The difference is that now, the payoff dominant equilibrium $\(a_1,a_1\)$ is also the risk dominant equilibrium. Hence, starting by an arbitrary point in the interior of $I_v$ and by a similar argument as above, we obtain that $\lim_{t\to\infty}v_k^{\text{explore}}\(t\)=2M$ and $\lim_{t\to\infty}v_k^{\text{exploit}}\(t\)=2$ for any agent $k\in \N$ which concludes the proof. \end{proof}

\section{Supplementary Experiments}
\label{app:experiments}

To understand the effect of the exploration policy in the collective outcome of the SQL dynamics, we study the three coordination games in Table~\ref{tab:games}. For each game, we consider all possible combinations of the exploration policies CLR-1 and ETE for the two agents. The experiments for the three coordination games (Pareto Coordination, Battle of the Sexes and Stag Hunt) are presented in Figure~\ref{fig:grid}. 

\begin{figure*}[!htb]
\centering
\small \textbf{Panel A}\hspace{140pt} \textbf{Panel B}\hspace{140pt} \small \textbf{Panel C}\\[0.2cm]
\includegraphics[width=0.312\textwidth]{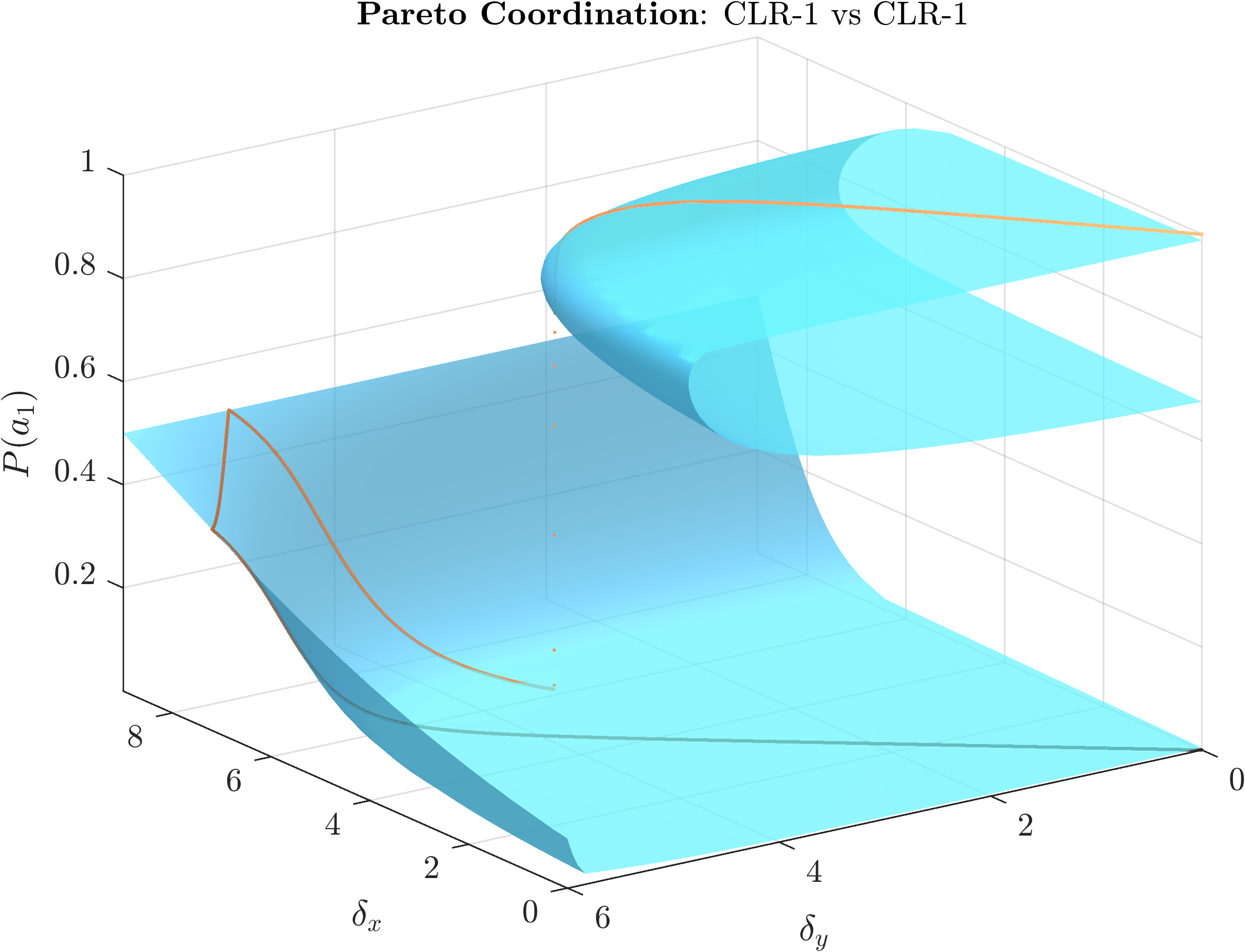}\hspace{0.4cm}
\includegraphics[width=0.312\textwidth]{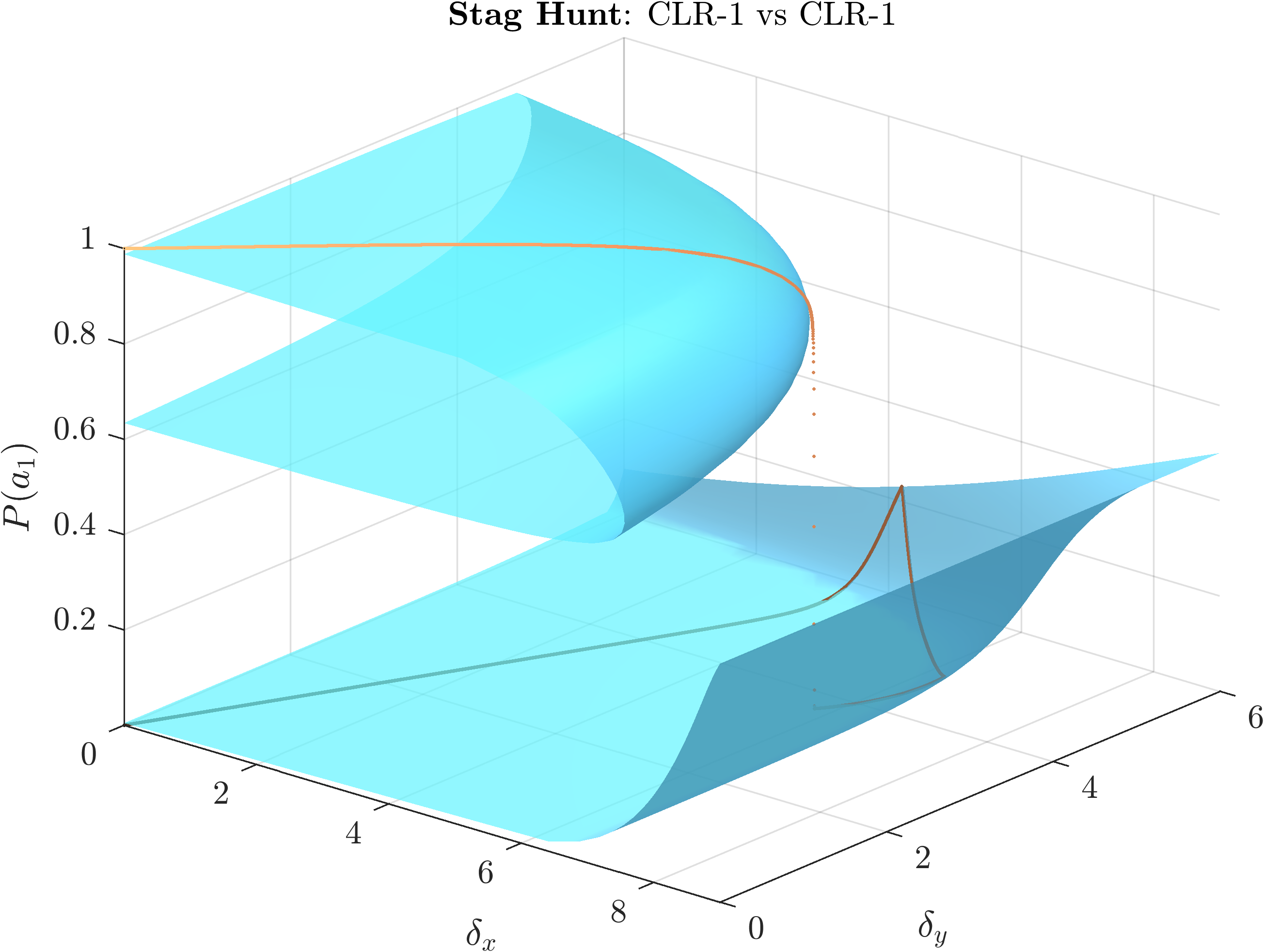}\hspace{0.4cm}
\includegraphics[width=0.312\textwidth]{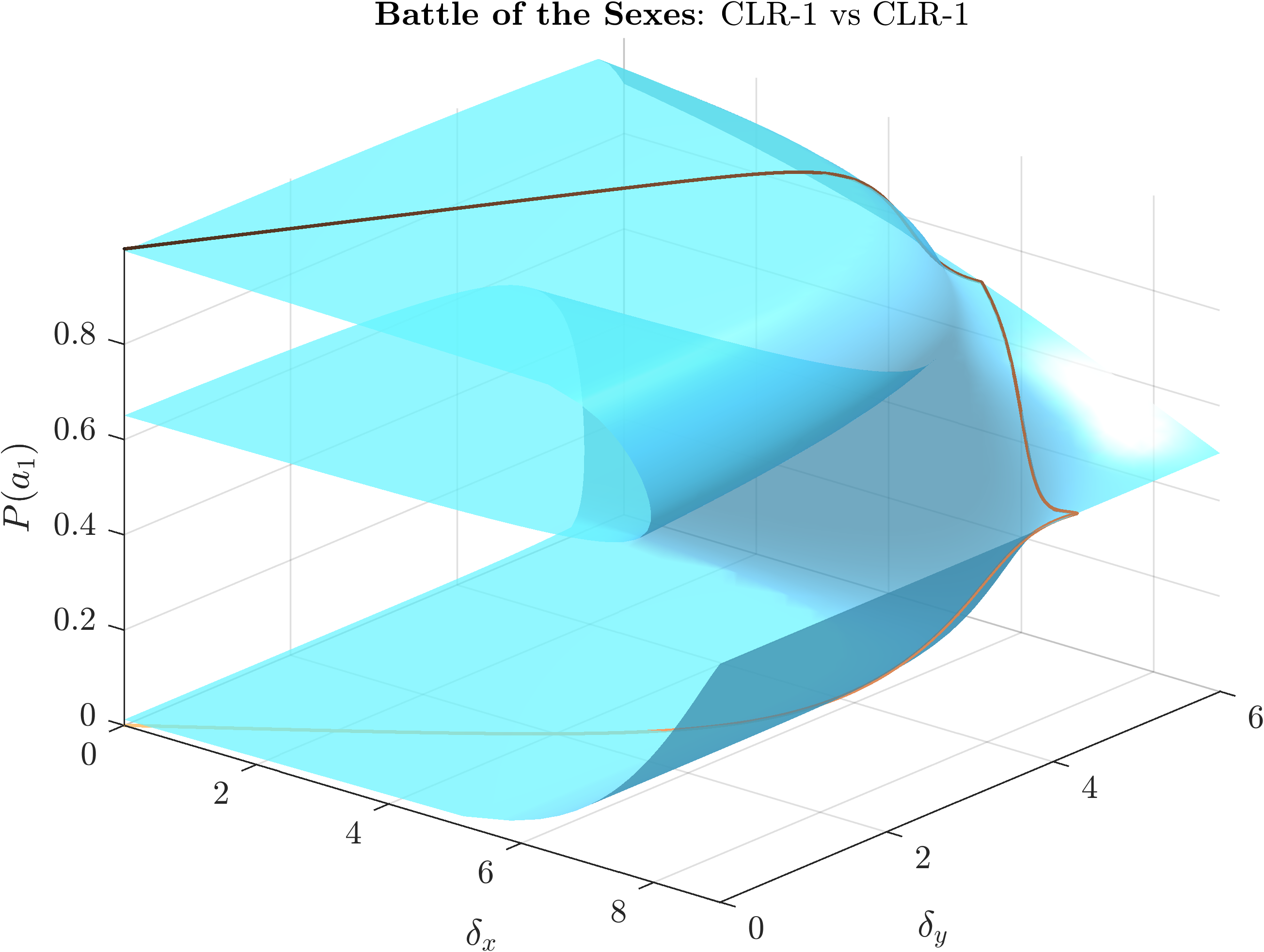}\\[0.2cm]
\includegraphics[width=0.312\textwidth]{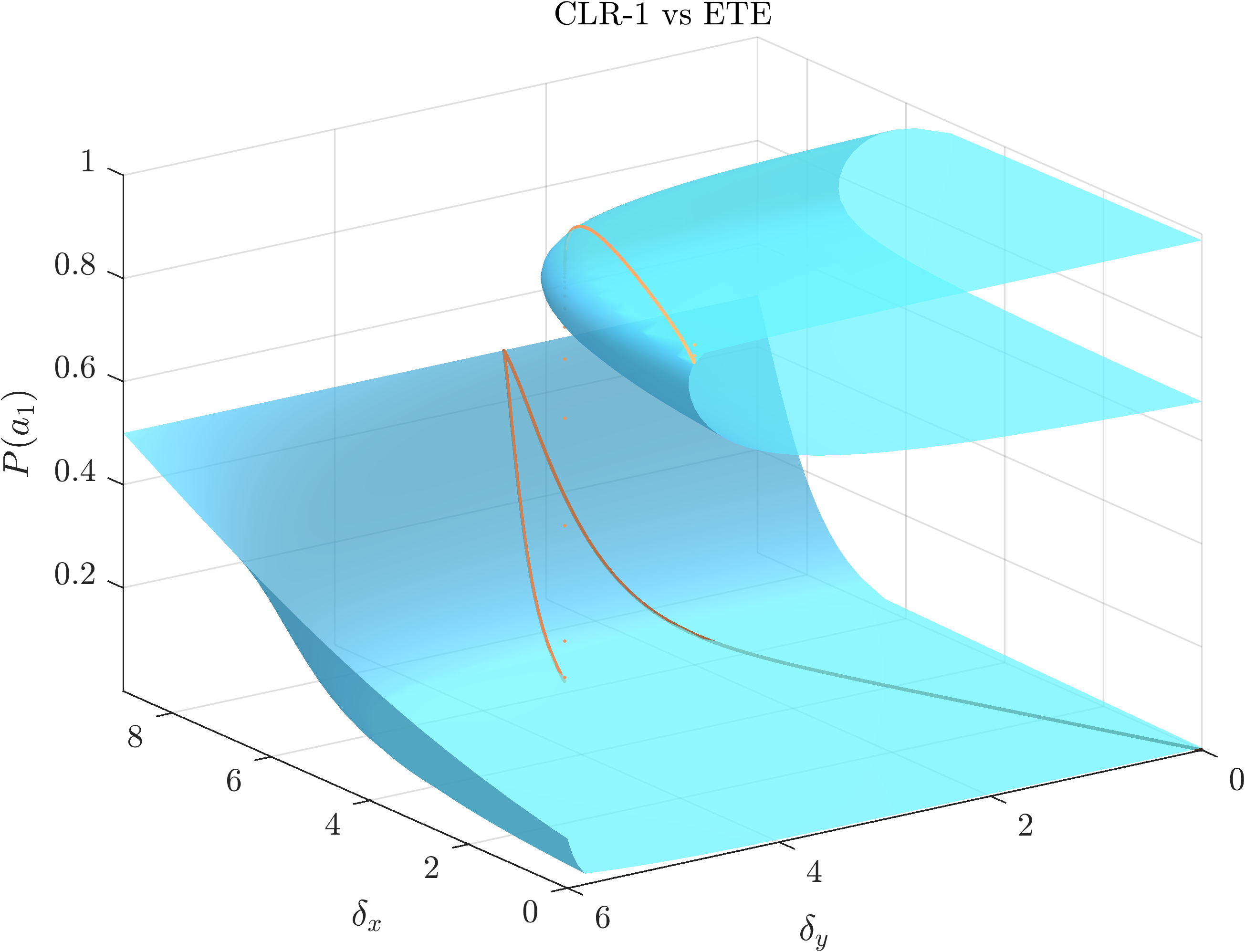}\hspace{0.4cm}
\includegraphics[width=0.312\textwidth]{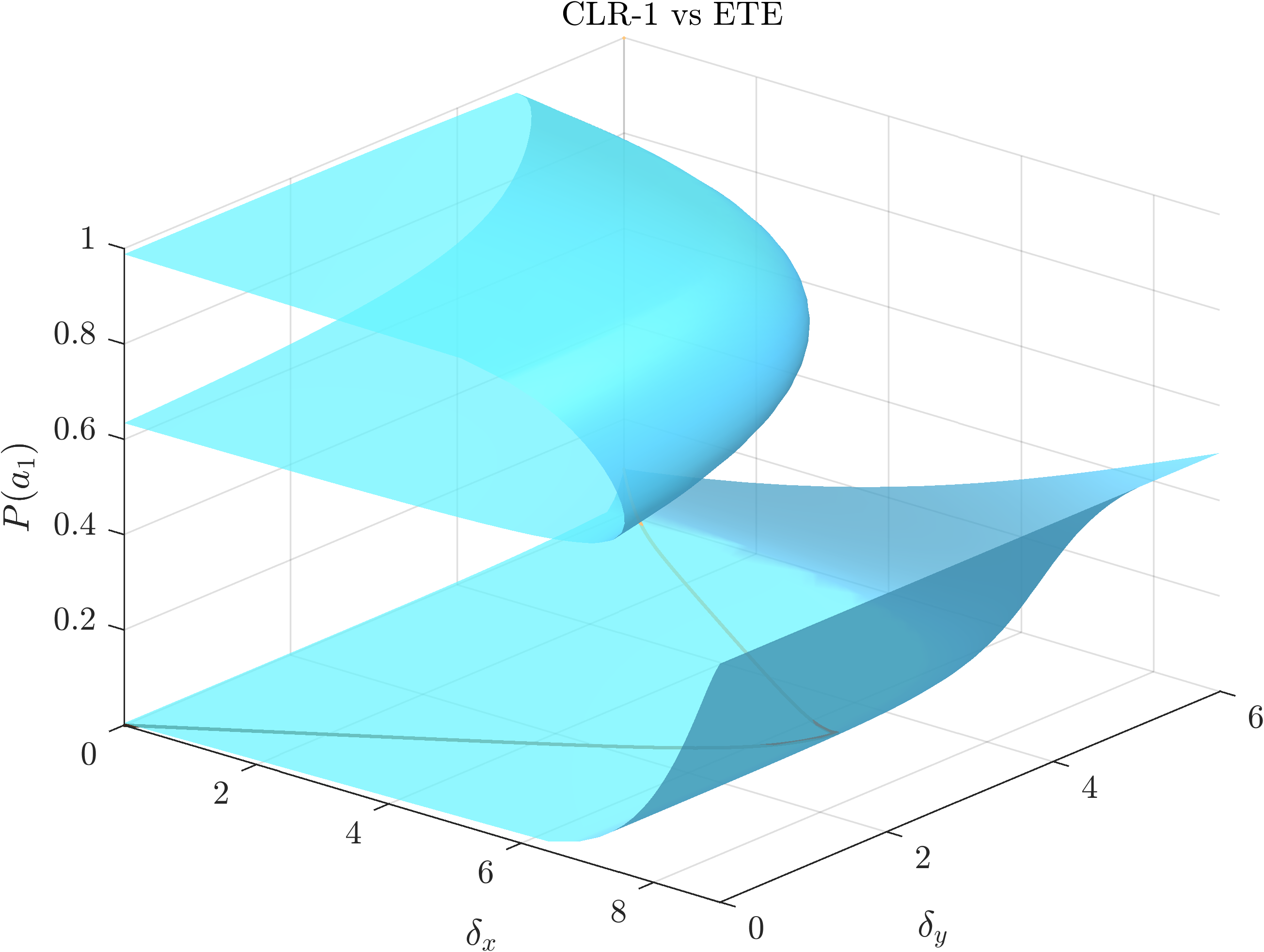}\hspace{0.4cm}
\includegraphics[width=0.312\textwidth]{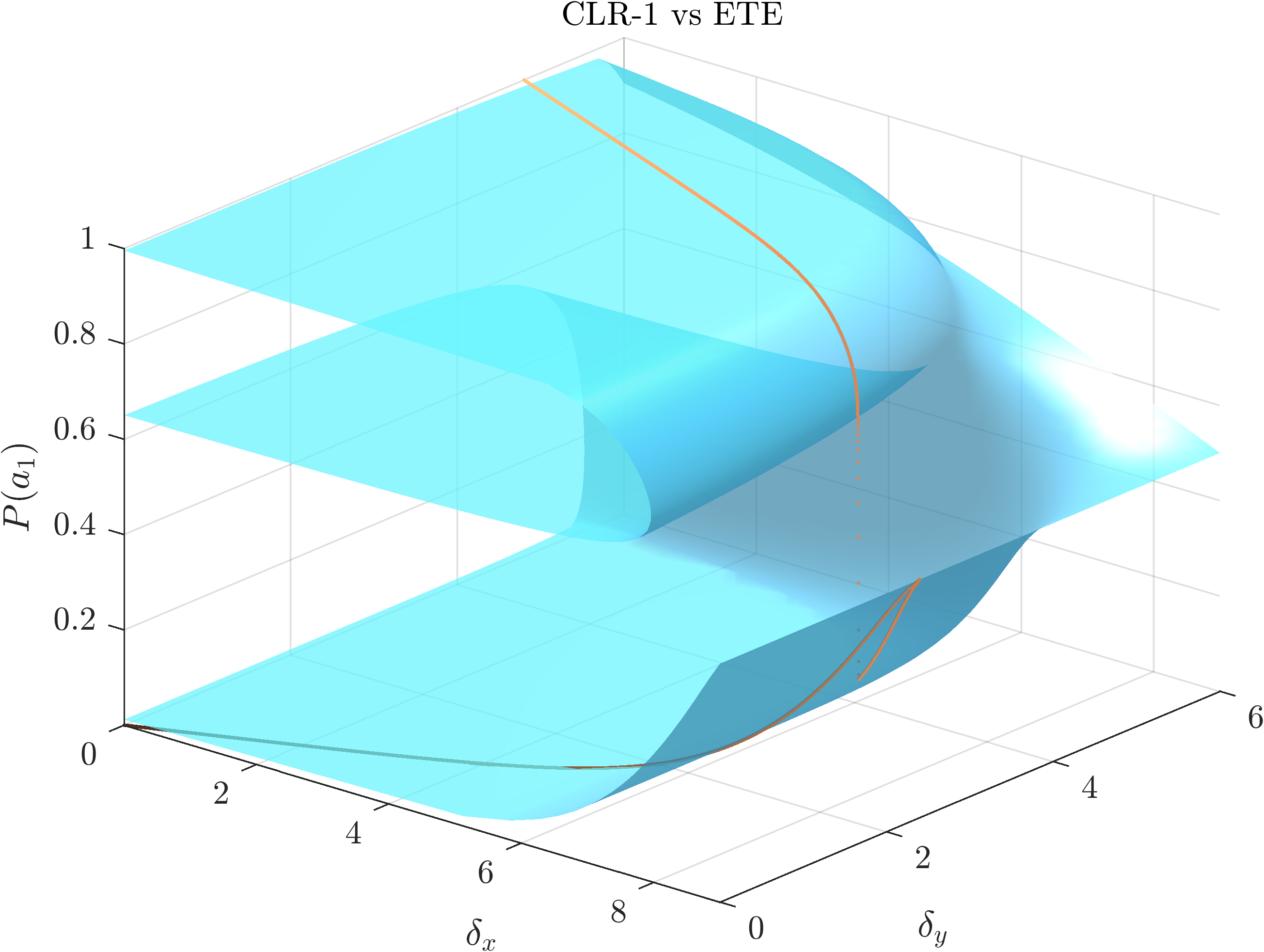}\\[0.2cm]
\includegraphics[width=0.312\textwidth]{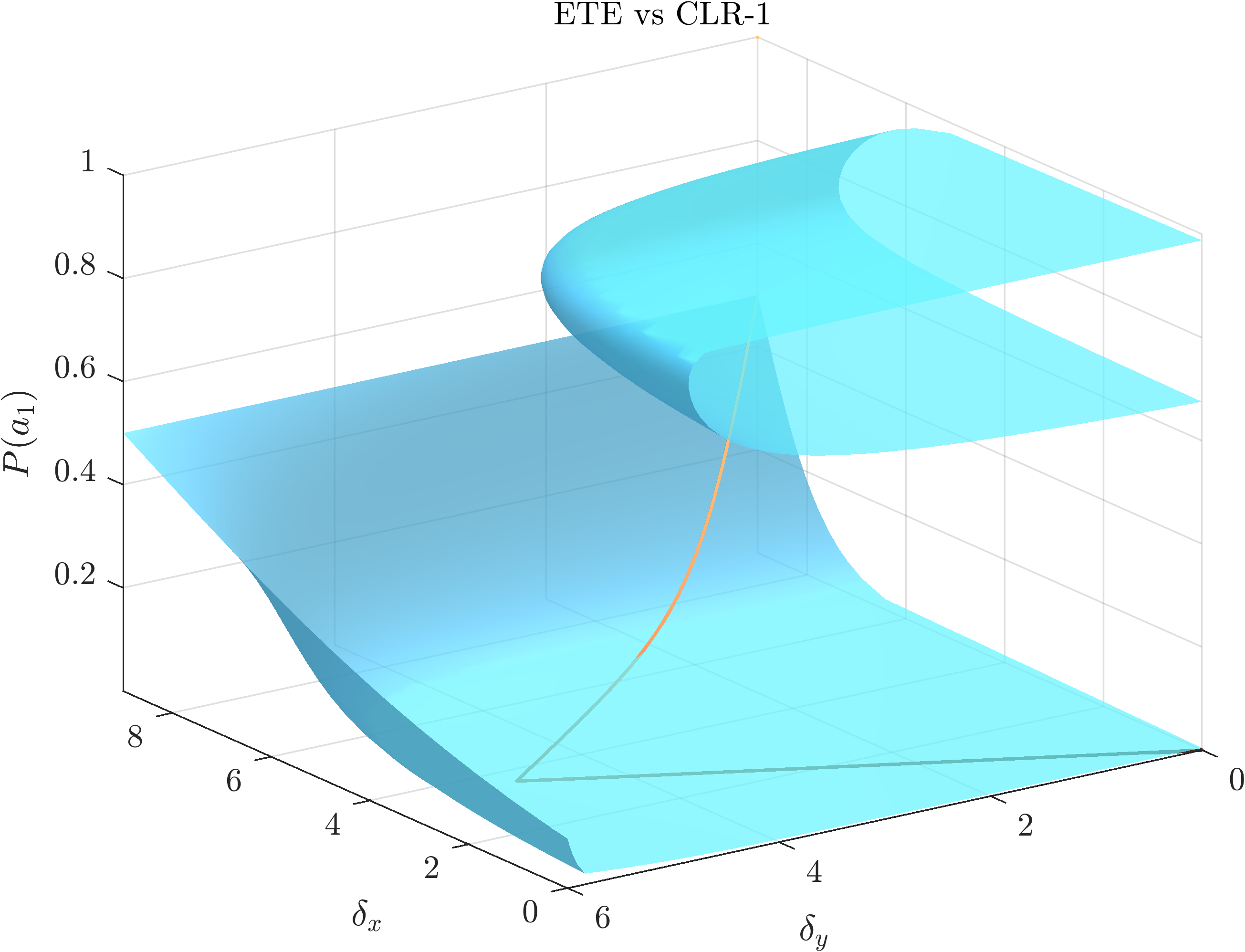}\hspace{0.4cm}
\includegraphics[width=0.312\textwidth]{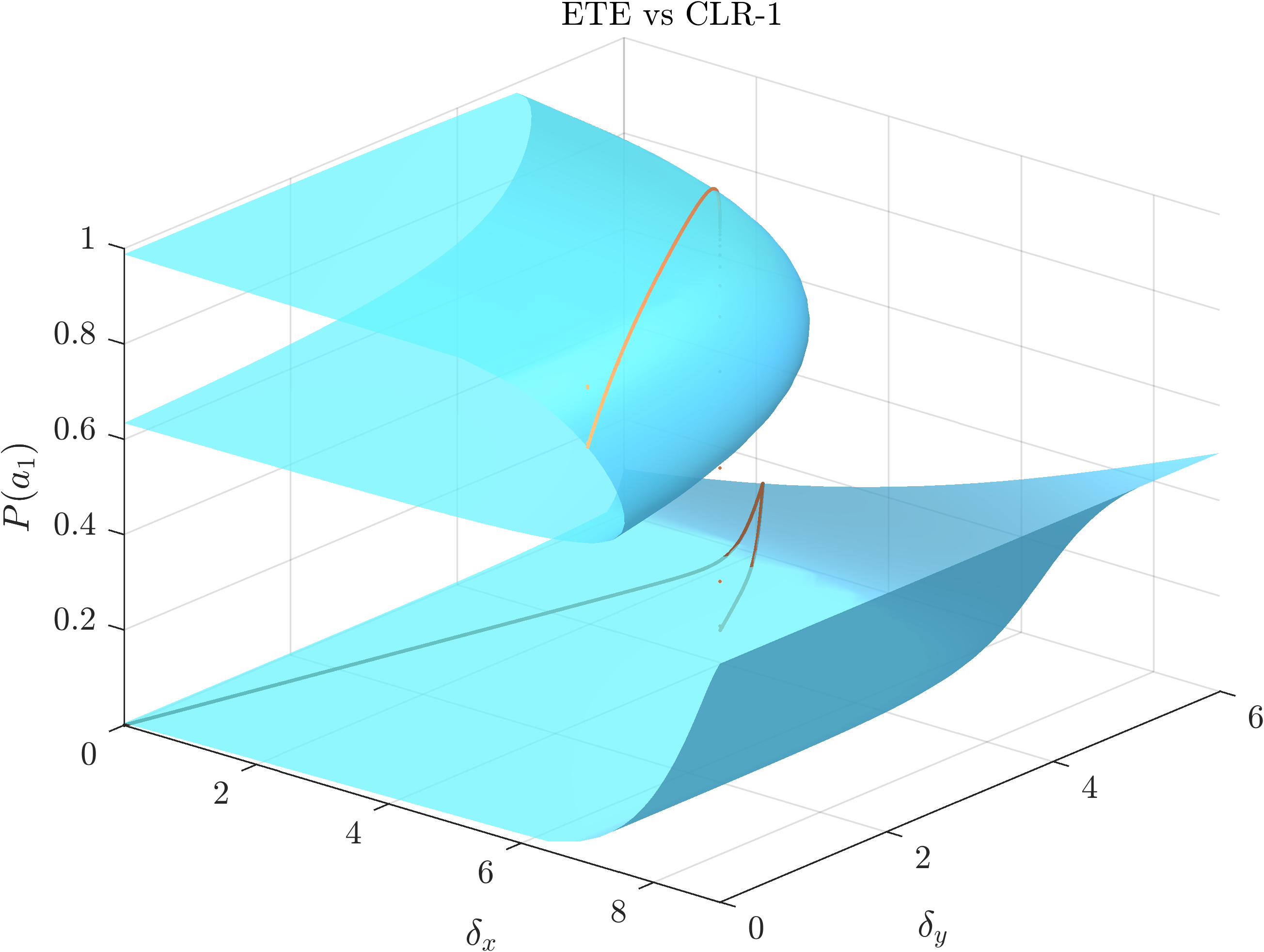}\hspace{0.4cm}
\includegraphics[width=0.312\textwidth]{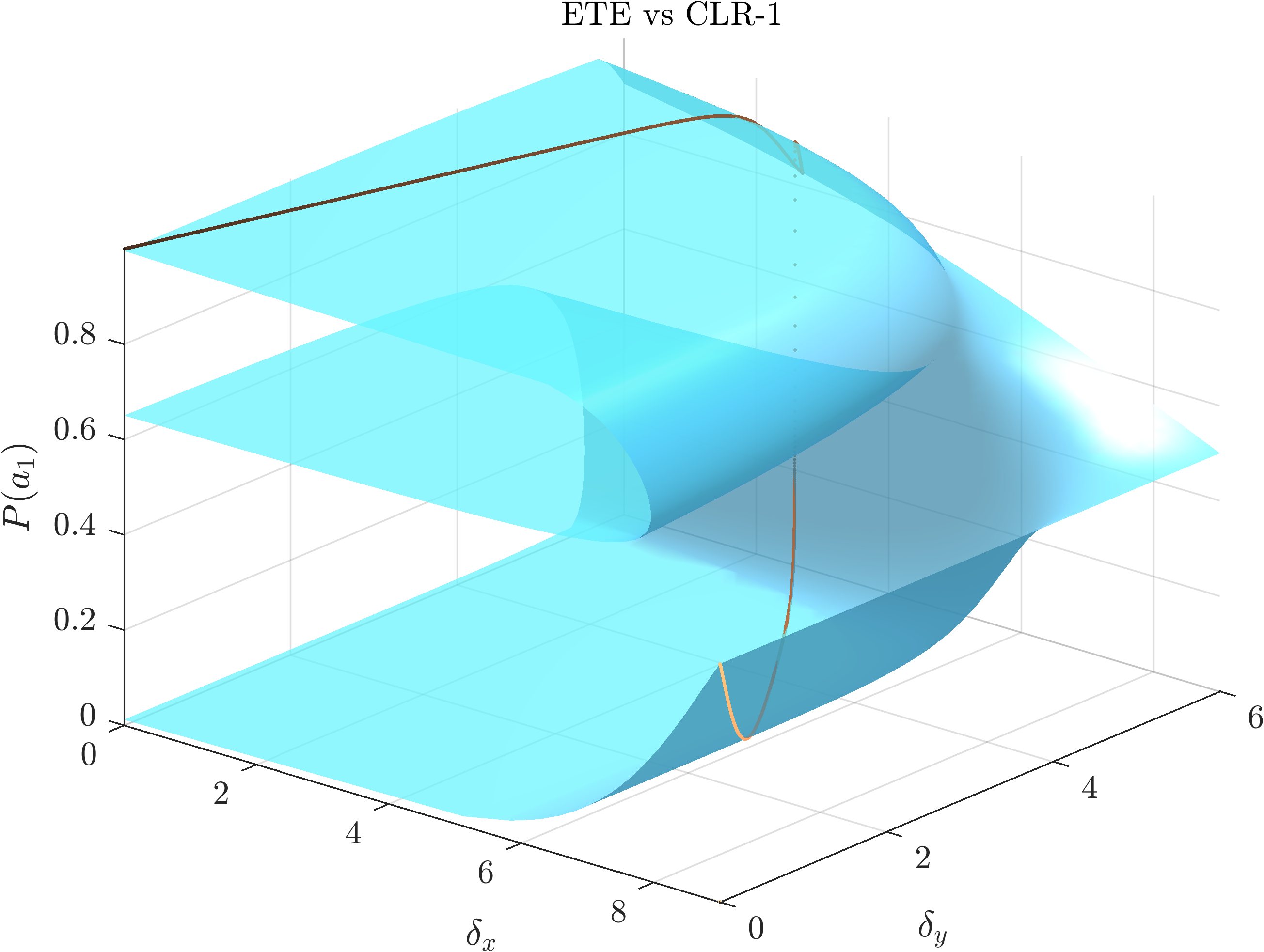}\\[0.2cm]
\includegraphics[width=0.312\linewidth]{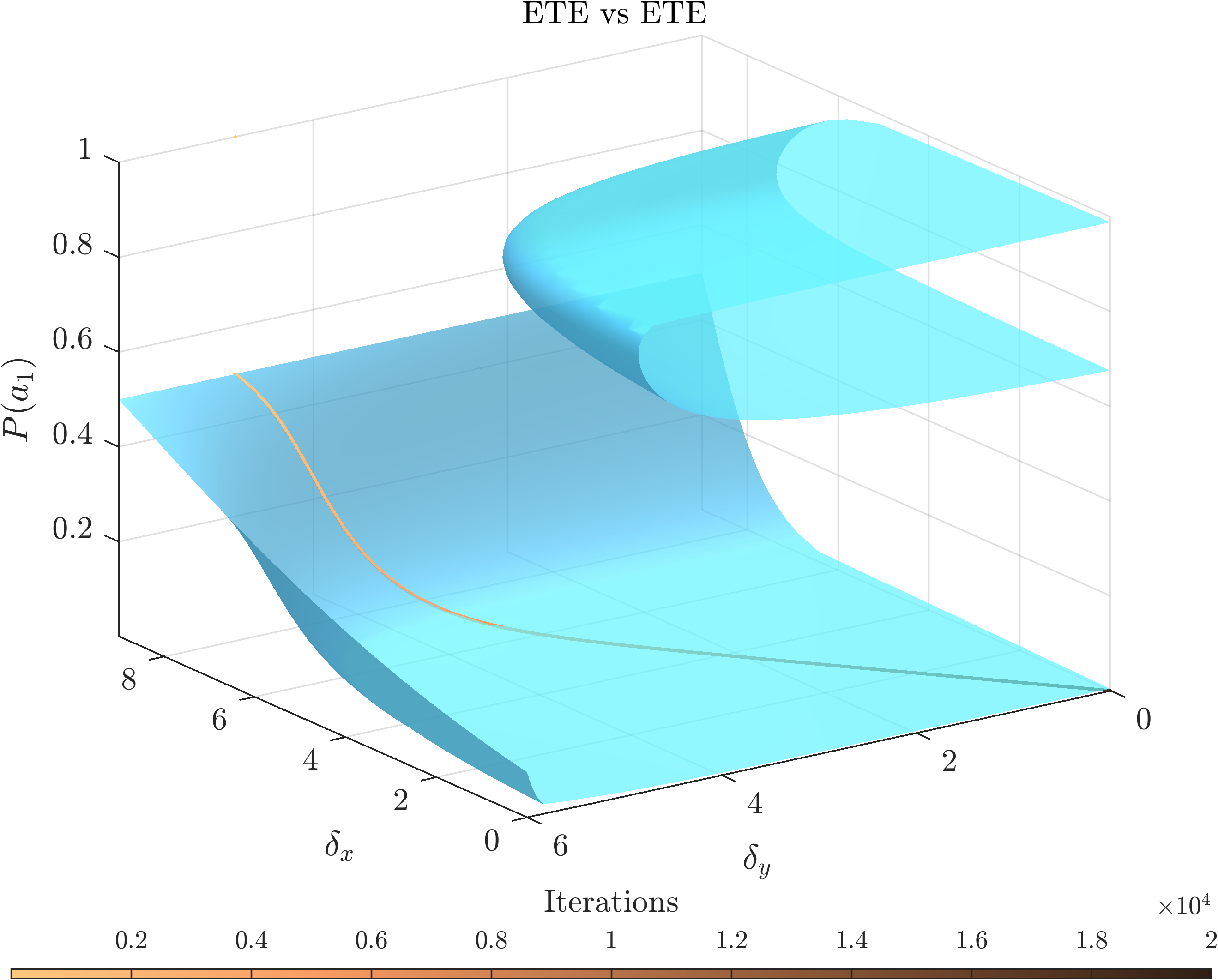}\hspace{0.4cm}
\includegraphics[width=0.312\linewidth]{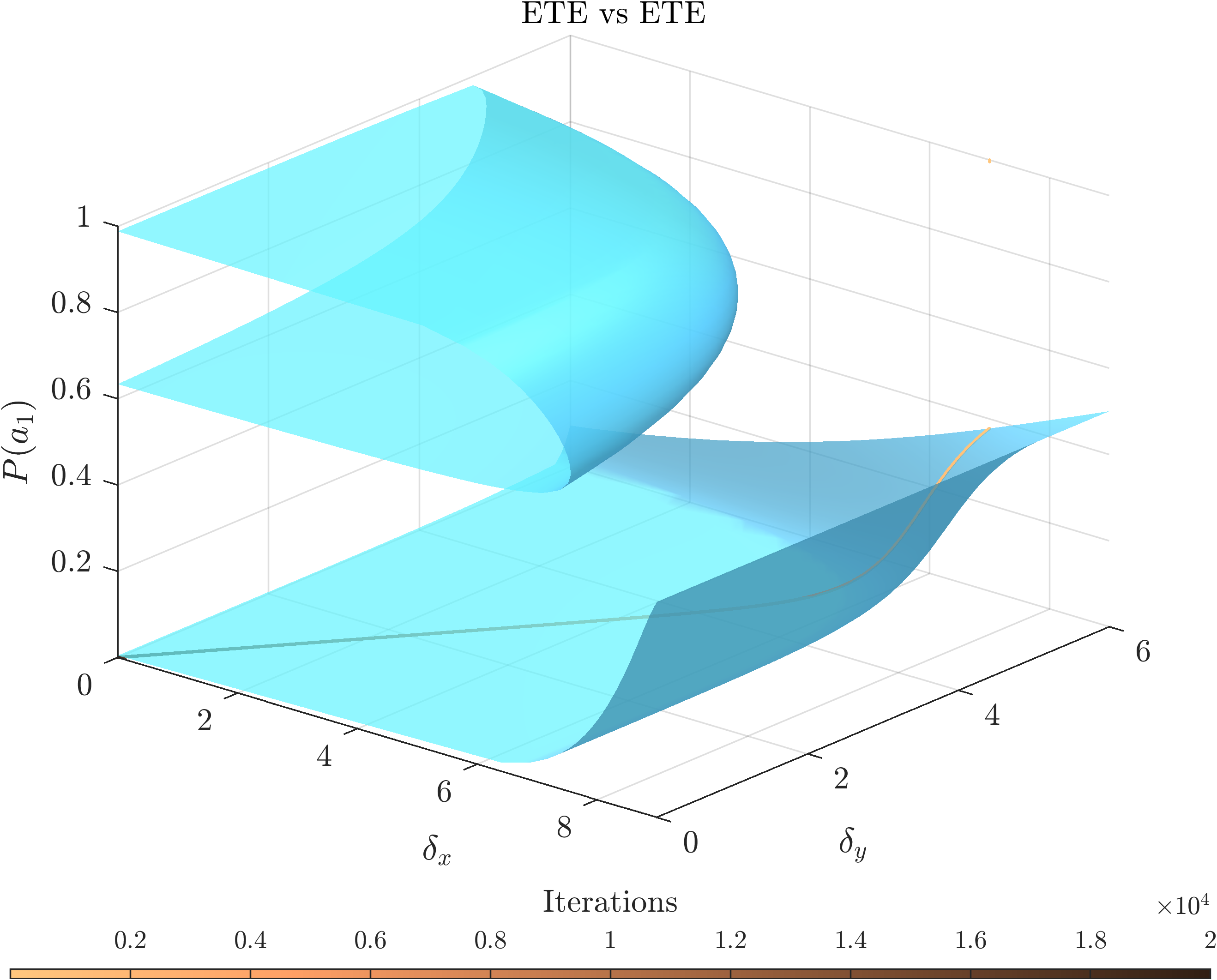}\hspace{0.4cm}
\includegraphics[width=0.312\linewidth]{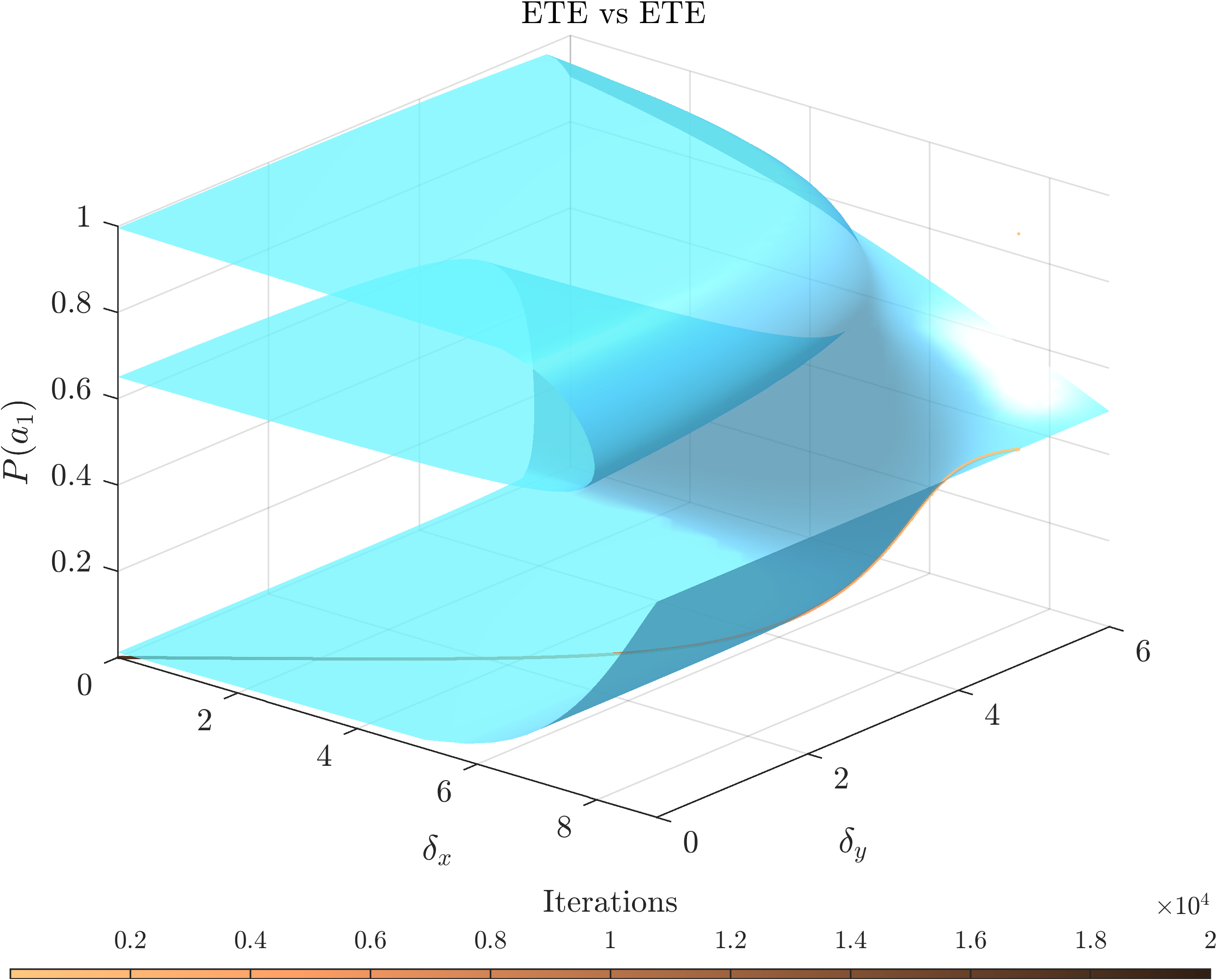}\\[0.2cm]
\caption{SQL dynamics ($1e+03$ Q-value updates for each of $2e+04$ choice distribution updates) with the ETE and CLR-1 policies in coordination games. In Pareto Coordination (Panel A) and Stag Hunt (Panel B), the SQL dynamics converge to the risk dominant equilibrium regardless of the starting point (\emph{saddle-node bifurcations}). By contrast, in the Battle of the Sexes (Panel C), the collective outcome of the exploration process is a priori ambiguous and the SQL dynamics may converge to either of the pure action equilibria (consistent with \emph{cusp bifurcations}). The decisive feature is the geometry of the QRE surface, i.e., whether it is connected or not, as formalized in Theorems~\ref{thm:catastrophe} and \ref{thm:location_qre}.}
\label{fig:grid}
\end{figure*}

Beyond potential games, e.g., in the Spoiled Child game of \cite{Wun10} (or in zero-sum games with unique interior equilibria), the behavior of the dynamics becomes critically dependent on the discretization scheme and is thus, unpredictable. Accordingly, we skip plots of such experiments.

\paragraph{Arbitrary Dimensions}
\label{par:diagonal}
As a warm-up, we study the SQL dynamics in pure coordination games --- coordination games with non-zero payoffs only on the diagonal --- with action spaces of arbitrary size \citep{Kim96}. Depending on the starting point and (especially) on the intensity and timing of the exploration performed by each agent, the SQL dynamics converge after the exploration phase to (typically) improved, yet possibly only locally optimal equilibria. \par
Specifically, we consider games $\Gamma=\(\N,\(A_k,u_k\)_{k\in\N}\)$ with two players $\N=\{1,2\}$ and $n$ actions $A_1=A_2=\{a_1,a_2,\dots,a_n\}$ with diagonal payoff matrices
\begin{equation}\label{eq:diagonal}
u_1=\bordermatrix{
~ & a_1 & a_2 & \dots & a_n\cr
a_1 & u_{11} & 0 & \dots & 0 \cr
a_2 & 0 & u_{22} & \dots & 0\cr
\dots & \dots & \dots & \ddots &\dots \cr
a_n & 0 & 0 & \dots & u_{nn}},
\quad u_2=u_1^T,
\end{equation}
with $0<u_{11}<u_{22}<\dots<u_{nn}$. Any symmetric profile $\(a_i,a_i\), i=1,2,\dots,n$ constitutes a pure NE of $\Gamma$. The equilibrium $\(a_n,a_n\)$ is payoff-dominant and (by properly generalizing the notion of risk-dominance) also risk dominant. In this case, Theorem~\ref{thm:potential} suggests that the Q-learning dynamics converge to a compact connected set of QRE of $\Gamma$. However, it is unclear whether this will be the payoff dominant equilibrium or not and if not, how this outcome depends on the starting point and exploration policy of both agents. Two different instances are given in Figures~\ref{fig:diag_10} and \ref{fig:diag_9}. 
\begin{figure}[!htb]
\centering
\includegraphics[width=0.48\linewidth]{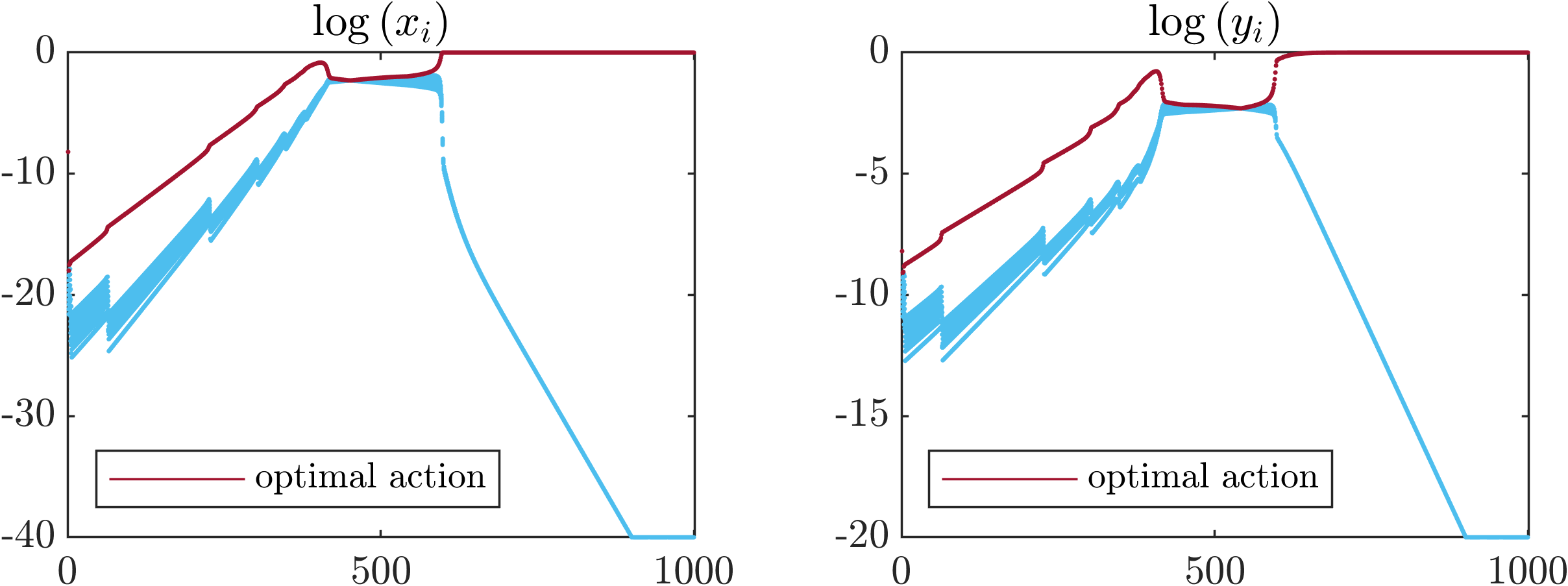}\hspace{10pt}
\includegraphics[width=0.48\linewidth]{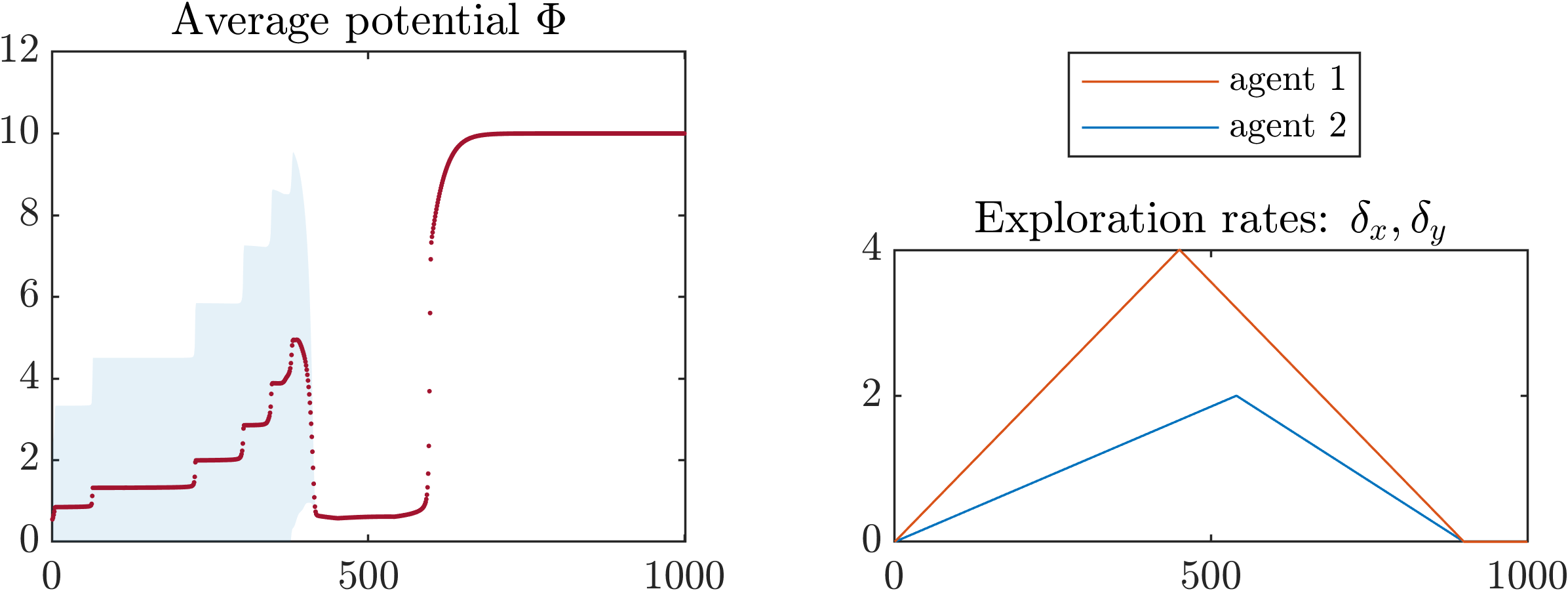}
\caption{SQL dynamics ($1e+20$ Q-value updates for each of $1e+3$ choice distribution updates) in a pure coordination game with $n=10$ actions.  The panels are as in Figure~\ref{fig:potential_10} and again show averages over a $10\times10$ grid of starting points, each of which is close to one pure action profile. Both agents perform CLR-1 exploration with different intensities and all trajectories of the SQL dynamics converge to the global optimum after the exploration phase regardless of the starting point (the standard deviation, depicted as the shaded region in the bottom left panel, vanishes).}
\label{fig:diag_10}
\end{figure}
\begin{figure}[!htb]
\centering
\includegraphics[width=0.48\linewidth]{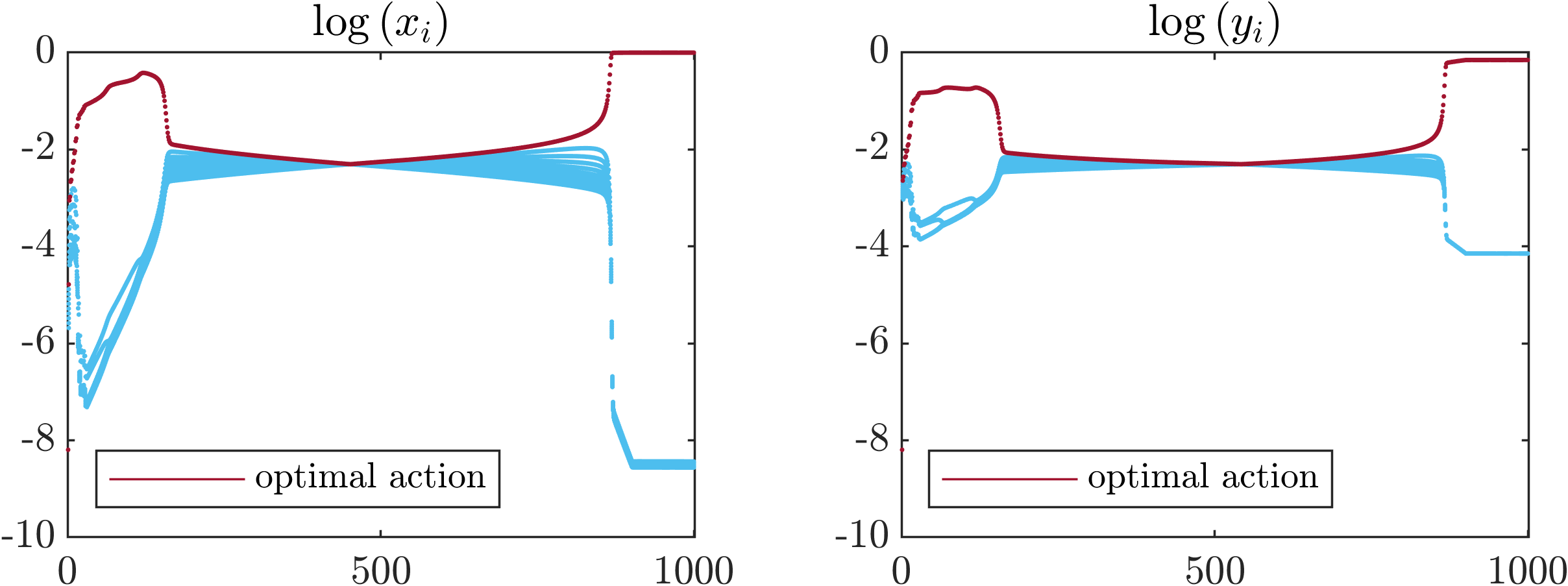}\hspace{10pt}
\includegraphics[width=0.48\linewidth]{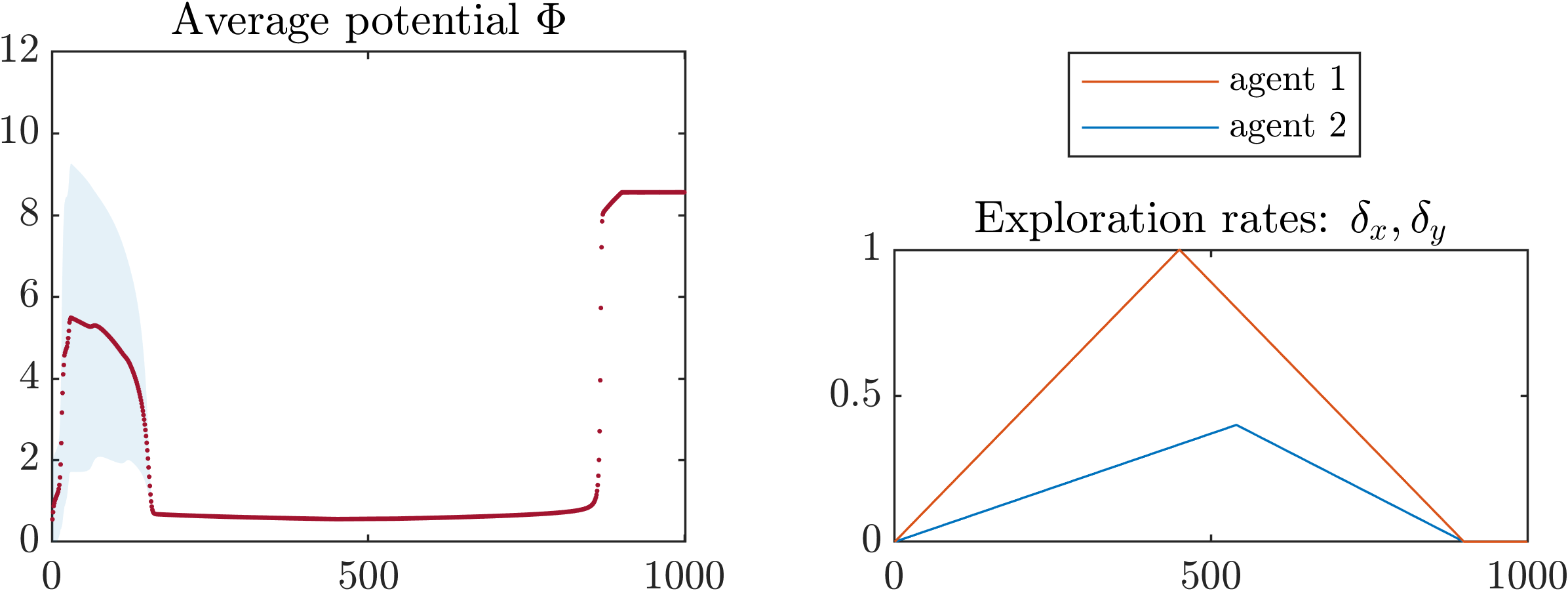}
\caption{The SQL dynamics for the same experiment as in Figure~\ref{fig:diag_10} but with less intense exploration rate in the CRL-1 policies for both agents. In this case, the SQL dynamics converge to a suboptimal outcome after the exploration phase (the mean lies at 8.559 instead of the absolute maximum 10).}
\label{fig:diag_9}
\end{figure}

Figure~\ref{fig:potential_9} shows a second experiment with $\Phi$, in which the SQL dynamics converge to a suboptimal outcome.
\begin{figure}[!htb]
\centering
\includegraphics[width=0.48\linewidth]{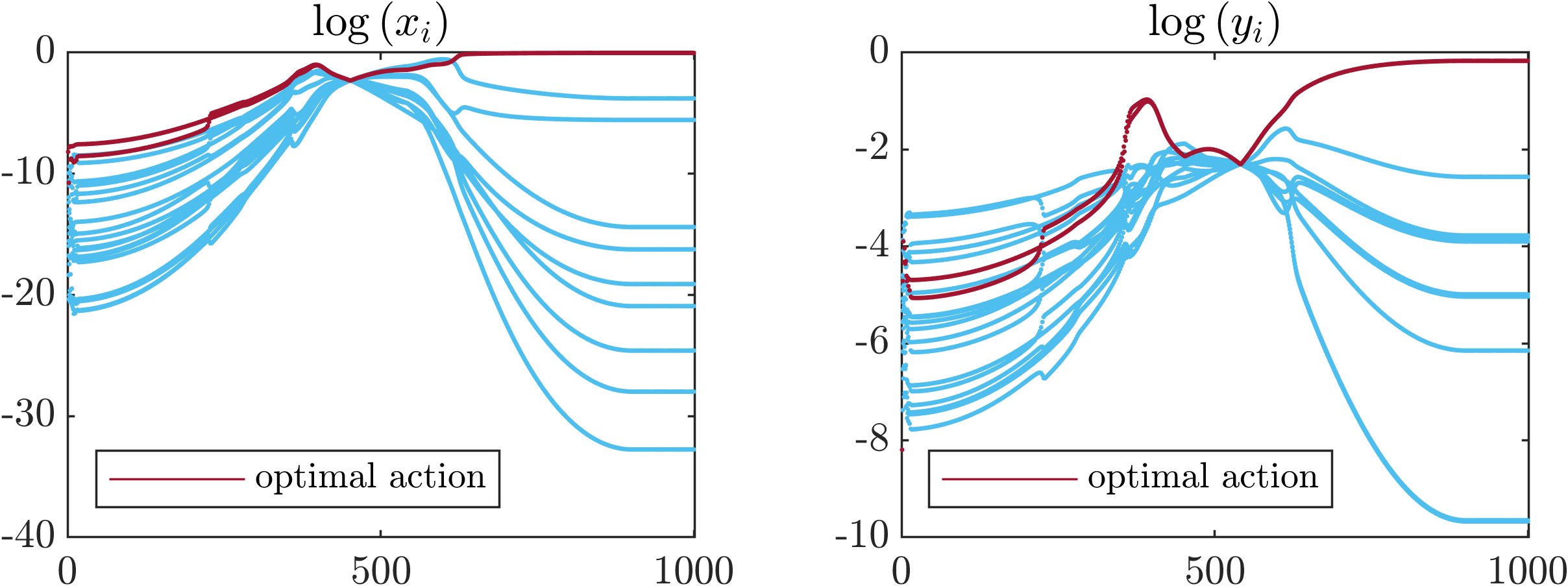}\hspace{10pt}
\includegraphics[width=0.48\linewidth]{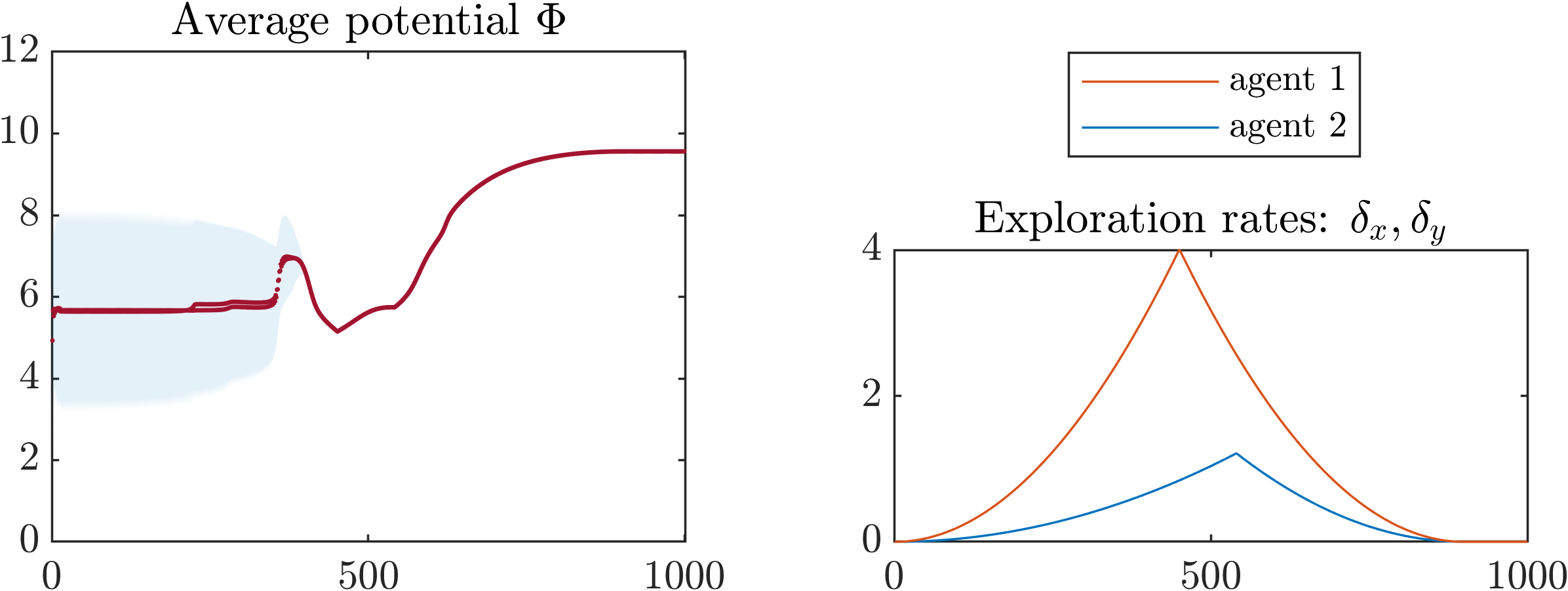}
\caption{The game and panels are as in Figure~\ref{fig:potential_10}. Agent 2 reduces their exploration rate in comparison to Figure~\ref{fig:potential_10} (lower right panel --- note also the quadratic increase-decrease in the exploration rate). The SQL dynamics converge to a suboptimal state with potential value $9.561$ (flat line in the right part of the bottom left panel) for all trajectories (even for those initially close to the global optimum) as can be inferred by the vanishing shaded region (one standard deviation).}
\label{fig:potential_9}
\end{figure}
The conclusive finding from the experiments is that after the exploration phase, the SQL dynamics converge to the same (local) optimum regardless of the starting point (the shaded region which represents one standard deviation around the mean vanishes). Finally, concerning the plots of the SQL dynamics in Figure~\ref{fig:potential_10}, the potential has been generated by the command \texttt{rnd(`1',twister)} of Matlab --- and the entry in $\(10,10\)$ deterministically set to $10$ --- and is given by 

\[\Phi=\left(\begin{array}{cccccccccc} 4 & 4 & 8 & 1 & 9 & 1 & 1 & 9 & 8 & 2\\ 7 & 7 & 9 & 4 & 7 & 7 & 4 & 2 & 6 & 9\\ 1 & 2 & 3 & 9 & 3 & 2 & 7 & 2 & 7 & 5\\ 3 & 8 & 7 & 5 & 8 & 3 & 4 & 8 & 4 & 6\\ 2 & 1 & 8 & 7 & 1 & 5 & 1 & 4 & 3 & 4\\ 1 & 7 & 9 & 3 & 5 & 1 & 5 & 2 & 9 & 3\\ 2 & 4 & 1 & 7 & 9 & 6 & 6 & 9 & 4 & 9\\ 4 & 6 & 1 & 8 & 3 & 2 & 5 & 4 & 9 & 6\\ 4 & 2 & 2 & 1 & 3 & 6 & 9 & 7 & 6 & 1\\ 5 & 2 & 8 & 7 & 2 & 7 & 6 & 7 & 6 & 10 \end{array}\right).\]

\renewcommand{\algorithmicrequire}{\textbf{Input:}}
\paragraph{Visualizing the Potential}
\label{par:visualization}
To study exploration in games with action spaces of arbitrary dimension (the findings are qualitatively equivalent when generalizing to arbitrary numbers of players), we adapt the method of \cite{Li18}. Specifically, for a two player game with choice distribution spaces $X\in \mathbb R^n$ and $Y\in \mathbb R^m$ consider the transformation $g:\mathbb R^n\to \mathbb R^{n-1}$ with $y_{1i}:=\ln{\(x_{1i}/x_{1n}\)}$, for each $i=1,2,\dots,n$ and the same for player 2. Together with the constraint $\sum_{i=1}^nx_{1i}=1$, the inverse $g^{-1}:\mathbb R^{n-1}\to \mathbb R^n$ of this transformation is given by $x_{1i}=x_{1n}e^{y_{1i}}$, for $i=1,\dots,n$.
\begin{algorithm}[!bth]
\caption{3D Visualization of the Potential}\label{alg:visual}
\begin{algorithmic}[1]
\Procedure{Input Game}{$\Phi,n,m$}
\Require $\Phi \gets $ Potential matrix
\Require $n \gets \#$ actions of player 1
\Require $m \gets \#$ actions of player 2\vspace*{0.2cm}
\EndProcedure
\Procedure{Generate random directions}{$n,m$}
\For {$i \gets u,v$}
\State{$u,v$ generate random vectors (not parallel)}
\State{$u=[u_1, u_2], v=[v_1, v_2]$ with}
\State{$u_1,v_1 \in \mathbb R^{n-1}, u_2,v_2\in\mathbb R^{m-1}$}\vspace*{0.2cm}
\EndFor
\EndProcedure
\Procedure{Transform Variables}{$u,v, \alpha,\beta$}
\hspace{0.4cm} \Require $\alpha,\beta \gets $ scalars in $\mathbb R$
\State {$x \gets \alpha\cdot u_1+\beta\cdot v_1$}
\State {$x \gets \exp{\(x\)}$}
\State {$x \gets$ normalized to sum up to $1$}
\State {Repeat to get $y$}\vspace*{0.2cm}
\EndProcedure
\Procedure{Evaluate Potential}{$x,y,\delta,\alpha,\beta$}
\hspace{0.4cm} \Require $\delta \gets $ Common exploration rate
\For {$\alpha,\beta$}
\State {$\Phi^H\(\alpha,\beta\) = x'\Phi y-\delta\sum x_i\ln{x_i}-\delta \sum y_j\ln y_j$}\vspace*{0.2cm}
\EndFor\\
\Return Plot tuples $\(\alpha,\beta,\Phi^H\(\alpha,\beta\)\)$
\EndProcedure
\end{algorithmic}
\end{algorithm}
The benefit of working with the transformed variables is that they are not subject to the Simplex constraints. Hence, we may choose random choice distributions $x_1,x_2$, transform them to $y_1,y_2$ and then scale them with two real scalars $\alpha,\beta$ of arbitrary sign and magnitude to obtain a point on the hyperplane $\alpha\cdot y_1+\beta\cdot y_2$. Then, we calculate the corresponding choice distributions (via the inverse transformation and the normalization equation) and evaluate the modified potential $\Phi^H$ at this point. This yields a tuple \[\(\alpha,\beta,\Phi^H\(g^{-1}\(\alpha\cdot y_1+\beta\cdot y_2\)\)\),\]
for each value of $\alpha,\beta$ which (if combined) yield the surface plots of Figures~\ref{fig:large} and \ref{fig:grid_pot}. The process is summarized in Algorithm~\ref{alg:visual}. Note that the technique readily generalizes to $n>2$ players with the same exploration rate. A specific instant of a potential game with $n=20$ which shows the transformation of the potential manifold as $\delta$ increases is included in the Multimedia Appednix.

\begin{figure*}[!htb]
\centering
\includegraphics[width=0.312\textwidth]{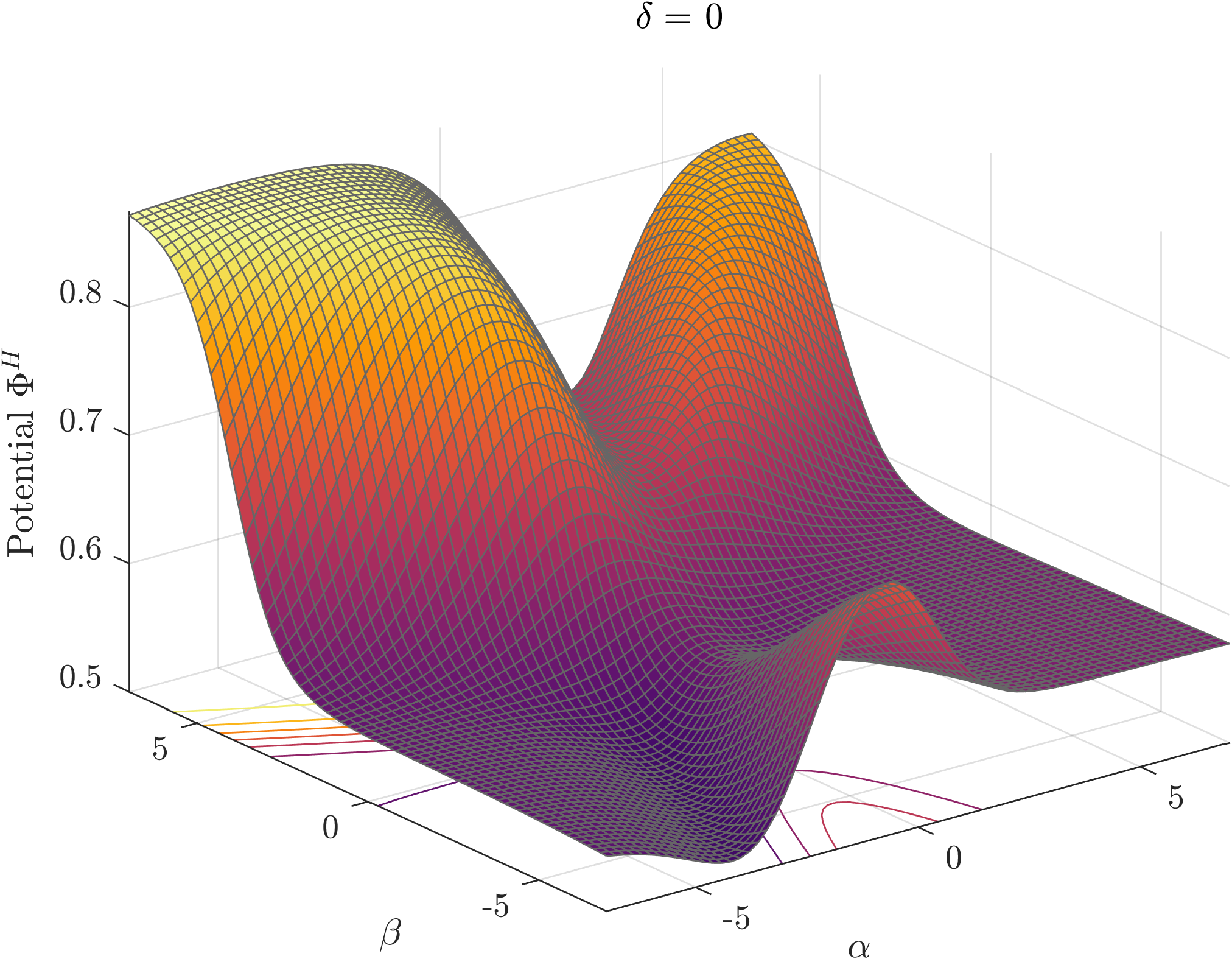}\hspace{10pt}
\includegraphics[width=0.312\textwidth]{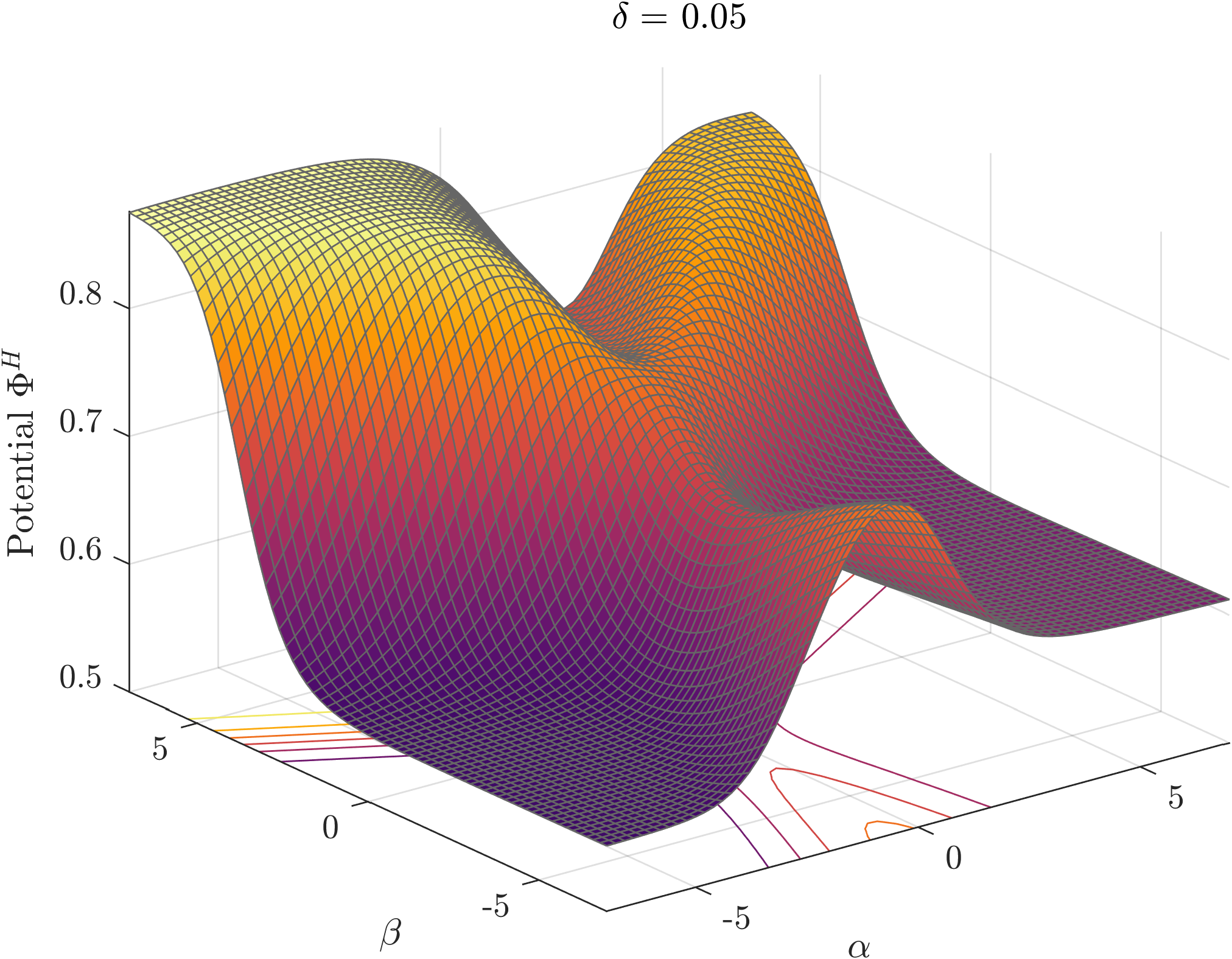}\hspace{10pt}
\includegraphics[width=0.312\textwidth]{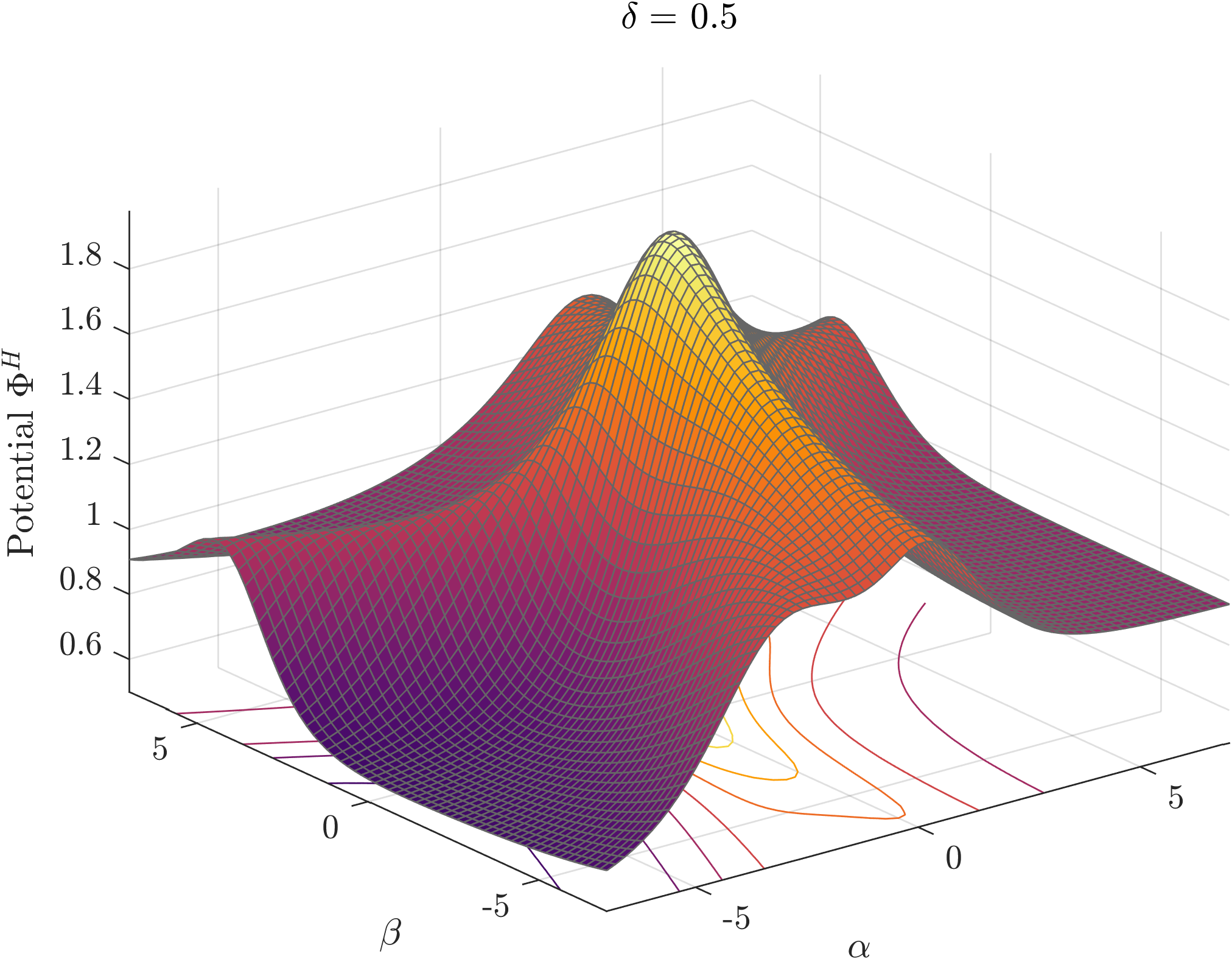}\\[0.2cm]
\includegraphics[width=0.312\textwidth]{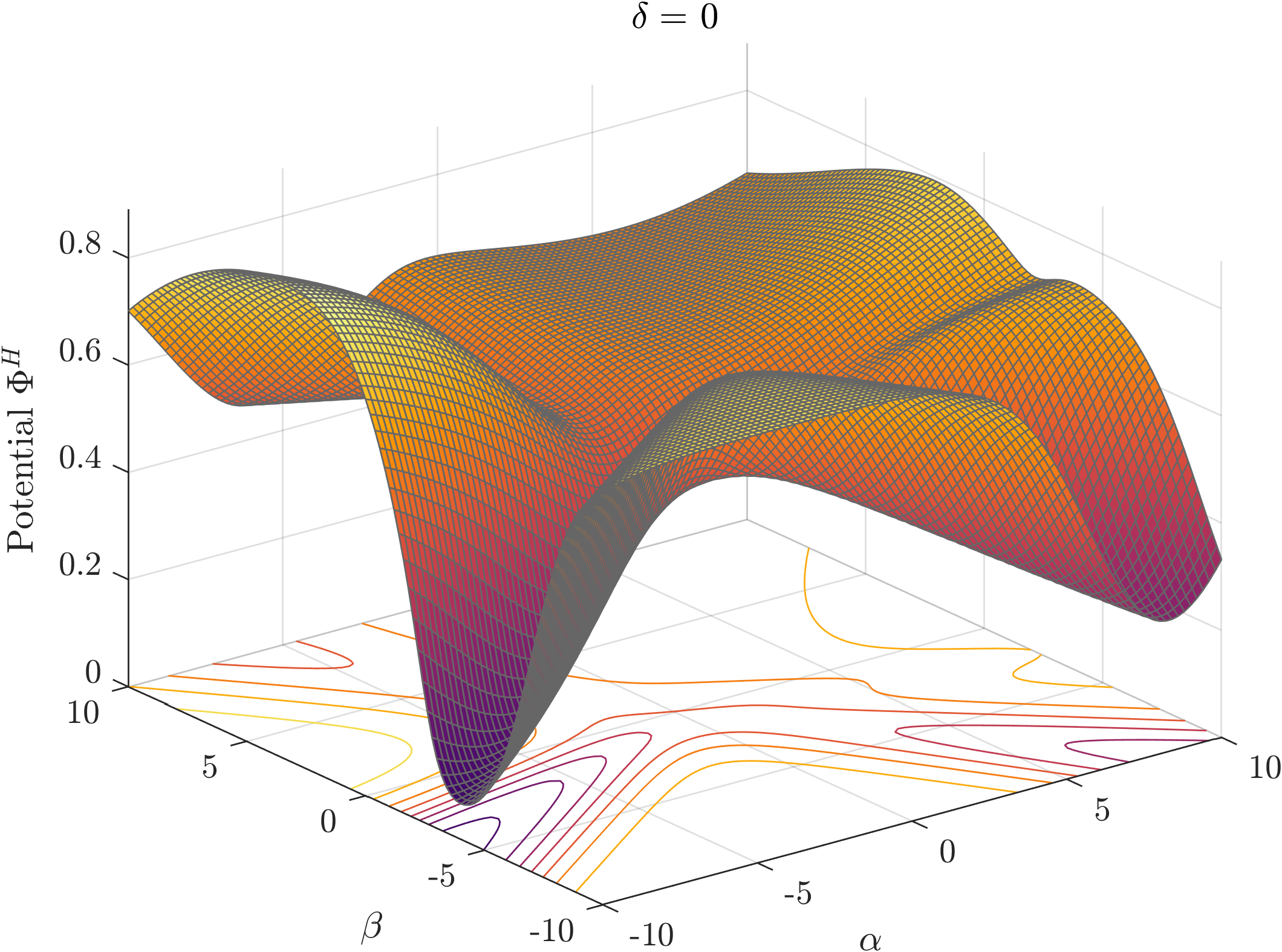}\hspace{10pt}
\includegraphics[width=0.312\textwidth]{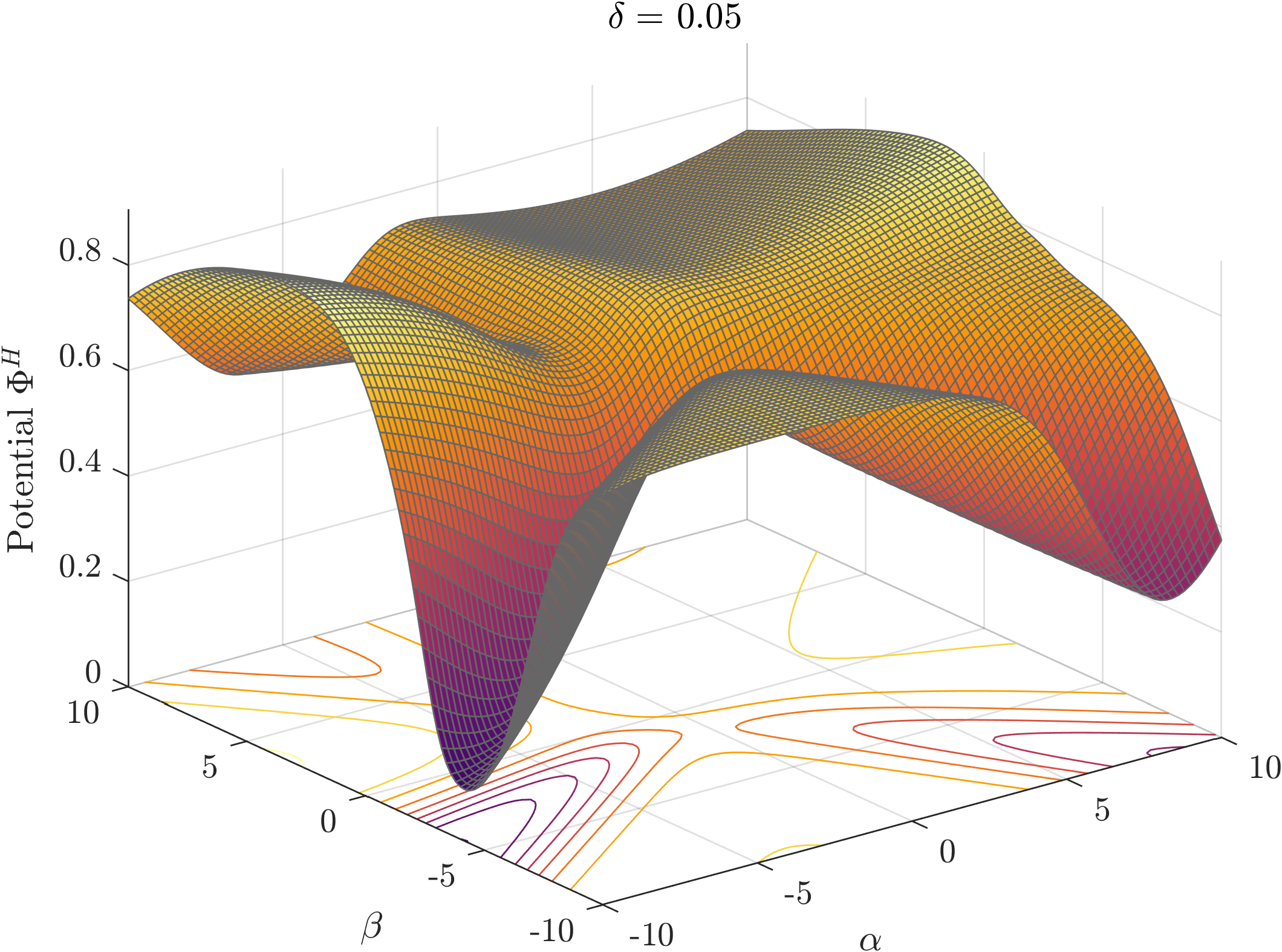}\hspace{10pt}
\includegraphics[width=0.312\textwidth]{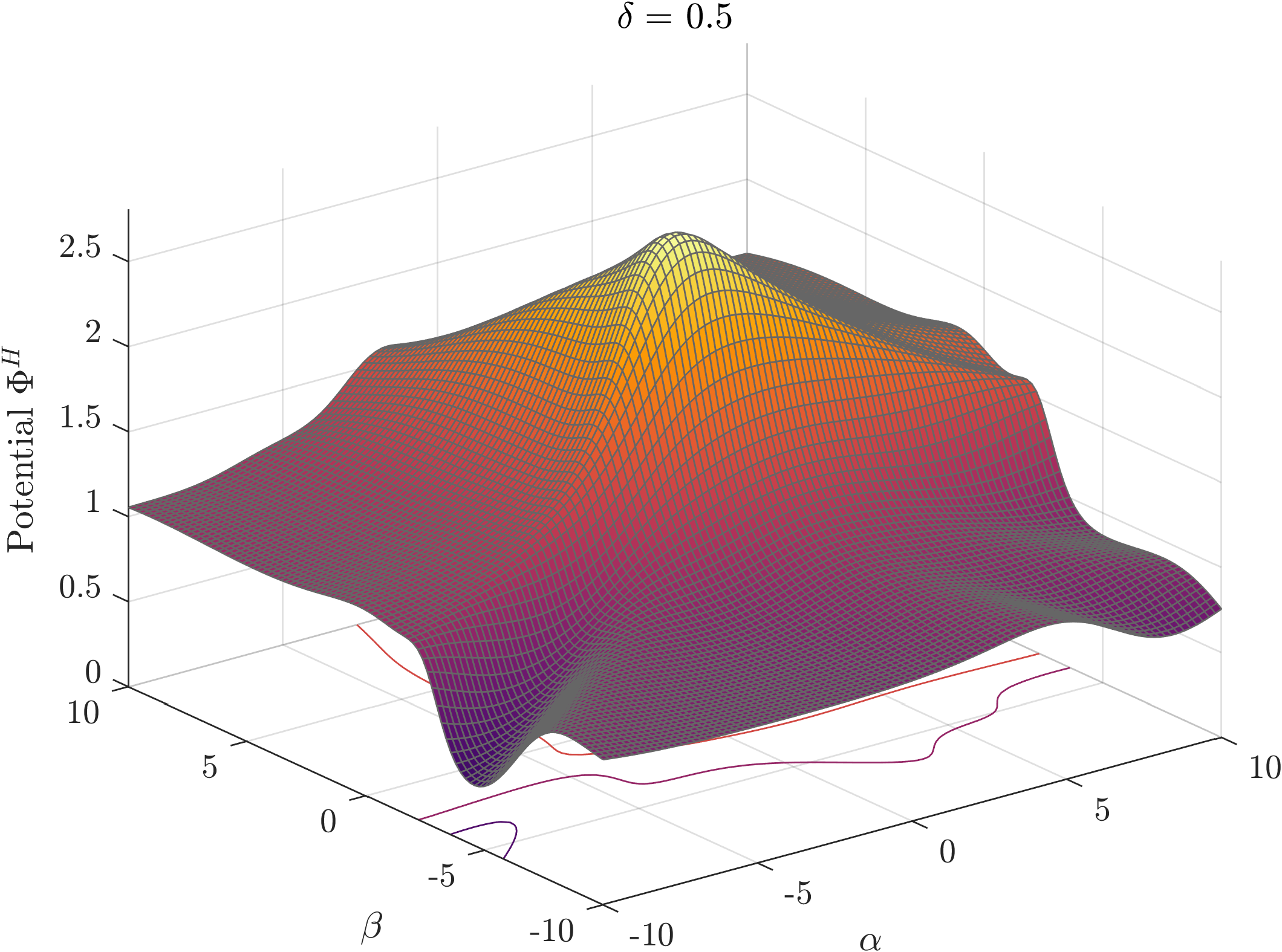}\\[0.2cm]
\includegraphics[width=0.312\textwidth]{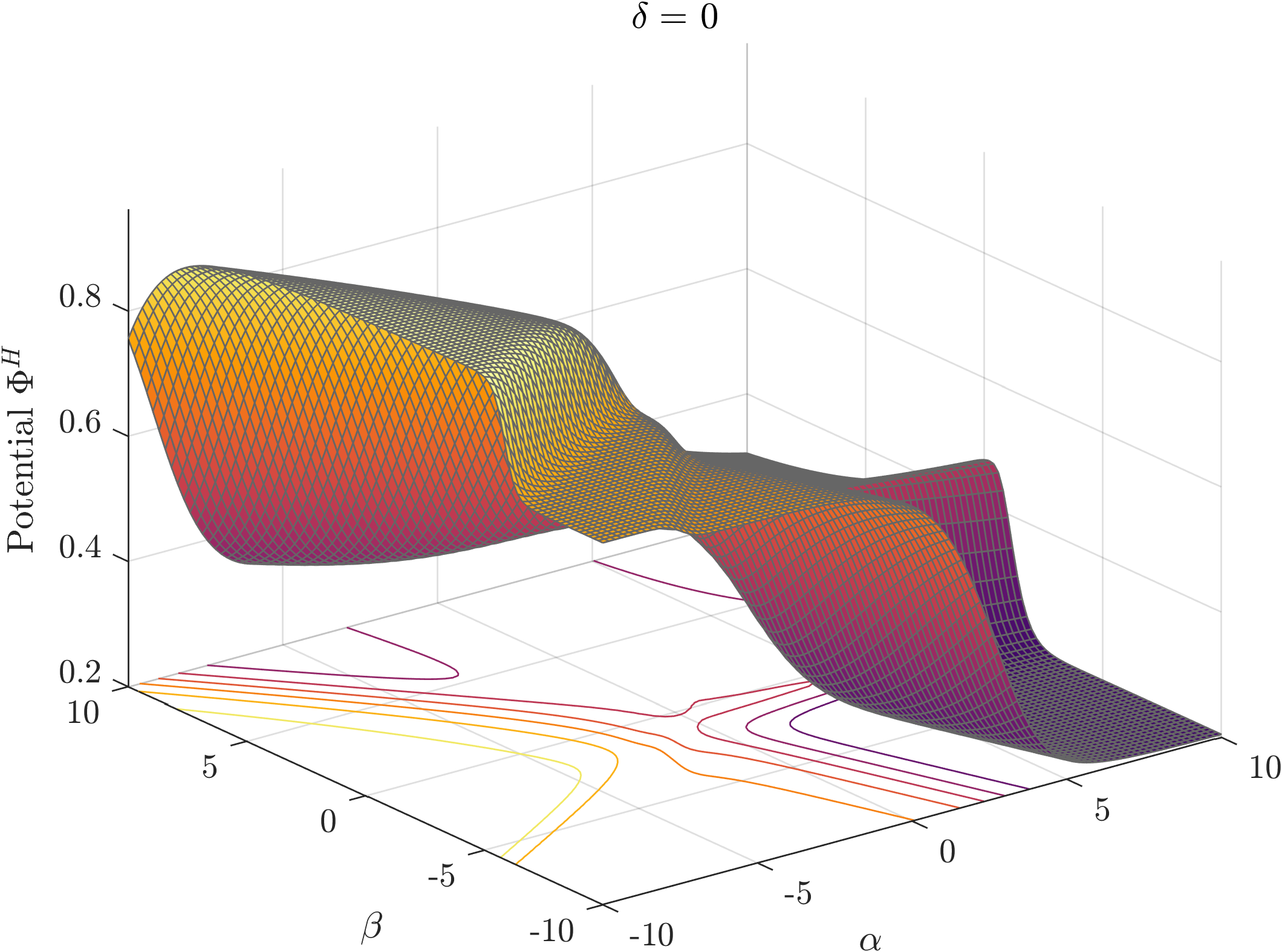}\hspace{10pt}
\includegraphics[width=0.312\textwidth]{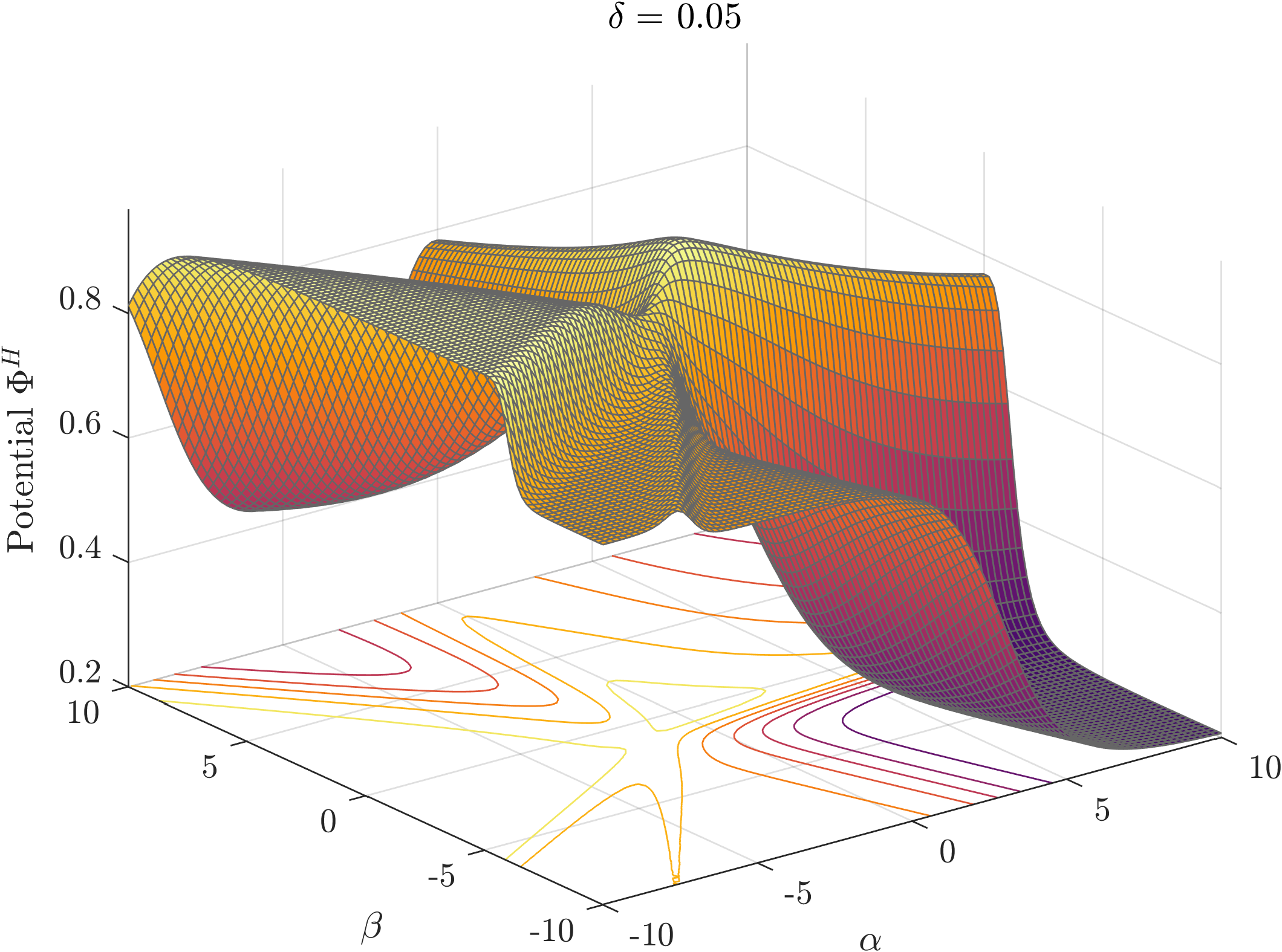}\hspace{10pt}
\includegraphics[width=0.312\textwidth]{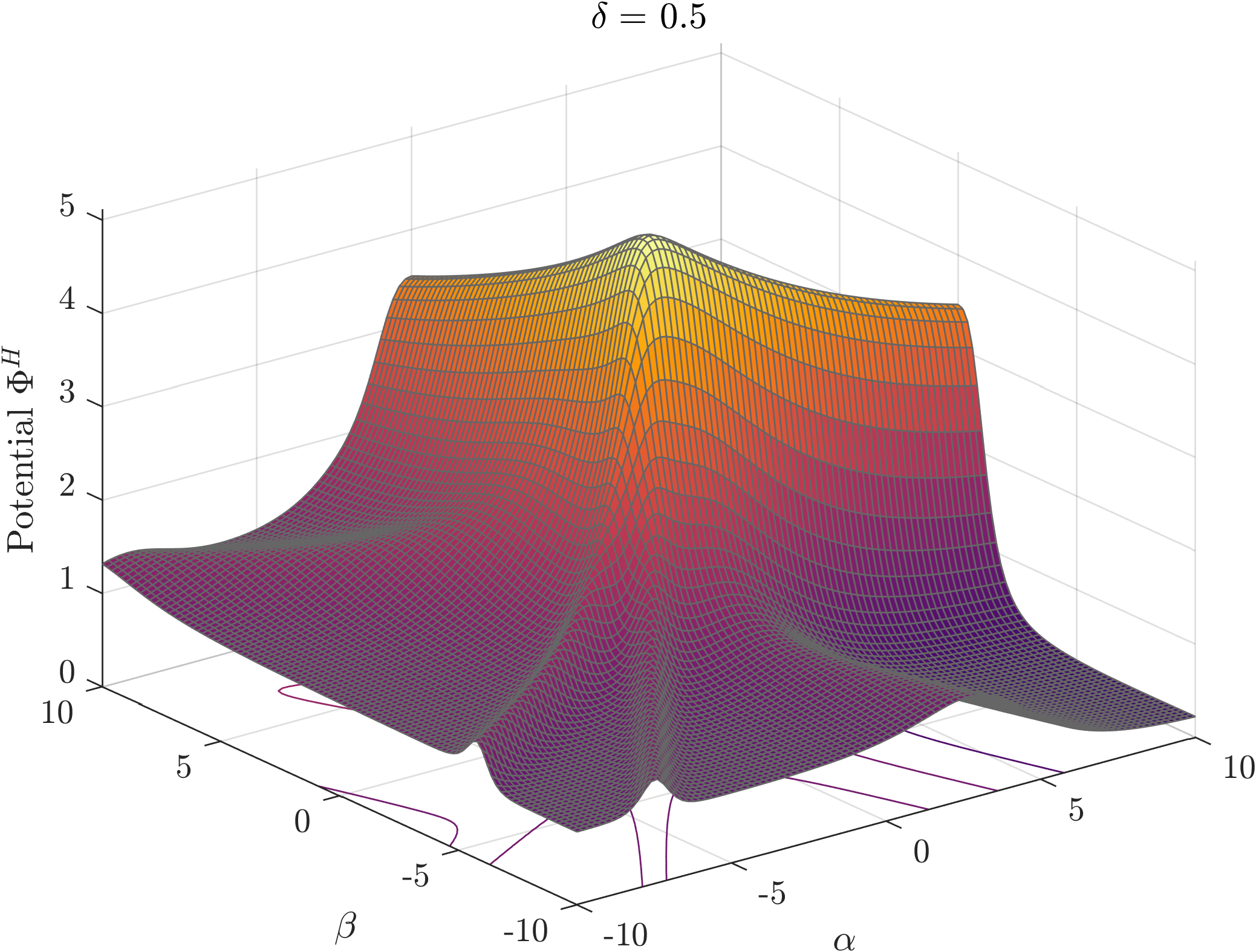}\\[0.2cm]
\includegraphics[width=0.312\textwidth]{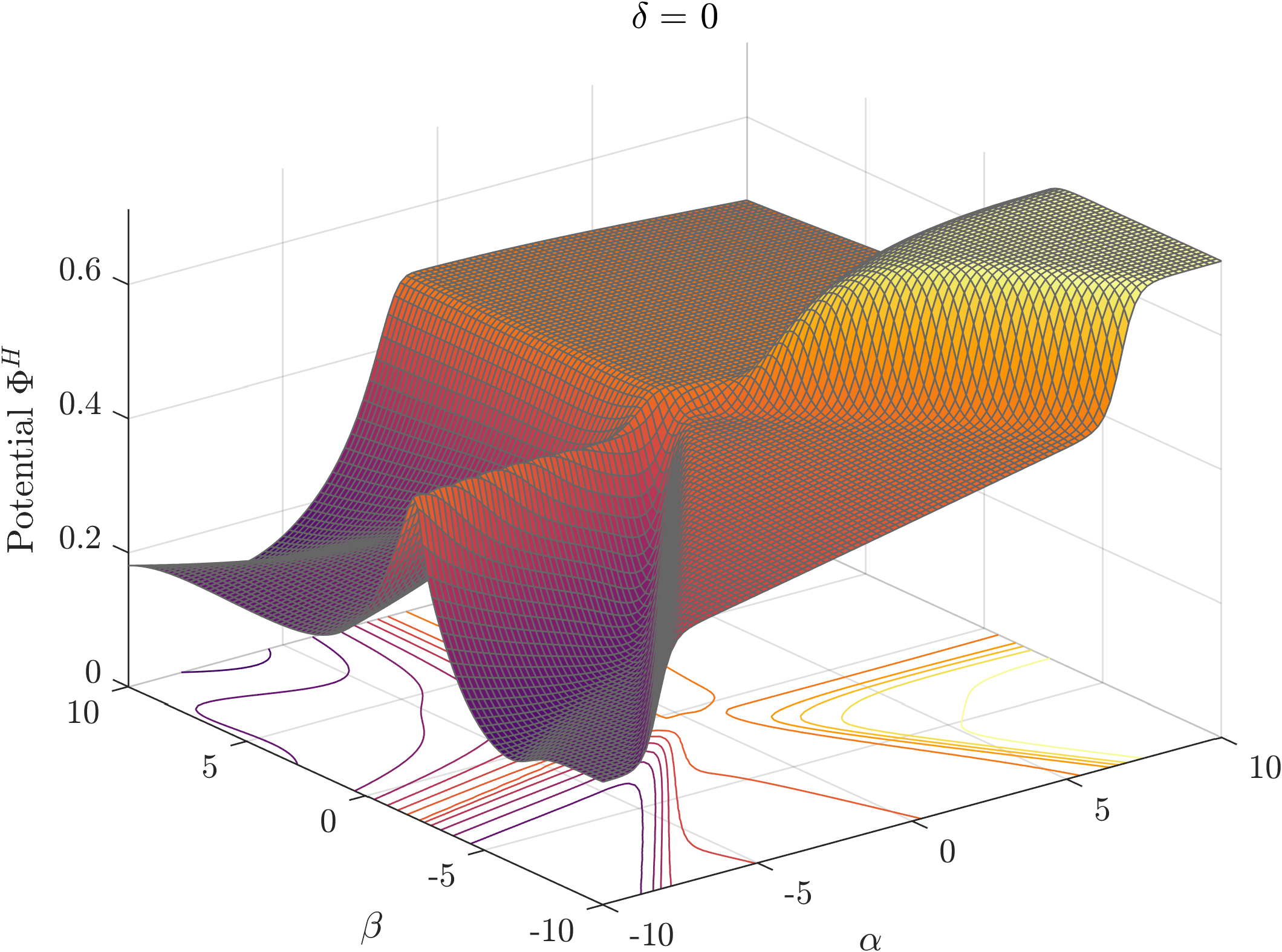}\hspace{10pt}
\includegraphics[width=0.312\textwidth]{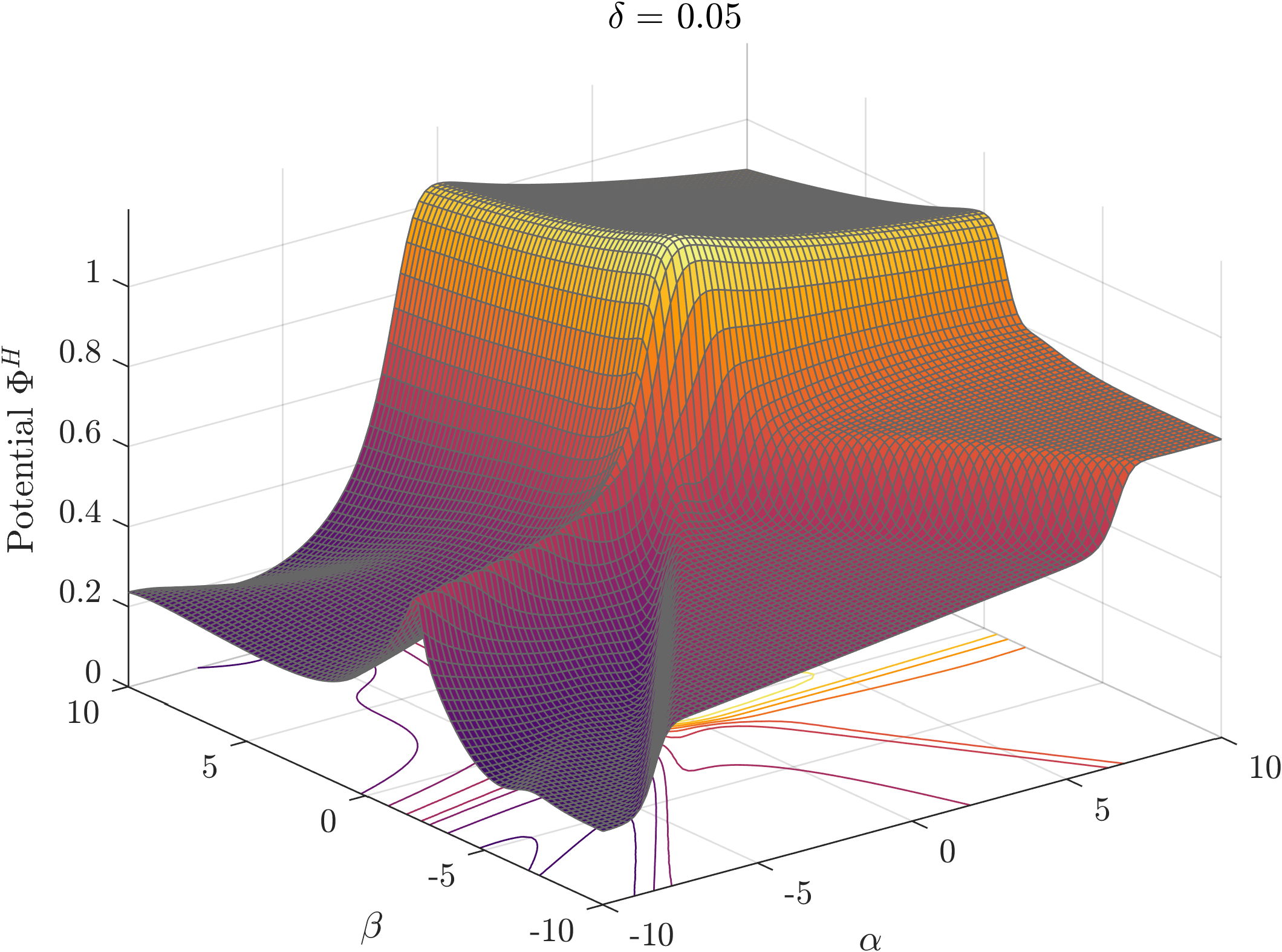}\hspace{10pt}
\includegraphics[width=0.312\textwidth]{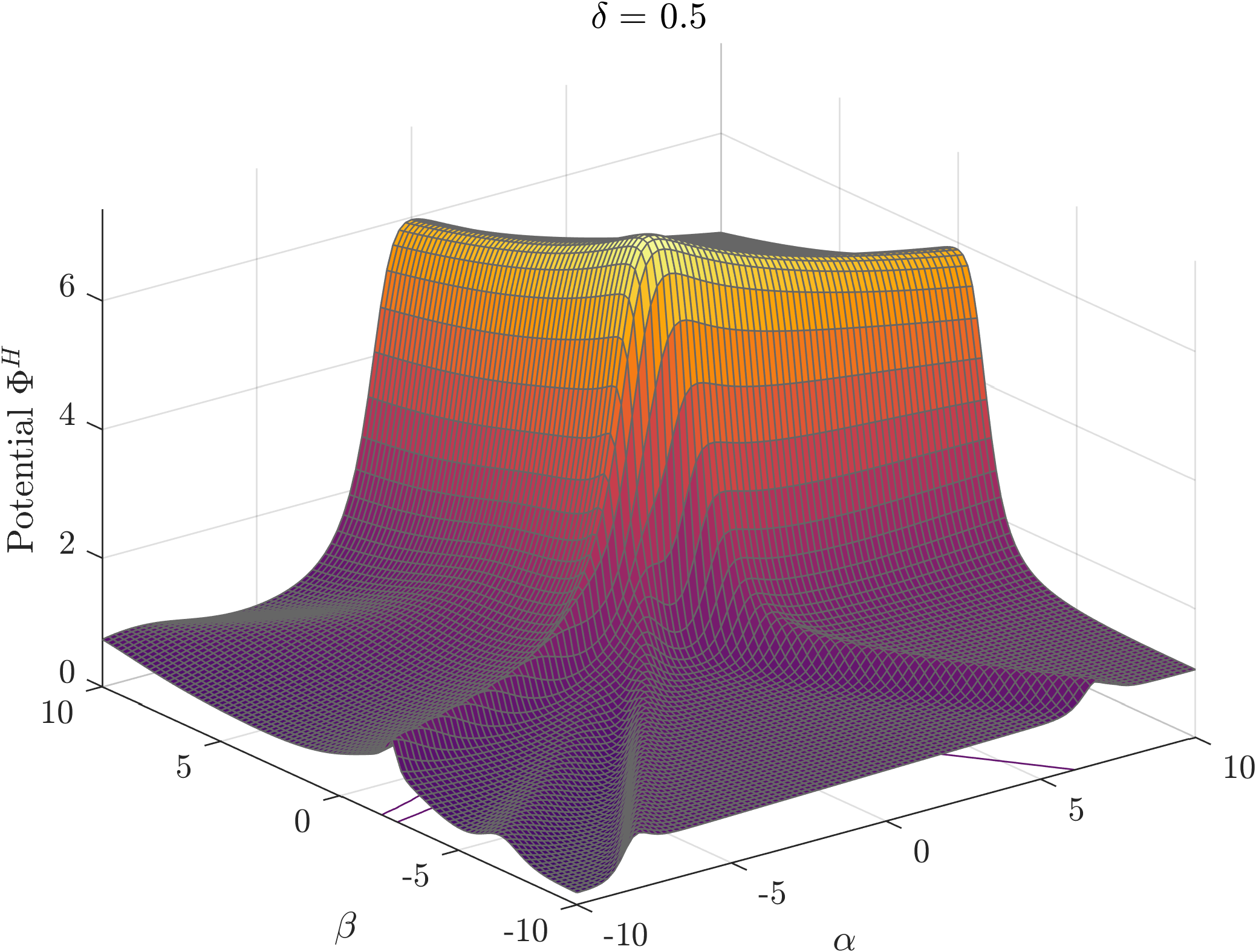}\\[0.2cm]
\caption{From top to bottom: snapshots of the modified potential $\Phi^H$ surface in $2$ player potential games with $n=4,10,100$ and $1000$ actions, respectively, and random payoffs in $[0,1]$. The surfaces are plotted using Algorithm~\ref{alg:visual} (see \cite{Li18}). From left to right: the exploration rate $\delta$ increases from $\delta=0$ to $\delta=0.05$ and $\delta=0.5$. The surface has arbitary maxima (resting points of the SQL dynamics) when $\delta$ is small (exploitation) but a single maximum at (or close to) $(0,0)$ which corresponds to the uniform distribution when $\delta$ is large (exploration). Intuitively, this is what agents \emph{see} as they increasingly incorporate exploration in their utilities. 
}
\label{fig:grid_pot}
\end{figure*}

\end{document}